\documentclass[10pt]{article}

\usepackage{amsmath}
\usepackage{bm} 
\usepackage{sansmath}
\usepackage{amsfonts}
\usepackage{amsthm}
\usepackage{amssymb}
\usepackage{fullpage}
\usepackage[ruled, linesnumbered]{algorithm2e}

\usepackage{tikz}
\usepackage[colorinlistoftodos]{todonotes}

\usepackage{tikz,tikz-qtree,tikz-qtree-compat,tikz-3dplot}
\usetikzlibrary{shapes,decorations.markings,arrows.meta,fit}

\newcommand\mybigxor[1]{
    \begin{tikzpicture}[baseline=(X.base), inner sep=0, outer sep=0]
        \node (X) {$\displaystyle{\bigvee_{#1}}$};
        \node at ($(X)+(0,.2)$)  {$\scriptstyle{+}$};
    \end{tikzpicture}
}

\newcommand\myxor{
    \begin{tikzpicture}[baseline=(X.base), inner sep=0, outer sep=0]
        \node (X) {$\bigvee$};
        \node at ($(X)+(0,.07)$)  {$\scriptscriptstyle{+}$};
    \end{tikzpicture}
}

\newcommand{\ed}{\text{\sf ED}}
\newcommand{\nop}[1]{}

\newcommand{\mycode}[1]{{\color{blue} \begin{verbatim} #1 \end{verbatim}}}
\newcommand{\suchthat}{\ | \ }

\newcommand{\gv}[1]{\ensuremath{\mbox{\boldmath$ #1 $}}}
\newcommand{\grad}[1]{\gv{\nabla} #1} 

\DeclareMathOperator*{\argmin}{arg\,min}
\DeclareMathOperator*{\argmax}{arg\,max}

\newcommand{\panda}{\text{\sf PANDA}}
\newcommand{\spanda}{\text{\sf \#PANDA}}

\newcommand{\faqw}{\text{\sf faqw}}
\newcommand{\subfaqw}{\text{\sf smfw}}
\newcommand{\ssubfaqw}{\text{\sf \#smfw}}
\newcommand{\ssubw}{\text{\sf \#subw}}

\newcommand{\Mod}{\text{\sf M}}

\newcommand{\InsideOut}{\text{\sf InsideOut}}

\newcommand{\K}{\mathbb{K}}    

\newcommand{\fhtw}{\text{\sf fhtw}}
\newcommand{\subw}{\text{\sf subw}}

\newcommand{\ksum}{\text{\sf $k$-sum}}

\newcommand{\opt}{\text{\sf OPT}}
\newcommand{\R}{{\mathbb R}}

\newcommand{\td}{\text{\sf TD}}

\newcommand{\dom}{\text{\sf Dom}}
\newcommand{\faq}{\text{\sf FAQ}}
\newcommand{\faqai}{\text{\sf FAQ-AI}}

\newcommand{\true}{\text{\sf true}}
\newcommand{\false}{\text{\sf false}}

\newcommand{\norm}[1]{\left\|#1\right\|}
\newcommand{\image}{\text{\sf image}}

\newcommand{\inner}[1]{\left\langle #1 \right\rangle}

\newcommand{\calH}{\mathcal H}

\newcommand{\calB}{\mathcal B}

\newcommand{\calL}{\mathcal L}

\newcommand{\calM}{\mathcal M}
\newcommand{\calV}{\mathcal V}
\newcommand{\calE}{\mathcal E}

\newcommand{\calW}{\mathcal W}
\newcommand{\bB}{{\mathbf B}}

\theoremstyle{plain}                   
\newtheorem{theorem}{Theorem}[section]
\newtheorem{lemma}[theorem]{Lemma}
\newtheorem{proposition}[theorem]{Proposition}
\newtheorem{cor}[theorem]{Corollary}

\theoremstyle{definition}              

\newtheorem{opm}{Question}
\newtheorem{conj}[theorem]{Conjecture}
\newtheorem{example}[theorem]{Example}

\newtheorem{definition}[theorem]{Definition}

\newtheorem{rmk}[theorem]{Remark}
\newtheorem{claim}{Claim}

\newcommand{\bi}{\begin{itemize}}
\newcommand{\ei}{\end{itemize}}
\newcommand{\be}{\begin{enumerate}}
\newcommand{\ee}{\end{enumerate}}

\newcommand{\bp}{\begin{proof}}
\newcommand{\ep}{\end{proof}}
\newcommand{\bcor}{\begin{cor}}
\newcommand{\ecor}{\end{cor}}
\newcommand{\bthm}{\begin{theorem}}
\newcommand{\ethm}{\end{theorem}}
\newcommand{\blmm}{\begin{lemma}}
\newcommand{\elmm}{\end{lemma}}
\newcommand{\bdefn}{\begin{definition}}
\newcommand{\edefn}{\end{definition}}
\newcommand{\bprop}{\begin{proposition}}
\newcommand{\eprop}{\end{proposition}}
\newcommand{\bconj}{\begin{conj}}
\newcommand{\econj}{\end{conj}}
\newcommand{\bopm}{\begin{opm}}
\newcommand{\eopm}{\end{opm}}
\newcommand{\brmk}{\begin{rmk}}
\newcommand{\ermk}{\end{rmk}}

\allowdisplaybreaks[1]
\renewcommand{\vec}[1]{\ensuremath\boldsymbol{#1}}

\begin{document}

\title{Functional Aggregate Queries with Additive Inequalities}

\author{
    Mahmoud Abo Khamis \\ {\small relational\underline{AI}} \and
    Ryan R. Curtin \\ {\small relational\underline{AI}} \and
    Benjamin Moseley \\ {\small Carnegie Mellon University} \and
    Hung Q. Ngo \\ {\small relational\underline{AI}} \and
    XuanLong Nguyen \\ {\small University of Michigan} \and
    Dan Olteanu \\ {\small University of Zurich} \and
    Maximilian Schleich \\ {\small University of Washington}
}

\date{}

\maketitle

\begin{abstract}
    Motivated by fundamental applications in databases and relational
    machine learning, we formulate and study the problem of answering
    {\em functional aggregate queries} ($\faq$) in which some of the input factors
    are defined by a collection of additive inequalities between variables.
    We refer to these queries as $\faqai$ for short.

    To answer $\faqai$ in the Boolean semiring, we define {\em relaxed} tree
    decompositions and {\em relaxed} submodular and fractional hypertree width
    parameters. We show that an extension of the $\InsideOut$ algorithm
    using Chazelle's geometric data structure for solving the semigroup range search
    problem can answer Boolean $\faqai$ in time given by these new width parameters.
    This new algorithm achieves lower complexity than known solutions for $\faqai$.
    It also recovers some known results in database query answering.

    Our second contribution is a relaxation of the set of
    polymatroids that gives rise to the {\em counting} version of the submodular width, denoted
    by $\ssubw$. This new width is sandwiched between the submodular and the fractional
    hypertree widths.
    Any $\faq$ and $\faqai$ over one semiring can be answered in time
    proportional to $\ssubw$ and respectively to the relaxed version of $\ssubw$.

    We present three applications of our $\faqai$ framework to relational machine learning:
    $k$-means clustering, training linear support vector machines, and training models
    using non-polynomial loss. These optimization problems can be solved over
    a database asymptotically faster than computing the join of the database relations.
\end{abstract}


\section{Introduction}

In this article we consider the problem of computing functional aggregate queries with additive inequalities, or $\faqai$ queries for short.
Although existing algorithms such as $\InsideOut$~\cite{DBLP:conf/pods/KhamisNR16,faq-arxiv} and $\panda$~\cite{panda-pods,panda-arxiv} are able to evaluate $\faqai$ queries, they do not exploit the structure of the additive inequalities. We introduce variants of these algorithms to this effect. Whereas the prior algorithms work on hypertree decompositions of the queries, our new algorithms work on relaxations of these decompositions to achieve lower computational complexities than $\InsideOut$ and $\panda$.

Functional aggregate queries with additive inequalities can express computation needed for various database workloads and supervised and unsupervised machine learning.

On the database side, queries with inequalities occur naturally in scenarios involving temporal and spatial relationships between objects in databases. In a retail scenario (e.g., TPC-H), we would like to compute the revenue generated by a customer's orders whose dates closely precede the ship dates of their lineitems. In streaming scenarios, we would like to detect patterns of events whose time stamps follow a particular order~\cite{DBLP:series/synthesis/2010Golab}. In spatial data management scenarios, we would like to retrieve objects whose coordinates are within a multi-dimensional range or in close proximity of other objects~\cite{Mamoulis:2011:SDM:2208106}.
The evaluation of Core XPath queries over XML documents amounts to the evaluation of conjunctive queries with inequalities expressing tree relationships in the pre/post plane~\cite{Grust:2002}.

For machine learning, we show that $\faqai$ can express computation needed for $k$-means clustering, training linear support vector machines, and training models using non-polynomial loss. These optimization problems can be solved over
    a database asymptotically faster than computing the join of the database relations.

\subsection{Motivating examples}
\label{subsec:motivations}

A key insight of this article is that the efficient computation of inequality joins can  reduce the computational complexity of supervised and unsupervised machine learning.

\begin{example}
\label{ex:intro-query2}
The $k$-means algorithm divides the input dataset $\bm G$ into $k$ clusters of similar data points~\cite{DBLP:journals/prl/Jain10}. Each cluster $\bm G_i$ has a mean $\bm \mu_i \in \mathbb{R}^n$, which is chosen according to the following optimization (similarity is defined here with respect to the $\ell_2$ norm):

\begin{align}
  \min_{(\bm G_1,\dots,\bm G_k)} \sum_{i=1}^{k}
  \sum_{\bm x \in \bm G_i} \norm{\bm x - \bm\mu_i}_2^2.
  \label{eq:kmeans-optimization}
\end{align}

Let $\mu_{i,\ell}$ be the $\ell$'th component of mean vector $\bm\mu_i$. For
a data point $\bm x \in \bm G$, the function $c_{ij}$ computes the difference
between the squares of the $\ell_2$-distances from $\bm x$ to $\bm\mu_i$ and from $\bm x$
to $\bm\mu_j$:
\begin{eqnarray*}
 c_{ij}(\bm x) &=&
 \norm{\bm x - \bm\mu_i}_2^2-\norm{\bm x - \bm\mu_j}_2^2\\
 &=&\sum_{\ell \in [n]} [\mu_{i,\ell}^2 - 2  x_\ell (\mu_{i,\ell} - \mu_{j,\ell})
    - \mu_{j,\ell}^2].
\end{eqnarray*}
A data point $\bm x \in \bm G$ is
closest to mean $\bm\mu_i$ from the set of $k$ means iff $\forall j \in [k] : c_{ij}(\bm x) \leq 0$.

To compute the mean vector $\bm\mu_{i}$, we need to compute the sum of values for each dimension $\ell \in [n]$ over $\bm G_i: \sum_{\bm x \in \bm
    G_i} x_\ell$. If the dataset $\bm
G$ is the join of database relations
$(R_p)_{p\in[m]}$ over schemas $S_p\subseteq [n], \forall
p\in[m]$, we can formulate this sum computation as a datalog-like query with aggregates~\cite{logiql}:
\begin{align*}
    Q^{(i,\ell)}_1\left(\sum x_\ell\right) \leftarrow \left(\bigwedge_{p\in[m]} R_p(\bm x_{S_p}) \right) \wedge \left(\bigwedge_{j \in [k]} c_{ij}(\bm x)  \leq 0\right).
\end{align*}
The above notation means that the answer to query $Q^{(i,\ell)}_1$ is the sum of $x_\ell$
over all tuples $\bm x$ satisfying the conjunction on the right-hand side.
Section~\ref{sec:applications} gives further queries necessary to compute the $k$-means.
As we show in this article, such queries with aggregates and inequalities can
be computed asymptotically faster than the join defining
$\bm G$. \qed
\end{example}

Simple queries can already highlight the limitations of state-of-the-art evaluation techniques, as shown next.

\begin{example}
\label{ex:intro-query3}
State-of-the-art techniques take time $O(N^2)$ to compute the following query over relations of size $\leq N$:

    \begin{align*}
        Q_2() \leftarrow R(a,b) \wedge S(b,c) \wedge T(c,d) \wedge a \leq d,
    \end{align*}

\noindent Examples~\ref{ex:4:path:ineq} and~\ref{ex:4:path:ineq:count} show how to compute $Q_2$ and its counting version in time $O(N^{1.5}\log N)$ using the techniques introduced in this article.\qed
\end{example}

\subsection{The $\faqai$ problem}
\label{sec:intro:faqai}

One way to answer the above queries is to view them as
{\em functional aggregate queries} ($\faq$)~\cite{DBLP:conf/pods/KhamisNR16}
formulated in sum-product form over some semiring.
We therefore briefly introduce $\faq$ over a single semiring.

We first establish notation.
For any positive integer $n$, let $\calV = [n]$.
For $i \in \calV$,
let $X_i$ denote a variable/attribute, and $x_i$ denote a value in
the discrete domain $\dom(X_i)$ of $X_i$.
For any $K\subseteq \calV$, define
$\bm X_K = (X_i)_{i\in K}$,
$\bm x_K = (x_i)_{i\in K} \in \prod_{i\in K}\dom(X_i)$.
That is, $\bm X_K$ is a tuple of variables and $\bm x_K$ is a
tuple of values for these variables.

Consider a semiring $(\bm D, \oplus, \otimes, \bm 0, \bm 1)$.
Let $\calH=(\calV=[n],\calE)$ be a {\bf multi}-hypergraph, which means that $\calV=[n]$
is a set of vertices and $\calE$ is a {\bf multiset}\footnote{A multiset is a
collection of elements each of  which can occur multiple times.} of edges
where each edge $K \in \calE$ is a subset of $\calV$.
To each edge $K \in \calE$ we associate a function $R_K : \prod_{v\in K}\dom(X_i)
\to \bm D$ called {\em factor}.
An $\faq$ query over one semiring with free variables $F\subseteq \calV$ has the form:
\begin{align}
    Q(\bm x_F) &=
    \bigoplus_{\bm x_{\calV \setminus F} \in \prod_{i \in \calV \setminus F} \dom(X_i)}\
    \bigotimes_{K\in\calE} R_K(\bm x_K).
    \label{eqn:faq}
\end{align}
Under the Boolean semiring $(\{\true,\false\}, \vee, \wedge, \false,\true)$, the que\-ry~\eqref{eqn:faq} becomes a conjunctive query:
The factors $R_K$ represent input relations,
where $R_K(\bm x_K) = \true$ iff $\bm x_K \in R_K$, with some notational overloading.
Under the sum-product semiring, the query~\eqref{eqn:faq} counts the number of tuples in the join result for each tuple $\bm x_F$, where the factors $R_k$ are indicator functions $R_K (\bm x_K) = \bm 1_{x_K \in R_K}$.
(The notation $\bm 1_A$ denotes the indicator function of the event $A$ in the semiring $(\bm D,
\oplus, \otimes, \bm 0, \bm 1)$: $\bm 1_A = \bm 1$ if $A$ holds, and $\bm 0$ otherwise.)
To aggregate over some input variable, say $X_k$, we can designate an identity factor $R_k(x_k) = x_k$.

Throughout the article, we assume the query size to be a constant and state runtimes in data complexity.
It is known~\cite{DBLP:conf/pods/KhamisNR16} that over an arbitrary semiring,
the query~\eqref{eqn:faq} can be answered in time $O(N^{\fhtw(Q)}\cdot\log N)$, where $N$ is the size of the largest relation $R_K$, $\fhtw$ denotes the {\em fractional
hypertree width} of the query, and $Q$ has no free
variables~\cite{DBLP:journals/talg/GroheM14}.  If $Q$ has free variables,
$\fhtw$-width becomes $\faq$-width instead~\cite{DBLP:conf/pods/KhamisNR16}.
Here $N$ is the size of the largest factor $R_K$.
Over the Boolean semiring, the time can be lowered to $\tilde O(N^{\subw(Q)})$~\cite{panda-pods}, where $\subw$ is the {\em submodular width}~\cite{MR3144912} and $\tilde O$ hides a polylogarithmic factor in $N$.

Motivated by the examples in Section~\ref{subsec:motivations}, we formulate a class
of $\faq$ queries called $\faqai$:

\bdefn[$\faqai$]
\label{defn:faqai}
Given a hyperedge multiset $\calE$ that is partitioned into two multisets
$\calE = \calE_s \cup \calE_\ell$,
where $s$ stands for ``skeleton'' and $\ell$ stands for ``ligament'',
the input to a query from the $\faqai$ class is the following:
\be
\item To each hyperedge $K \in \calE_s$, there corresponds a function $R_K : \prod_{i \in
K} \dom(X_i) \to \bm D$, as in the $\faq$ case.
\item To each hyperedge $S \in \calE_\ell$, there corresponds $|S|$ functions $\theta^S_{v} :
\dom(X_v) \to \R$, one for every variable $v \in S$.
\ee
The output to the $\faqai$ query is the following:
\begin{align}
    Q(\bm x_F) &=
    \bigoplus_{\bm x_{\calV \setminus F}}
    \left(\bigotimes_{K\in \calE_s} R_K(\bm x_K)\right) \otimes
    \left(\bigotimes_{S \in \calE_\ell} \bm 1_{\sum_{v \in S} \theta^S_v(x_v) \leq 0} \right).
    \label{eqn:our:query}
\end{align}
The summation $\bigoplus$ is over tuples
$\bm x_{\calV \setminus F} \in  \prod_{i \in \calV \setminus F} \dom(X_i)$.
The (uni-variate) functions $\theta^S_v(x_v)$ can be user-defined functions, e.g.,
$\theta_1^S(x_1) = x_1^2/2$, or binary predicates with one key in $\dom(X_v)$ and a numeric value, e.g., a table \verb|salary(employee_id, salary_value)| where \verb|employee_id| is a key.
The only requirement we impose is that, given $x$, the value $\theta^S_v(x)$
can be accessed/computed in $O(1)$-time (in data complexity).
\edefn
If $\calE_\ell = \emptyset$, then we get back the $\faq$
formulation~\eqref{eqn:faq}.

\begin{example}
The queries in Section~\ref{subsec:motivations} are instances of~\eqref{eqn:our:query}:
\begin{align}
    Q^{(i,\ell)}_1() &=
    \bigoplus_{\bm x_{[n]}} x_\ell \otimes
                \left(\bigotimes_{p\in [m]} R_p(\bm x_{S_p})\right)  \otimes \left(\bigotimes_{j \in [k]}
    {\bm 1}_{c_{ij}(\bm x) \leq 0} \right), \label{eq:kmeans-query}\\
    Q_2() &=
    \bigoplus_{\bm x_{[4]}}
    R(x_1,x_2)\otimes S(x_2,x_3)\otimes T(x_3,x_4) \otimes \bm 1_{x_1-x_4 \leq 0}. \nonumber
\end{align}
\nop{$Q_1$ is over the Boolean semiring: The result of an XPath query is a set of nodes;}
Note that for a given $\bm x$, $c_{ij}(\bm x)$ can be computed in
$O(1)$-time in data complexity, which in this context means when the number of dimensions $n$
is a constant.
$Q_1$ is over the sum-product semiring. $Q_2$ can be over any semiring: Example~\ref{ex:4:path:ineq} discusses the case of the Boolean semiring while Example~\ref{ex:4:path:ineq:count} discusses the sum-product semiring. \qed
\end{example}

\subsection{Our contributions}

To answer $\faq$ queries of the form~\eqref{eqn:faq}, currently there are two dominant width
parameters: fractional hypertree width ($\fhtw$~\cite{DBLP:journals/talg/GroheM14}) and
submodular width ($\subw$~\cite{MR3144912}).\footnote{Section~\ref{subsec:width:params} overviews other
notions of widths.} It is known that $\subw \leq \fhtw$ for any query, and in the Boolean
semiring we can answer~\eqref{eqn:faq} in $\tilde
O(N^{\subw})$-time~\cite{panda-pods,MR3144912}. For non-Boolean semirings, the best known
algorithm, called $\InsideOut$ \cite{DBLP:conf/pods/KhamisNR16,DBLP:journals/sigmod/Khamis0R17},
evaluates~\eqref{eqn:faq} in time $O(N^{\fhtw}\log N)$. For queries with free variables, $\fhtw$ is replaced by the more general notion of $\faq$-width
($\faqw$)~\cite{DBLP:conf/pods/KhamisNR16}; however, for brevity we discuss
the non-free variable case here.

Following~\cite{DBLP:journals/sigmod/Khamis0R17},
both width parameters $\subw$ and $\fhtw$ can be defined via two constraint sets: the first
is the set $\td$ of all tree decompositions of the query hypergraph $\calH$,
and the second is the set of polymatroids $\Gamma_n$ on $n$ vertices of $\calH$.
The widths $\subw$ and $\fhtw$ are then defined as maximin and respectively minimax optimization problems
on the domain pair $\td$ and $\Gamma_n$, subject to ``edge domination'' constraints for
$\Gamma_n$.
Section~\ref{sec:prelim} presents these notions and other related preliminary
concepts in detail.

Our contributions include the following:

\paragraph*{Answering $\faqai$ over Boolean semiring}
On the Boolean semiring, one way to answer query~\eqref{eqn:our:query} is to apply the
$\panda$ algorithm~\cite{MR3144912}, using edge domination constraints on $\calE_s$
and the set $\td$ of {\em all} tree decompositions of $\calH = (\calV, \calE = \calE_s \cup
\calE_\ell)$. However, we can do better.
In Section~\ref{sec:relaxed-td} we define a new notion of tree decomposition:
{\em relaxed tree decomposition},
in which the
hyperedges in $\calE_\ell$ only have to be covered by adjacent TD bags.
Then, we present a variant of the $\InsideOut$ algorithm running on these
relaxed TDs
using Chazelle's classic geometric data structure~\cite{MR941937} for solving the
semigroup range search problem. We show that our $\InsideOut$ variant
meets the ``relaxed $\fhtw$'' runtime, which is the analog of $\fhtw$ on relaxed TD.
The $\panda$ algorithm can use the $\InsideOut$ variant as a blackbox to meet the
``relaxed $\subw$'' runtime. The relaxed widths are smaller than the non-relaxed
counterparts, and are strictly smaller for some classes of queries, which means our
algorithms yield asymptotic improvements over existing ones.

\paragraph*{Answering $\faq$ over an arbitrary semiring}
Next, to prepare the stage for answering $\faqai$ over an arbitrary semiring, in
Section \ref{subsec:relaxed-poly} we revisit
$\faq$ over a non-Boolean semiring, where {\bf no} known algorithm can
achieve the $\subw$-runtime. Here, we relax the set $\Gamma_n$ of polymatroids to a superset
$\Gamma_n'$ of {\em relaxed polymatroids}.
Then, by adapting the $\subw$ definition to relaxed polymatroids, we obtain a new width
parameter called ``sharp submodular width'' ($\ssubw$). We show how a variant of
$\panda$, called $\spanda$, can achieve a runtime of $\tilde O(N^{\ssubw})$ for
evaluating $\faq$ over an arbitrary semiring.
We prove that $\subw \leq \ssubw \leq \fhtw$, and that there are classes of
queries for which $\ssubw$ is unboundedly smaller than $\fhtw$.

\paragraph*{Answering $\faqai$ over an arbitrary semiring}
Getting back to $\faqai$, we apply the $\ssubw$ result under both relaxations: relaxed
TD and relaxed polymatroids, to obtain a new width parameter called the relaxed $\ssubw$. We show that the new variants of $\panda$ and $\InsideOut$ can achieve the relaxed $\ssubw$ runtime.
We also show that there are queries for which relaxed $\ssubw$ is essentially the best we can hope for, modulo $\ksum$-hardness.

\paragraph*{Applications to relational Machine Learning}
Equipped with the algorithms for answering $\faqai$, in
Section~\ref{sec:applications} we return to relational machine learning
applications over training datasets defined by feature extraction queries over relational
databases.
We show how one can train linear SVM, $k$-means, and ML models using Huber/hinge loss functions without completely materializing the output of the
feature extraction queries.
In particular, this shows that for these important classes of ML
models, one can sometimes train models in time sub-linear in the size of the training dataset.

An early version of this work appeared in the proceedings of the 38th ACM Symposium on Principles of Database Systems {\em (PODS'19)}~\cite{faqai-pods}. This article goes beyond that early version by
extending the class of loss functions supported by our framework for relational machine learning,
introducing new applications for our framework on the (probabilistic) database side,
and including detailed proofs and derivation steps for various key results.

\subsection{Related work}

Appendix~\ref{sec:applications-db}
revisits two prior results on the evaluation of queries with inequalities through $\faqai$ lenses: Core XPath queries over XML documents~\cite{GKP:XPath:2002} and inequality joins over tuple-independent probabilistic databases~\cite{OH09}.
\nop{We observe that the hypergraphs of such queries admit trivial relaxed TDs, which  explain the linearithmic time complexity for their evaluation.}
Throughout the article, we contrast our new width notions with $\fhtw$ and $\subw$  and our new algorithm $\spanda$ with the state-of-the-art algorithms $\panda$ and $\InsideOut$ for $\faq$ and $\faqai$ queries.

Prior seminal work considers the containment and minimization problem for queries with inequalities~\cite{Klug:JACM:1988}.
The efficient evaluation of such queries continues to receive good attention in the database community~\cite{Khayyat:2015:LFS:2831360.2831362}.
There is a bulk of work on queries with {\em dis}equalities (not-equal), which are at times referred to as inequalities.
Queries with disequalities are a proper subclass of $\faqai$ (since $x\neq y$ can be represented as $x < y \vee x > y$).
Prior works~\cite{Koutris2017, DBLP:journals/corr/abs-1712-07445} present several results for this proper subclass that are stronger than our general results for $\faqai$ in this work.
In particular, for queries with disequalities it suffices to consider tree decompositions only for ``skeleton'' edges (ignoring ``ligament'' edges which -in this case- are the disequalities)~\cite{Koutris2017, DBLP:journals/corr/abs-1712-07445}, whereas for the more general $\faqai$ we need to consider ``relaxed'' tree decompositions (see Def.~\ref{defn:relaxed:td}).

Section~\ref{sec:applications} reviews relevant works on machine learning.


\section{Preliminaries}
\label{sec:prelim}

We assume without loss of generality that semiring operations $\oplus$ and
$\otimes$ can be performed in $O(1)$-time. (When the assumption does not hold, for the set
semiring for instance, we can multiply the claimed runtime with the real operation's
runtime.)

\subsection{Tree decompositions and polymatroids}
\label{subsec:width:params}

We briefly define tree decompositions, $\fhtw$ and $\subw$ parameters. We refer the reader to the recent survey by
Gottlob et al.~\cite{DBLP:conf/pods/GottlobGLS16} for more details and historical
contexts.
In what follows, the hypergraph $\calH$ should be thought of as the hypergraph of
the input
query, although the notions of tree decomposition and width parameters are defined
independently of queries.

A {\em tree decomposition} of a hypergraph $\calH=(\calV,\calE)$
is a pair
$(T,\chi)$, where $T$ is a tree whose nodes are $V(T)$ and $\chi: V(T) \to 2^{\calV}$ maps each node
$t$ of the tree to a subset $\chi(t)$ of vertices such that
\begin{enumerate}
  \item every hyperedge $S\in \calE$ is a subset of some $\chi(t)$, $t\in V(T)$
(i.e. every edge is covered by some bag),
  \item for every vertex $v \in \calV$,
the set $\{t \suchthat v \in \chi(t)\}$ is a non-empty (connected) sub-tree of $T$.
This is called the {\em running intersection property}.
\end{enumerate}
The sets $\chi(t)$ are called the {\em bags}
of the tree decomposition.

Let $\td(\calH)$ denote the set of all tree decompositions of $\calH$.
When $\calH$ is clear from context, we use $\td$ for brevity.

To define width parameters, we use the polymatroid characterization
from Abo Khamis et al.~\cite{panda-pods}.
A function $f : 2^{\calV} \to \R_+$ is called a (non-negative)
{\em set function} on $\calV$.
A set function $f$ on $\calV$ is {\em modular} if
$f(S) = \sum_{v\in S} f(\{v\})$ for all $S\subseteq \calV$,
{\em monotone} if $f(X) \leq f(Y)$ whenever $X \subseteq Y$,
and {\em submodular} if $f(X\cup Y)+f(X\cap Y)\leq f(X)+f(Y)$
for all $X,Y\subseteq \calV$.
A monotone, submodular set function $h : 2^{\calV} \to \R_+$ with $h(\emptyset)
= 0$ is called a {\em polymatroid}.
Let $\Gamma_n$ denote the set of all polymatroids on $\calV=[n]$.

Given $\calH$, define the set of {\em edge dominated} set
functions:
\begin{align}
    \ed &:= \{ h \suchthat h : 2^{\calV} \to \R_+, h(S) \leq 1,
\forall S \in \calE\}.\label{eqn:ed}
\end{align}
We next define the submodular width and fractional hypertree width of
a given hypergraph $\calH$:
\begin{align}
    \fhtw(\calH) &:= \min_{(T,\chi) \in \td} \max_{h \in \ed \cap \Gamma_n}
    \max_{t\in V(T)}h(\chi(t)) \label{eqn:fhtw},\\
    \subw(\calH) &:=  \max_{h \in \ed \cap \Gamma_n}\min_{(T,\chi) \in \td}
    \max_{t\in V(T)}h(\chi(t)). \label{eqn:subw}
\end{align}
It is known~\cite{MR3144912} that $\subw(\calH) \leq \fhtw(\calH)$, and there are
classes of
hypergraphs with bounded $\subw$ and unbounded $\fhtw$. Furthermore, $\fhtw$ is
strictly less than other width notions such as (generalized) hypertree width and
tree width.

\brmk
\label{rmk:alternate-fhtw}
Prior to Abo Khamis et al.~\cite{panda-pods}, the commonly used definition of
$\fhtw(\calH)$ is~\cite{DBLP:journals/talg/GroheM14}
\[\fhtw(\calH) :=\min_{(T,\chi) \in \td} \max_{t\in V(T)}
\rho^*_\calE(\chi(t)),\]
where $\rho^*_\calE(B)$ is the fractional edge cover number of a vertex set
$B$ using the hyperedge set $\calE$. It is straightforward to show, using
linear programming duality~\cite{panda-pods}, that
\begin{align}
    \max_{t \in V(T)} \max_{h \in \ed \cap \Gamma_n} h(\chi(t)) = \max_{t\in V(T)}
    \rho^*_\calE(\chi(t)), \label{eqn:dual:fhtw}
\end{align} proving the equivalence of the two definitions. However, the
characterization~\eqref{eqn:fhtw} has two primary advantages: (i) it exposes the
minimax / maximin duality between $\fhtw$ and $\subw$, and more importantly (ii)
it makes it completely straightforward to relax the definitions by replacing the
$\ed \cap \Gamma_n$ constraints by other applicable constraints, as shall be
shown in later sections.\qed
\ermk

\bdefn[$F$-connex tree
decomposition~\cite{Bagan:CSL:07,Segoufin:2013:ECD:2448496.2448498}]
\label{defn:F-connex}
Given a hypergraph $\calH=(\calV,\calE)$ and a set $F\subseteq \calV$,
a tree decomposition $(T,\chi)$ of $\calH$ is {\em $F$-connex} if there is a subset $V'\subseteq V(T)$ that
forms a connected subtree of $T$ and satisfies $\bigcup_{t\in V'}\chi(t)=F$.
(Note that $V'$ could be empty.)

We use $\td_F$ to denote the set of all $F$-connex tree decompositions
of $\calH$. (Note that when $F=\emptyset$, $\td_F=\td$.)
\edefn

\bdefn[Non-redundant tree decomposition]
A tree decomposition $(T, \chi)$ is {\em redundant} if there are $t_1 \neq t_2 \in V(T)$
where $\chi(t_1) \subseteq \chi(t_2)$.
A tree decomposition is {\em non-redundant} if it is not redundant.
\label{defn:non-redundant-td}
\edefn
The following proposition is folklore. For completeness, we prove it in Appendix~\ref{appendix:td}.
\bprop
For every tree decomposition $(T, \chi)$ of a query $Q$, there exists a non-redundant
tree decomposition $(T', \chi')$ of $Q$ that satisfies
\[\{\chi'(t) \suchthat t \in V(T')\} \subseteq \{\chi(t) \suchthat t \in V(T)\}. \]
Moreover, if $(T, \chi)$ is $F$-connex, then $(T', \chi')$ can be chosen to be $F$-connex as well.
\label{prop:non-redundant-td}
\eprop
Based on the above proposition, we only need to consider non-redundant tree decompositions
$(T, \chi)$ in~\eqref{eqn:fhtw} and~\eqref{eqn:subw} (and later on in~\eqref{eqn:faqw} and~\eqref{eqn:subfaqw}).

\subsection{$\InsideOut$ and $\panda$}
\label{subsec:insideout:panda}

To answer the $\faq$ query~\eqref{eqn:faq}, we need a model for the representation of
the input factors $R_K$. The support of the function $R_K$ is the set of tuples
$\bm x_K$ such that $R(\bm x_K) \neq \bm 0$. We use $|R_K|$ to denote the size of its support.
For example, if $R_K$ represents an input relation, then $|R_K|$ is the number of
tuples in $R_K$. In practice, there often are factors with infinite support, e.g.,
$R_K$ represents a built-in function in a database, an arithmetic
operator, or a comparison operator as in~\eqref{eqn:our:query}.
To deal with this more general setting, the edge set $\calE$ is partitioned into two sets
$\calE = \calE_{\not\infty} \cup \calE_{\infty}$, where $|R_K|$ is finite for all $K \in
\calE_{\not\infty}$
and $|R_K|=\infty$ for all $K \in \calE_{\infty}$.
For simplicity, we often state runtimes of algorithms in terms of the ``input size''
$N := \max_{K \in \calE_{\not\infty}} |R_K|$.
Moreover, we use $|Q|$ to denote the output size of $Q$.
We always assume that $\bigcup_{S\in \calE_{\not\infty}} S = \calV$; otherwise the output size $|Q|$ could be infinite.

\paragraph*{$\InsideOut$~\cite{DBLP:conf/pods/KhamisNR16,faq-arxiv,DBLP:journals/sigmod/Khamis0R17}}
To answer~\eqref{eqn:faq}, the $\InsideOut$
algorithm
works by
eliminating variables, along with an idea called the ``indicator projection''
(see Appendix~\ref{app:insideout} for more details).
The runtime is
described by the {\em $\faq$-width} of the query, a slight generalization of $\fhtw$.
For one semiring, we can define $\faqw(Q)$ by applying
Definition~\eqref{eqn:fhtw} over a restricted set of tree decompositions and edge dominated
polymatroids. In particular, let $F \subseteq \calV$ denote the set of free variables
in~\eqref{eqn:faq}, and recall $\td_F$ from Definition~\ref{defn:F-connex}.
Then,
\begin{align}
    \ed_{\not\infty} &:= \{ h \suchthat h : 2^{\calV} \to \R_+, h(S) \leq 1,
    \forall S \in \calE_{\not\infty}\}\label{eqn:ed:f}, \\
    \faqw(Q) &:=
    \min_{(T,\chi) \in \td_F} \max_{h \in \ed_{\not\infty} \cap \Gamma_n}
    \max_{t\in V(T)}h(\chi(t)) \label{eqn:faqw}\\
    (\text{remark~\ref{rmk:alternate-fhtw}})&= \min_{(T,\chi) \in \td_F}
    \max_{t\in V(T)} \rho^*_{\calE_{\not\infty}}(\chi(t))\label{eqn:faqw2}
\end{align}
Note that $\faqw(Q) = \fhtw(\calH)$ when $F =\emptyset$ and $\calE_\infty =
\emptyset$ (i.e. $\calE =\calE_{\not\infty}$).
A simple result from Abo Khamis et al.~\cite{DBLP:conf/pods/KhamisNR16} is the following:
(Recall that throughout the article we assume the query size to be a constant and state runtimes in data complexity.)
\bthm[\cite{DBLP:conf/pods/KhamisNR16,faq-arxiv}]
$\InsideOut$ answers query~\eqref{eqn:faq} in time $O(N^{\faqw(Q)}\cdot\log N+|Q|)$.
\label{thm:insideout}
\ethm
A proof sketch of the above theorem can be found in Appendix~\ref{app:insideout}.
To solve the $\faqai$~\eqref{eqn:our:query}, we can apply
Theorem~\ref{thm:insideout} with $\calE_\infty \supseteq \calE_\ell$
since all ligament
factors are infinite. But this is suboptimal---later, we show a
new $\InsideOut$ variant that is polynomially better.

\paragraph*{$\panda$~\cite{panda-pods,panda-arxiv}}
For the Boolean semiring, i.e., when the $\faq$ query~\eqref{eqn:faq} is of the form
\begin{align}
    Q(\bm x_F) &=
    \bigvee_{\bm x_{\calV \setminus F} \in \prod_{i \in \calV \setminus F} \dom(X_i)}\
    \bigwedge_{K\in\calE} R_K(\bm x_K),
    \label{eqn:faq:boolean}
\end{align}
we can do much better than Theorem~\ref{thm:insideout}. When $F =\emptyset$,
Marx~\cite{MR3144912} showed that~\eqref{eqn:faq:boolean} can be answered in time $\tilde
O(N^{O(\subw(Q))})$. The $\panda$ algorithm~\cite{panda-pods,panda-arxiv} generalizes Marx's result
to deal with general degree constraints, and to meet precisely the $\tilde
O(N^{\subw(Q)})$-runtime (see Appendix~\ref{app:panda} for more details).
In fact, $\panda$ works with queries such as~\eqref{eqn:faq:boolean}
with free variables as well. In the context of this article, we can define the following
notion of {\em submodular $\faq$-width} in a natural way:

\begin{align}
    \subfaqw(Q) &:=
    \max_{h \in \ed_{\not\infty} \cap \Gamma_n}\min_{(T,\chi) \in \td_F}
    \max_{t\in V(T)} h(\chi(t)).
    \label{eqn:subfaqw}
\end{align}

\noindent Then, the results from Abo Khamis et al.~\cite{panda-pods} imply:
\bthm[\cite{panda-pods,panda-arxiv}]
$\panda$ answers query~\eqref{eqn:faq:boolean} in time $\tilde O(N^{\subfaqw(Q)}+|Q|)$.
\label{thm:panda}
\ethm
Appendix~\ref{app:panda} presents an overview of the core $\panda$ algorithm and its analysis.
The $\panda$ results only work for the Boolean semiring.
Section~\ref{sec:faqai} introduces a variant of $\panda$, called $\spanda$, that
also works for non-Boolean semirings.

\subsection{Semigroup range searching}
\label{subsec:semigroup}

{\sf Orthogonal range counting} (and searching)
is a classic and ubiquitous problem in computational geometry~\cite{MR2723879}:
given a set $S$ of
$N$ points in a $d$-dimensional space, build a data structure that,
given any $d$-dimensional rectangle, can efficiently return the number of enclosed
points.
More generally, there is the {\sf semigroup range searching} problem~\cite{MR941937}, where
each point $\bm p \in S$ of the $N$ input points also has a weight $w(\bm p) \in G$,
where $(G, \oplus)$ is a {\em semigroup}.\footnote{In a semigroup we
can add two elements using $\oplus$, but there is no additive inverse.}
The problem is: given a $d$-dimensional rectangle
$R$, compute $\bigoplus_{\bm p \in S \cap R} w(\bm p)$.

Classic results by Chazelle~\cite{MR941937} show that there are data
structures for {\sf semigroup range searching} which can be constructed in time
$O(N\log^{d-1} N)$, and answer rectangular queries in $O(\log^{d-1} N)$-time.
Also, this is almost the best we can hope for~\cite{MR1072265}.
There are more recent improvements to Chazelle's result (see, e.g.,
Chan et al.~\cite{Chan:2011:ORS:1998196.1998198}),
but they are minor (at most a $\log$ factor), as the
original results were already very close to matching the lower bound.

Most of these range search/counting problems can be reduced to the {\sf dominance range
searching} problem (on semigroups),
where the query is represented by a point $\bm q$, and the objective is
to return $\bigoplus_{\bm q \preceq \bm p, \bm p \in S} w(\bm p)$. Here, $\preceq$ denotes
the ``dominance'' relation (coordinate-wise $\leq$). We can think of $\bm q$ as the
lower-corner of an infinite rectangle query.

\section{Relaxed tree decompositions and relaxed polymatroids}
\label{sec:faqai}

\subsection{Connection to semigroup range searching}
We always assume that $\bigcup_{S\in \calE_{s}} S = \calV$; otherwise the output size $|Q|$ could be infinite.
We start with a special case of~\eqref{eqn:our:query} in which the skeleton part $\calE_s$
contains only {\em two} hyperedges $U$ and $L$.
Consider the aggregate query of the form

\begin{align}
    Q(\bm x_F) = \bigoplus_{\bm x_{\calV \setminus F}} \Phi_1(\bm x_U) \otimes
    \Phi_2(\bm x_L) \otimes
\left(\bigotimes_{S \in \calE_\ell}  \bm 1_{\sum_{v \in S} \theta^S_v(x_v) \leq 0} \right),
    \label{eqn:two:bags}
\end{align}

\noindent where $\Phi_1$ and $\Phi_2$ are two input functions/relations over variable sets $U$ and $L$, respectively.
We prove the following simple but important lemma:
\blmm
Let $N = \max\{|\Phi_1|, |\Phi_2|\}$, and $k = |\calE_\ell|$. For $F \subseteq U$,
query~\eqref{eqn:two:bags} can be answered in time $O(N \cdot (\log N)^{\max(k-1,1)})$.
\label{lmm:two:bags}
\elmm
\bp
If there is a hyperedge $S \in \calE_\ell$ for which $S \subseteq U$, then in a
$O(N\log N)$-time pre-processing step we can ``absorb'' the factor
$\bm 1_{\sum_{v\in S}\theta^S_v(x_v) \leq 0}$ into the factor $\Phi_1$, by replacing
$\Phi_1(\bm x_U)$ with the product $\Phi_1(\bm x_U) \otimes \bm 1_{\sum_{v\in S}\theta^S_v(x_v) \leq 0}$.
In particular, this product
can be computed by iterating over tuples $\bm x_U$ satisfying $\Phi_1(\bm x_U) \neq \bm 0$
and for each such tuple $\bm x_U$, testing whether the inequality $\sum_{v\in S}\theta^S_v(x_v) \leq 0$
holds. If it does, then the indicator $\bm 1_{\sum_{v\in S}\theta^S_v(x_v) \leq 0}$ takes a value of $\bm 1$,
hence the value of $\Phi_1(\bm x_U)$ remains unchanged after the product. Otherwise,
both the indicator $\bm 1_{\sum_{v\in S}\theta^S_v(x_v) \leq 0}$ and its product with $\Phi_1(\bm x_U)$
take a value of $\bm 0$. A similar absorption can be done with $S \subseteq L$.
Hence, without loss of generality we can assume that $S \not\subseteq L$ and $S \not\subseteq U$
for all $S \in \calE_\ell$.

Moreover, we only need to show that we can compute~\eqref{eqn:two:bags} for $F=U$,
because after $Q(\bm x_U)$ is computed, we can ``aggregate away'' variables
$\bm x_{U\setminus F}$ in $O(N\log N)$-time by computing the aggregation:
\[Q(\bm x_F) = \bigoplus_{\bm x_{U\setminus
F}} Q(\bm x_U).\]
The above aggregation can be computed by sorting tuples $\bm x_U$ satisfying $Q(\bm x_U) \neq \bm 0$
lexicographically based on $(\bm x_F, \bm x_{U\setminus F})$
so that tuples $\bm x_U$ sharing the same $\bm x_F$-prefix become consecutive.
Then for each distinct $\bm x_F$-prefix, we aggregate away $Q(\bm x_U)$ over all
tuples $\bm x_U$ sharing that prefix.

Abusing notation somewhat, for each $S \in \calE_\ell$ and each $T \subseteq S$, define the
function $\theta^S_T : \prod_{v\in T}\dom(X_v) \to \R$ by
\begin{align}
    \theta^S_T(\bm x_T) &:= \sum_{v \in T}\theta^S_v(x_v).
\end{align}
Fix a tuple $\bm x_U$ such that $\Phi_1(\bm x_U) \neq \bm 0$. A tuple $\bm x_L$ is said
to be {\em $\bm x_U$-adjacent} if $\pi_{U\cap L} \bm x_U = \pi_{U\cap L}\bm x_L$.
We show how to compute the following sum in poly-logarithmic time:
\begin{eqnarray}
    \bigoplus_{\bm x_{L \setminus U}} \Phi_1(\bm x_U) \otimes \Phi_2(\bm x_L)
    \otimes \left(\bigotimes_{S \in \calE_\ell}  \bm 1_{\sum_{v \in S} \theta^S_v(x_v) \leq 0} \right) =
    \label{eqn:partial:sum}\\
    \ \ \ \Phi_1(\bm x_U)\otimes
    \bigoplus_{\bm x_{L \setminus U}}
    \Phi_2(\bm x_L)\otimes
    \left(\bigotimes_{S \in \calE_\ell}  \bm 1_{
    \theta^S_{S\cap U}(\bm x_{S\cap U}) \leq -\theta^S_{S\setminus U} (\bm
    x_{S\setminus U})} \right). \label{eqn:total:weight}
\end{eqnarray}
where the inner sum ranges only over tuples $\bm x_L$
which are $\bm x_U$-adjacent. This is because the value of $\bm x_{U\cap L}$ has been fixed
and tuples $\bm x_L$ that are not $\bm x_U$-adjacent are inconsistent with the fixed value of
$\bm x_{U\cap L}$.

Now, for the fixed $\bm x_U$ and for each $\bm x_L$ define the following $k$-dimensional points:
\begin{align*}
    \bm q(\bm x_U) &= (q_S(\bm x_U))_{S\in \calE_\ell} \quad \text{where} \quad
    q_{S}(\bm x_U) := \theta^S_{S\cap U}(\bm x_{S\cap U}), \\
    \bm p(\bm x_L) &= (p_S(\bm x_L))_{S\in \calE_\ell} \quad \text{where} \quad
    p_{S}(\bm x_L) := -\theta^S_{S\setminus U}(\bm x_{S\setminus U}).
\end{align*}
We write $\bm q(\bm x_U) \preceq \bm p(\bm x_L)$ to say that
$\bm q(\bm x_U)$ is
dominated by $\bm p(\bm x_L)$ {\em coordinate-wise}: $q_S(\bm x_U)
\leq p_S(\bm x_L) \;\forall\; S \in \calE_\ell$.
Assign to each point $\bm p(\bm x_L)$ a ``weight'' of $\Phi_2(\bm x_L)$.
Now, taking~\eqref{eqn:total:weight},

\begin{align}
    \bigoplus_{\bm x_{L \setminus U}}
    \Phi_2(\bm x_L)
    \otimes \left(\bigotimes_{S \in \calE_\ell}  \bm 1_{\theta^S_{S\cap U} (\bm
    x_{S\cap U}) \leq -
    \theta^S_{S\setminus U}(\bm x_{S\setminus U})} \right)
    &=
    \bigoplus_{\bm x_{L \setminus U}}
    \left(\bigotimes_{S \in \calE_\ell}  \bm 1_{q_S(\bm x_U) \leq p_S(\bm
    x_L)}\right) \otimes \Phi_2(\bm x_L) \\
    &=
    \bigoplus_{\bm x_{L \setminus U}}
    \bm 1_{\bm q(\bm x_U) \preceq \bm p(\bm x_L)} \otimes \Phi_2(\bm x_L).
\end{align}

\noindent (The equality $\bigotimes_{S \in \calE_\ell}  \bm 1_{q_S(\bm x_U) \leq p_S(\bm
x_L)} = \bm 1_{\bm q(\bm x_U) \preceq \bm p(\bm x_L)}$ used above follows from the definition
of the component-wise $\preceq$.)
The expression thus computes, for a given ``query point'' $\bm q(\bm x_U)$,
the weighted sum over all points $\bm p(\bm x_L)$ that dominate the query point.
This is precisely the {\sf dominance range counting} problem, which---modulo a
$O(N(\log N)^{\max(k-1,1)})$-preprocessing step---can be solved in time $O((\log
N)^{\max(k-1,1)})$ \cite{MR941937}, as reviewed in
Section~\ref{subsec:semigroup}.

\ep

\begin{example}
    Let $R$ be a binary relation.
    Suppose we want to count the number of tuples satisfying $R(a,b) \wedge R(b,c) \wedge
    a<c$. By setting $F=\emptyset$, $U=\{a,b\}$, $L=\{b,c\}$,
    the problem can be reduced to the form~\eqref{eqn:two:bags} with $k=1$,
    $\calE_\ell=\{ \{a,c\} \}$. We can thus compute this count in time
    $O(N\log N)$.\qed
\end{example}

\subsection{Relaxed tree decompositions}
\label{sec:relaxed-td}

Equipped with this basic case, we can now proceed to solve the general setting
of~\eqref{eqn:our:query}. To this end, we define a new width parameter.

\bdefn[Relaxed tree decomposition]
Let $\calH = (\calV,$ $\calE = \calE_s \cup \calE_\ell)$ denote a multi-hypergraph
whose edge multiset is partitioned into $\calE_s$ and $\calE_\ell$.
A {\em relaxed} tree decomposition of $\calH$ (with respect to the partition $\calE_s \cup
\calE_\ell$) is a pair $(T, \chi)$, where $T=(V(T),E(T))$ is
a tree whose nodes and edges are $V(T)$ and $E(T)$ respectively, and $\chi : V(T) \to 2^{\calV}$ satisfies the following properties:
\bi
\item[(a)] The running intersection property holds: for each node $v\in\calV$ the set $\{t \in V(T) \suchthat v \in \chi(t)\}$ is a connected subtree in $T$.
\item[(b)] Every ``skeleton'' edge $S \in \calE_s$ is covered by some bag $\chi(t)$, $t \in V(T)$.
\item[(c)] Every ``ligament'' edge $S \in \calE_\ell$ is covered by the union of two {\em adjacent}
    bags $s$ and $t$, i.e. $S \subseteq \chi(s) \cup \chi(t)$, where $\{s, t\} \in E(T)$.
\ei
Let $\td^\ell(\calH)$ denote the set of all relaxed tree decompositions of $\calH$
(with respect to the skeleton-ligament partition).
When $\calH$ is clear from context we use $\td^\ell$ for the sake of brevity.
Given $F\subseteq \calV$, let $\td^\ell_F$ denote the set of all relaxed $F$-connex tree decompositions
of $\calH$.
\label{defn:relaxed:td}
\edefn

The new condition (c) in the above definition is needed so that later we can utilize
Lemma~\ref{lmm:two:bags} to compute aggregate queries over the relaxed tree decomposition.
In particular, the two adjacent bags $s$ and $t$ in condition (c) will play the role of
$U$ and $L$ from Lemma~\ref{lmm:two:bags} and the corresponding query~\eqref{eqn:two:bags}.

\subsubsection{$\faqai$ on a general semiring}

We use relaxed TDs in conjunction with Lemma~\ref{lmm:two:bags} to answer
$\faqai$ with a relaxed notion of $\faqw$.
In particular, the {\em relaxed} width parameters of $\calH$ are defined in
exactly the
same way as the usual width parameters defined in Section~\ref{sec:prelim},
except we allow the TDs to range over relaxed ones.

\bdefn[Relaxed $\faqw$]
Let $Q$ be an $\faqai$ query~\eqref{eqn:our:query}, and
$\calH = (\calV, \calE = \calE_s \cup \calE_\ell)$ be its hypergraph. Furthermore, let
$\calE_{\not\infty} \subseteq \calE_s$ denote the set of hyperedges $K \in
\calE$ for which
$|R_K|<\infty$. Then, the {\em relaxed $\faq$-width} of
$Q$ is defined by

\begin{align}
    \faqw_\ell(Q) &:=
    \min_{(T,\chi) \in \td^\ell_F} \max_{h \in \ed_{\not\infty} \cap \Gamma_n}
    \max_{t\in V(T)}h(\chi(t)) \label{eqn:relaxed:faqw}
\end{align}

When $F=\emptyset$, $\faqw_\ell$ collapses to $\fhtw_\ell$ which is the {\em relaxed $\fhtw$} for $\faqai$ $Q$ {\em without} free variables:

\begin{align}
    \fhtw_\ell(Q) &:=
    \min_{(T,\chi) \in \td^\ell_\emptyset}
    \max_{h \in \ed_{\not\infty} \cap \Gamma_n}
    \max_{t\in V(T)} h(\chi(t))
    \label{eqn:relaxed:fhtw}
\end{align}
\edefn

A relaxed tree decomposition $(T,\chi)$ of $Q$ is {\em optimal} if its width is
equal to $\faqw_\ell$, i.e., \[\faqw_\ell(Q) =  \max_{h \in \ed_{\not\infty}
\cap \Gamma_n} \max_{t\in V(T)}h(\chi(t)).\]

\bthm
Any $\faqai$ query $Q$ of the form~\eqref{eqn:our:query} on any semiring can be
answered in time
$O(N^{\faqw_\ell(Q)} \cdot (\log N)^{\max(k-1,1)}+|Q|)$, where $k$ is the maximum number of additive inequalities covered by a pair
of adjacent bags in an optimal relaxed tree decomposition.\footnote{Note that $k$ can be a lot smaller than $|\calE_\ell|$ since different additive inequalities can be covered by different pairs of adjacent bags in an optimal relaxed hypertree decomposition.}
\label{thm:relaxed:faqw}
\ethm
\bp
We first consider the case of no free variables because
this case captures the key idea.
Fix an optimal relaxed tree decomposition $(T, \chi)$. We first compute, for each bag $t \in V(T)$
of the tree decomposition, a factor $\Phi_t : \prod_{i \in \chi(t)} \dom(X_i) \to \bm D$
such that
\begin{align}
    Q() &=
    \bigoplus_{\bm x_{\calV}}
    \left(\bigotimes_{K\in \calE_s} R_K(\bm x_K)\right) \otimes
    \left(\bigotimes_{S \in \calE_\ell} \bm 1_{\sum_{v \in S} \theta^S_v(x_v) \leq 0} \right) \\
    &=\bigoplus_{\bm x_{\calV}}
    \left(\bigotimes_{t\in V(T)} \Phi_t(\bm x_{\chi(t)})\right) \otimes
    \left(\bigotimes_{S \in \calE_\ell} \bm 1_{\sum_{v \in S} \theta^S_v(x_v) \leq 0} \right). \label{eqn:query:F:empty}
\end{align}
To define the factors $\Phi_t$, we need the notion of
{\em indicator projection}~\cite{DBLP:journals/sigmod/Khamis0R17,faq-arxiv,DBLP:conf/pods/KhamisNR16};
see Appendix~\ref{app:insideout} for some background about the $\InsideOut$ algorithm
where this notion was originally developed.
\bdefn[Indicator Projection~\cite{DBLP:conf/pods/KhamisNR16,faq-arxiv,DBLP:journals/sigmod/Khamis0R17}]
For a given $K \in \calE_{\not\infty}$ and $I \subseteq \calV$ such that
$J := K \cap I \neq \emptyset$, the indicator
projection of $R_K$ onto the set $I$ is a function
$\pi_{K,I} : \prod_{v\in J} \dom(X_v) \to \{\bm 0, \bm 1\}$ defined by

\begin{align}
    \pi_{K,I}(\bm x_J) &:=
    \begin{cases}
        \bm 1 & \exists \bm x_{K\setminus J} \text { s.t. } R_K((\bm x_J, \bm x_{K\setminus J})) \neq \bm 0,\\
        \bm 0 & \text{otherwise.}
    \end{cases}
\end{align}
\label{defn:indicator:proj}
\edefn
Based on the above definition, it is easy to verify that for any
$K \in \calE_{\not\infty}$ and $I \subseteq \calV$ such that
$J := K \cap I \neq \emptyset$, we have the identity
\begin{equation}
    R_K(\bm x_k) = R_K(\bm x_k) \otimes \pi_{K,I}(\bm x_J).
    \label{eq:indicator:proj:prop}
\end{equation}

Recall from Definition~\ref{defn:relaxed:td} that every $K \in \calE_s$ is
covered by
at least one bag $\chi(t)$ for $t \in V(T)$.
Fix an arbitrary coverage assignment $\alpha : \calE_s \to V(T)$, where $K$ is
covered
by the bag $\chi(\alpha(K))$. Then, the factors $\Phi_t$ are defined by:
\begin{align}
    \Phi_t(\bm x_{\chi(t)}) &:= \bigotimes_{K \in \alpha^{-1}(t)} R_K(\bm x_K)
    \otimes \bigotimes_{\substack{K \in \calE_{\not\infty} \\ K \cap \chi(t) \neq \emptyset}}
    \pi_{K,\chi(t)}(\bm x_{K \cap \chi(t)}). \label{eqn:phi:t}
\end{align}

\begin{claim}
    The factors $\Phi_t$ defined by~\eqref{eqn:phi:t} satisfy~\eqref{eqn:query:F:empty}.
\end{claim}
The above claim can be proved as follows:
\begin{eqnarray}
\bigotimes_{K\in \calE_s} R_K(\bm x_K) &=&
\bigotimes_{t\in V(t)}\bigotimes_{K \in \alpha^{-1}(t)} R_K(\bm x_K)\\
\text{(by~\eqref{eq:indicator:proj:prop})}
&=& \bigotimes_{t\in V(t)}\underbrace{\left[\bigotimes_{K \in \alpha^{-1}(t)} R_K(\bm x_K)
\otimes \bigotimes_{\substack{K \in \not\infty \\ K \cap \chi(t) \neq \emptyset}}
\pi_{K,\chi(t)}(\bm x_{K \cap \chi(t)})\right]}_{\Phi_t(\bm x_{\chi(t)})}
\end{eqnarray}

For every $t \in V(T)$, the query $\Phi_t$ can be reduced to a join query and solved
using a worst-case optimal join algorithm~\cite{DBLP:conf/pods/000118,
DBLP:conf/pods/NgoPRR12, DBLP:conf/icdt/Veldhuizen14} as follows.
For every $K\in\calE_{\not\infty}$ where $K\cap \chi(t) \neq \emptyset$,
define $\overline\pi_{K,\chi(t)}$ to be the {\em support} of the factor $\pi_{K,\chi(t)}$,
which is the set of tuples $\bm x_{K\cap \chi(t)}$ satisfying
$\pi_{K,\chi(t)}(\bm x_{K \cap \chi(t)}) \neq \bm 0$:
\[\overline\pi_{K,\chi(t)}:=
\left\{\bm x_{K\cap \chi(t)} \in \prod_{v\in K \cap \chi(t)} \dom(X_v)\suchthat
\pi_{K,\chi(t)}(\bm x_{K \cap \chi(t)}) \neq \bm 0
\right\}.\]
$\overline\pi_{K,\chi(t)}$ can be viewed as a relation over variables $\bm x_{K\cap \chi(t)}$.
Computing $\Phi_{t}$ can be reduced to solving the
join query $\overline\Phi_{t}$ defined as:
\begin{eqnarray}
    \overline\Phi_{t}(\bm x_{\chi(t)}) &:=&
\Join_{\substack{K\in\calE_{\not\infty}\\K\cap \chi(t) \neq \emptyset}} \overline\pi_{K,\chi(t)}(\bm x_{K\cap \chi(t)})
\end{eqnarray}
This is because once we solve the join query $\overline\Phi_{t}$,
the factor $\Phi_{t}$ can be computed as follows:
\begin{eqnarray*}
    \Phi_{t}(\bm x_{\chi(t)}) =
    \begin{cases}
        \bigotimes_{K \in \alpha^{-1}(t)} R_K(\bm x_K)
        &\text{if $\bm x_{\chi(t)}\in \overline\Phi_{t}$}\\
        \bm 0&\text{ otherwise}
    \end{cases}
\end{eqnarray*}
where $\overline\Phi_{t}$ above denotes the {\em output} of the join query $\overline\Phi_{t}$.
The join query $\overline\Phi_{t}$ can be computed using a worst-case optimal
join algorithm in time
\begin{equation}
O(N^{\rho^*_{\calE_{\not\infty}}(\chi(t))}\cdot \log N) =
O(N^{\max_{h\in\ed_{\not\infty}\cap \Gamma_n} h(\chi(t))}\cdot \log N).
\end{equation}
Over all $t \in V(T)$, our runtime is bounded by
$O(N^w\log N)$, where
\begin{equation}w := \max_{t \in V(T)}\max_{h\in\ed_{\not\infty}\cap \Gamma_n}
h(\chi(t)).
\label{eq:w-of-td}
\end{equation}
Moreover for every $t \in V(T)$, the output size of the join query $\overline\Phi_{t}$ is bounded by
$N^{\rho^*_{\calE_{\not\infty}}(\chi(t))} \leq N^w$, thanks to the AGM bound~\cite{10.1109/FOCS.2008.43,DBLP:journals/talg/GroheM14}.

Next we compute~\eqref{eqn:query:F:empty} in time $O(N^w\log N)$.
We will make use of the fact that $(T,\chi)$ is a relaxed TD.
Fix an arbitrary root of the tree decomposition $(T,\chi)$;
following $\InsideOut$ (Appendix~\ref{app:insideout}), we compute~\eqref{eqn:query:F:empty} by eliminating
variables from the leaves of $(T,\chi)$ up to the root.
Thanks to Proposition~\ref{prop:non-redundant-td}, we can assume the tree decomposition
to be non-redundant.
Let $t_1$ be any
leaf of $(T,\chi)$, $t_2$ be its parent, where $L = \chi(t_1)$ and $U
= \chi(t_2)$.
Because of non-redundancy, we have $L\setminus U \neq \emptyset$.
Now write~\eqref{eqn:query:F:empty} as follows:
\begin{align}
&
\bigoplus_{\bm x_{\calV}}
\left(\bigotimes_{t\in V(T)} \Phi_t(\bm x_{\chi(t)})\right) \otimes
\left(\bigotimes_{S \in \calE_\ell} \bm 1_{\sum_{v \in S} \theta^S_v(x_v) \leq 0} \right) \nonumber \\
    &\ \ \ \ =
 \bigoplus_{\bm x_{\calV\setminus (L\setminus U)}} \bigoplus_{\bm x_{L \setminus U}}
\left(\bigotimes_{t\in V(T)}  \Phi_t(\bm x_{\chi(t)})\right) \otimes
\left(\bigotimes_{S \in \calE_\ell}  \bm 1_{\sum_{v \in S} \theta^S_v(x_v) \leq 0} \right) \nonumber \\
    &\ \ \ \ =
 \bigoplus_{\bm x_{\calV\setminus (L\setminus U)}}
\left( \bigotimes_{t\in V(T)\setminus\{t_1,t_2\}} \Phi_t(\bm x_{\chi(t)})\right) \otimes
\left(
\bigotimes_{\substack{S \in \calE_\ell \nonumber\\ S \cap (L\setminus U)=\emptyset}}
 \bm 1_{\sum_{v \in S} \theta^S_v(x_v) \leq 0} \right) \nonumber \\
&\ \ \ \ \ \ \otimes \underbrace{\left[ \bigoplus_{\bm x_{L \setminus U}}
\Phi_{t_1}(\bm x_L)  \otimes  \Phi_{t_2}(\bm x_U)
 \otimes
\left(
\bigotimes_{\substack{S \in \calE_\ell \\ S \cap (L\setminus U)\neq \emptyset}}
 \bm 1_{\sum_{v \in S} \theta^S_v(x_v) \leq 0} \right)
\right].}_{\text{a sub-query $\varphi_U(\bm x_{U})$ of the form~\eqref{eqn:two:bags} with free vars $U$}}
\label{eqn:two-bags-subquery}
\end{align}
The third equality uses the semiring's distributive law.
(Note that $S\cap(L\setminus U)\neq \emptyset$ implies that $S \subseteq (L
\cup U)$ thanks to Definition~\ref{defn:relaxed:td} and the fact that $t_2$ is
the
only neighbor of $t_1$.)
Lemma~\ref{lmm:two:bags} implies
that we can compute the sub-query $\varphi_U(\bm x_{U})$ from~\eqref{eqn:two-bags-subquery} in the allotted time. The above step eliminates all
variables in $L \setminus U$.
In particular after this step,
the original query $Q()$ from~\eqref{eqn:query:F:empty} becomes:
\begin{eqnarray}
Q() &=& \bigoplus_{\bm x_{\calV\setminus (L\setminus U)}}
\left( \varphi_U(\bm x_{U}) \otimes \bigotimes_{t\in V(T)\setminus\{t_1,t_2\}} \Phi_t(\bm x_{\chi(t)})\right) \otimes
\left(
\bigotimes_{\substack{S \in \calE_\ell\\ S \cap (L\setminus U)=\emptyset}}
 \bm 1_{\sum_{v \in S} \theta^S_v(x_v) \leq 0} \right)
 \label{eqn:eliminating:L-U}
\end{eqnarray}
The above is an $\faq$ of the same form~\eqref{eqn:our:query} except that
it no longer involves the variables $\bm x_{L\setminus U}$ (Recall
that $L\setminus U\neq \emptyset$).
It admits a tree decomposition that results from the original tree decomposition
$(T, \chi)$ by removing the leaf $t_1$.
In particular, the new factor $\varphi_U(\bm x_{U})$ in~\eqref{eqn:eliminating:L-U}
is covered by the bag $\chi(t_2)$ and all other properties of tree decompositions continue to hold
after the removal of $t_1$.
By induction on the number of variables, we solve the new query~\eqref{eqn:eliminating:L-U}
in time $O(N^w\log N)$. Induction completes the proof.
(In the base case, we have a query with no variables where the theorem holds trivially.)

When the query has free variables, the algorithm proceeds similarly to the case of an $\faq$ with free
variables~\cite{DBLP:conf/pods/KhamisNR16,faq-arxiv}. See Appendix~\ref{app:insideout}
for a recap of how to handle free variables in an $\faq$.
\ep

\begin{example}
\label{ex:relaxed-td}
Given three binary relations $R, S$ and $T$,
consider a query $Q$ that counts the number of tuples $(a, b, c, d)$ that
satisfy:
\begin{equation}
R(a, b) \wedge S(b, c) \wedge T(c, d) \wedge
(a \leq c) \wedge (c \leq b) \wedge (d \leq b).
\end{equation}
The query $Q$ has $\calE_s = \calE_{\not\infty} = \{\{a, b\}, \{b, c\}, \{c,
d\}\}$
and $\calE_\ell = \calE_{\infty} =$ $\{\{a, c\}, \{b, c\},$ $\{b, d\}\}$.
Let $N = \max\{|R|, |S|, |T|\}$.
Note that $\faqw(Q) = 2$.
In fact, any of the previously known algorithms,
e.g.~\cite{DBLP:conf/pods/KhamisNR16,DBLP:journals/sigmod/Khamis0R17}, would
take time $O(N^2)$ to answer $Q$.
However, this query has $\faqw_\ell(Q) = 1$, and by
Theorem~\ref{thm:relaxed:faqw}, it can be answered in time $O(N\cdot \log N)$.
(Note that here $2=k < |\calE_\ell| = 3$.)
An optimal relaxed tree decomposition is shown in Figure~\ref{fig:relaxed-td}.\qed
\end{example}

\begin{figure}
\centering{
\tikzstyle{TDNode} = [draw, ellipse, align=center]
\tikzstyle{TDEdge} = [thick]
\tikzstyle{LEdge} = [dashed]
\begin{tikzpicture}[scale=.9, every node/.style={transform shape}]
\pgfmathsetmacro{\radius}{.5}
\pgfmathsetmacro{\width}{.4}
\node[TDNode] at(0,0) (ab) {$a$\hspace{\width in}$b$};
\node[TDNode] at(4,0) (cd) {$c$\hspace{\width in}$d$};
\node[TDNode] at(2,1) (bc) {$b$\hspace{\width in}$c$};
\draw[TDEdge] (ab)--(bc);
\draw[TDEdge] (cd)--(bc);
\draw[LEdge] ($(ab)+(-\radius,0.15)+(-0.05,0)$) to[out = 80, in = 130] ($(bc)+(\radius,0.1)$);
\draw[LEdge] ($(cd)+(+\radius,0.15)+(0.05,0)$) to[out = 100, in = 50]
($(bc)+(-\radius,0.1)$);
\draw[LEdge] ($(bc)+(-\radius,-0.1)$) to[out = -30, in = 210]
($(bc)+(+\radius,-0.1)$);
\node at (-0.1, 1.2) [rotate=0] {$\leq$};
\node at (4.1, 1.2) [rotate=0] {$\geq$};
\node at (2, 0.95) {$\geq$};
\end{tikzpicture}
}
\caption{An optimal relaxed tree decomposition for the query in
Example~\ref{ex:relaxed-td}. {\em Ligament} edges are dashed.
Each {\em skeleton} edge is held in one bag.}
\label{fig:relaxed-td}
\end{figure}
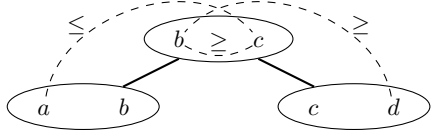

We next give a couple of simple lower and upper bounds for $\faqw_\ell$.
The upper bound shows that, effectively $\faqw_\ell$ is the best we can hope for, if the $\faqai$ query is arbitrary. The lower bound shows that, while the relaxed tree decomposition idea can improve the runtime by a polynomial factor, it cannot improve the runtime over straightforwardly applying $\InsideOut$ (over non-relaxed tree decompositions) by more than a polynomial factor.

\bprop
For any positive integer $m$, there exists an $\faqai$ query of the
form~\eqref{eqn:our:query} for which $F=\emptyset$, $\faqw_\ell(Q) \geq m$ and
it cannot be answered in time $o(N^{\faqw_\ell(Q)})$, modulo $\ksum$ hardness.
\label{prop:kum:lowerbound}
\eprop

\bp
It is widely assumed~\cite{DBLP:conf/soda/PatrascuW10,DBLP:conf/icalp/LincolnWWW16} that $O(N^{\lceil k/2 \rceil})$ is the best runtime for
$\ksum$, which is the following problem: given $k$ number sets $R_1,\dots,R_k$ of  maximum size $N$,
determine whether there is a tuple $\bm t \in R_1 \times \cdots \times R_k$ such that
$\sum_{i\in[k]} t_i = 0$. We can reduce $\ksum$ to our problem: Consider the
query $Q$ over the Boolean semiring:
\begin{align}
	Q() \leftarrow \left(\bigwedge_{i\in[k]} R_i(x_i)\right) \wedge
	\left(\sum_{i\in[k]} x_i\leq 0\right) \wedge \left(\sum_{i\in[k]} x_i\geq
	0\right).
	\label{eqn:our:ksum:query}
\end{align}
The answer to $Q$ is true iff there is a tuple  $(x_1,\dots,x_k) \in
R_1\times\cdots\times R_k$ such that $\sum_{i\in[k]}x_i = 0$. The reduction
shows that our
query~\eqref{eqn:our:ksum:query} is $\ksum$-hard. For this query, $\faqw_\ell(Q) = \lceil k / 2 \rceil$.
\nop{
	By answering the above $\faqai$ query and then perturbing the instance to reduce each number
	$x_i$ by a tiny amount, the difference between the two answers is exactly the number of
	tuples $(x_1,\dots,x_k)$ summing up to $0$. The reduction also shows that our
	query~\eqref{eqn:our:ksum:query} is $\ksum$ hard even in the Boolean semiring!
	For query~\eqref{eqn:our:ksum:query}, $\faqw_\ell(Q) = \lceil k / 2 \rceil$. This  completes the proof.
}
\ep

\bprop
For any $\faqai$ query $Q$ of the form~\eqref{eqn:our:query},
we have $\faqw_\ell(Q) \geq \frac 1 2 \faqw(Q)$;
in particular, when $Q$ has no free variables $\fhtw_\ell(Q) \geq \frac 1 2 \fhtw(Q)$.
\label{prop:faqw:1:2}
\eprop

\bp
Let $(T, \chi)$ denote a relaxed tree decomposition of $Q$ with fractional hypertree
width $\faqw_\ell(Q)$. Construct a new (non-relaxed) tree decomposition
$(T',\chi')$ for $Q$ as
follows.
Each vertex $t$ in $V(T)$ is also a vertex in $V(T')$
with $\chi'(t) = \chi(t)$.
Moreover, to each edge $\{s,t\} \in E(T)$ there corresponds an additional vertex
$st$ in $V(T')$ whose
bag
is $\chi'(st) = \chi(s) \cup \chi(t)$. As for the edge set of $T'$, for each
edge $\{s, t\} \in
E(T)$, there are two corresponding edges in $E(T')$, namely $\{s, st\}$ and
$\{t, st\}$.
We can verify
that $(T',\chi')$ is a (non-relaxed) tree decomposition of $Q$.
Moreover because each bag of $(T', \chi')$ is covered by at most two bags of $(T, \chi)$, the $\faq$-width of $(T', \chi')$ is at most
$2\cdot\faqw_\ell(Q)$.
Finally, if $(T, \chi)$ is $F$-connex, then so is $(T', \chi')$.
\ep

\subsubsection{$\faqai$ on the Boolean semiring}

Before explaining how we can adapt $\panda$ to solve an $\faqai$ query on the Boolean semiring, we give the intuition with an example.

\begin{example}
    \label{ex:4:path:ineq}
    Consider the following $\faqai$:
    \begin{align}
        Q() \leftarrow R(a,b) \wedge S(b,c) \wedge T(c,d) \wedge a \leq d.
        \label{eq:4:path:ineq}
    \end{align}
    Here $\faqw_\ell(Q) = \faqw(Q) = 2$.
    Using fractional hypertree width measure and $\InsideOut$ (even with
    relaxed
    TDs and Theorem~\ref{thm:relaxed:faqw}), the
best runtime is $O(N^2)$, because no matter which (relaxed) TD
we choose,
the worst-case bag relation size is $\Theta(N^2)$.
However the $\panda$ framework~\cite{panda-pods,panda-arxiv} can solve many queries, including this one,
in time smaller than the FAQ-width. At a very high level, the way $\panda$ achieves this is
by carefully
partitioning the input data and then choosing a possibly different tree decomposition
for each part. Query~\eqref{eq:4:path:ineq} accepts two non-redundant
and non-trivial\footnote{A tree decomposition is trivial if it consists of only one bag containing
all the variables.} relaxed tree decompositions. The first tree decomposition consists of the
bags $\{a, b, c\}$ and $\{c, d\}$ while the second has the bags $\{a, b\}$ and $\{b, c, d\}$.
The $\panda$ framework utilizes both tree decompositions simultaneously to solve this query.
In particular, for each tuple $(a, b, c, d)$ satisfying the body of query~\eqref{eq:4:path:ineq},
we make sure that this tuple is ``captured by'' at least one of the two tree decompositions
in the sense that it will reported by a query over this tree decomposition.
We realize this intuition using the following {\em disjunctive Datalog rule}:
    \begin{align}
        U(a,b,c) \vee W(b,c,d) \leftarrow R(a,b) \wedge S(b,c) \wedge T(c,d).
        \label{eqn:simple:ddlog}
    \end{align}
    In the above rule, there are two relations in the head $U$ and $W$, and they form a solution to
the rule iff the following holds: if $(a,b,c,d)$ satisfies the body, then either
$(a,b,c) \in U$ or $(b,c,d) \in W$. Via information-theoretic
inequalities~\cite{panda-pods,panda-arxiv}, we are able to show that $\panda$ can compute a
solution $(U,W)$ to the above disjunctive Datalog rule in time $\tilde O(N^{1.5})$. In particular, both $|U|$ and $|W|$ are bounded by $N^{1.5}$.

    Given such a solution $(U,W)$ to~\eqref{eqn:simple:ddlog} (which is not necessarily unique), it is
    straightforward
to verify that the following also holds, using the distributivity of $\vee$ over
$\wedge$:
    \begin{equation}
        (R(a,b) \wedge W(b,c,d)) \vee (U(a,b,c) \wedge T(c,d))
        \quad\leftarrow\quad R(a,b) \wedge S(b,c) \wedge T(c,d).
    \end{equation}
    By semijoin-reducing $W$ against $S$ and  $T$ (i.e. by replacing $W(b, c, d)$ with
    $(W(b, c, d)\Join S(b, c)) \Join T(c, d)$),
    and similarly by semjoin-reducing $U$ against
    $R$ and $S$, we conclude that
    \begin{equation*}
        (R(a,b) \wedge W(b,c,d)) \vee (U(a,b,c) \wedge T(c,d))
    \quad\equiv\quad R(a,b) \wedge S(b,c) \wedge T(c,d).
    \end{equation*}
    Finally, we have a rewrite of the original body:
    \begin{multline}
        [R(a,b) \wedge W(b,c,d) \wedge a \leq d] \vee [U(a,b,c) \wedge T(c,d) \wedge a \leq d] \\
        \equiv R(a,b) \wedge S(b,c) \wedge T(c,d) \wedge a \leq d.
        \label{eqn:simple:ddlog:final}
    \end{multline}
    The above captures precisely our intuition that every tuple $(a, b, c, d)$ satisfying
    the body of~\eqref{eq:4:path:ineq} should be reported by either one of the two relaxed
    tree decompositions.
    By defining intermediate rules, we can compute $Q$ from them:
    \begin{align}
        Q_1() &\leftarrow R(a,b) \wedge W(b,c,d) \wedge a \leq d, \\
        Q_2() &\leftarrow U(a,b,c) \wedge T(c,d) \wedge a \leq d, \\
        Q() &\leftarrow Q_1() \vee Q_2().
    \end{align}
    $Q_1$ and $Q_2$ are of the
    form~\eqref{eqn:two:bags}, and thus they each can be answered in $\tilde
    O(N^{1.5})$-time (since $|U|, |W| \leq N^{1.5}$). This implies that $Q$ can
    be answered in $\tilde
    O(N^{1.5})$-time overall.\qed
\end{example}

The strategy outlined in the above example uses $\panda$ to evaluate an $\faqai$
query over the Boolean semiring. The resulting algorithm achieves a natural generalization of the submodular $\faq$-width defined in~\eqref{eqn:subfaqw}:

\bdefn
    Given an $\faqai$ query $Q$~\eqref{eqn:our:query} over the Boolean semiring. The {\em relaxed submodular $\faq$-width} of $Q$ is defined by
\begin{align}
    \subfaqw_\ell(Q) &:=  \max_{h \in \ed_{\not\infty} \cap
    \Gamma_n}\min_{(T,\chi) \in \td^\ell_F} \max_{t\in V(T)}h(\chi(t)).
    \label{eqn:relaxed:subfaqw}
\end{align}
\edefn
(Recall that the set of relaxed tree decompositions $\td_F^\ell$ was defined in
Definition~\ref{defn:relaxed:td}.)

\bthm
Any $\faqai$ query $Q$ of the form~\eqref{eqn:our:query} on the Boolean semiring can be answered in time
$\tilde O(N^{\subfaqw_\ell(Q)}+|Q|)$.
\label{thm:subfaqw:boolean}
\ethm

\bp
As in the proof of Theorem~\ref{thm:relaxed:faqw}, we first assume there are no
free variables; the generalization to $F \neq \emptyset$ is a straightforward generalization
of techniques developed in~\cite{DBLP:conf/pods/KhamisNR16,faq-arxiv} and reviewed in
Appendix~\ref{app:insideout}.
When
$F=\emptyset$, the query~\eqref{eqn:our:query} is written in Datalog as:
\begin{align}
	Q() &\leftarrow \bigwedge_{K \in \calE_s} R_K \wedge
	\bigwedge_{S \in \calE_\ell} \left[\sum_{v \in S}\theta^S_v \leq 0 \right].
	\label{eqn:new:body}
\end{align}
We write $R_K$ instead of $R_K(\bm x_K)$ and $\theta^S_v$ instead of $\theta^S_v(x_v)$ to avoid clutter. It will be implicit throughout this proof that the subscript of a factor/function
indicates its arguments.
To answer query~\eqref{eqn:new:body}, the first step is to find
one relation $S^{(T,\chi)}_{\chi(t)}$ (over variables $\chi(t)$) for every bag $t \in V(T)$ of every relaxed tree decomposition $(T,\chi) \in \td^\ell_\emptyset$ such that the relations $S^{(T,\chi)}_{\chi(t)}$ together form a solution to the following equation:
\begin{align}
	\bigwedge_{K \in \calE_s} R_K
	&\equiv
	\bigvee_{(T,\chi) \in \td^\ell_\emptyset} \bigwedge_{t\in V(T)}
	S^{(T,\chi)}_{\chi(t)}.
	\label{eqn:equiv:form}
\end{align}
Note that the right-hand side of~\eqref{eqn:equiv:form} is a
{\em Boolean tensor decomposition} of
the left-hand side:
In particular under the Boolean semiring
$(\bm D, \oplus, \otimes, \bm 0, \bm 1) = (\{\true,\false\}, \vee, \wedge, \false,\true)$,
the left-hand side of~\eqref{eqn:equiv:form} can be viewed as an $n$-dimensional tensor
where $n = |\calV| = |\cup_{K \in \calE_s} K|$ while the right-hand side is an equivalent sum of a product of tensors.
The idea of using Boolean tensor decomposition to speed up query
evaluation was used in the context of queries
with {\em dis}equalities~\cite{DBLP:journals/corr/abs-1712-07445}.
Assuming that we can compute the intermediate relations
$S^{(T,\chi)}_{\chi(t)}$ efficiently satisfying~\eqref{eqn:equiv:form},
then~\eqref{eqn:new:body} can be answered by answering
for each $(T,\chi) \in \td_\emptyset^\ell$ an intermediate query:
\begin{align}
	Q^{(T,\chi)}() &\leftarrow \bigwedge_{t\in V(T)} S^{(T,\chi)}_{\chi(t)}
	\wedge
	\bigwedge_{S \in \calE_\ell} \left[\sum_{v \in S}\theta^S_v \leq 0 \right].
	\label{eqn:intermediate}
\end{align}
The final answer $Q$ is obtained by the Datalog rule:
\begin{align}
	Q() &\leftarrow
    \bigvee_{(T,\chi) \in \td^\ell_\emptyset} Q^{(T,\chi)}().
    \label{eqn:intermediate:union}
\end{align}
The key point here is that each intermediate query~\eqref{eqn:intermediate} is
an $\faqai$ query~\eqref{eqn:our:query} with $\faqw_\ell \leq 1$.
\footnote{We can also show here that $\faqw_\ell$ is exactly 1 although this is not needed for the proof
of Theorem~\ref{thm:subfaqw:boolean}. In particular, by comparing~\eqref{eqn:relaxed:faqw}
to~\eqref{eqn:fhtw}, we can see that for any query $Q$,
$\faqw_\ell(Q) \geq \fhtw(\calH_{\not\infty} := (\calV, \calE_{\not\infty}))$,
and $\fhtw$ for any hypergraph is at least
$1$~\cite{DBLP:journals/talg/GroheM14}.}
This is because $Q^{(T,\chi)}()$ admits a relaxed tree decomposition $(T,\chi)$ where
each bag $\chi(t)$ for $t \in V(T)$ is covered by one relation
$S^{(T,\chi)}_{\chi(t)} \in \calE_{\not\infty}$, hence
$\faqw_\ell(Q^{(T,\chi)}()) \leq
\max_{h \in \ed_{\not\infty} \cap \Gamma_n}
\max_{t\in V(T)}h(\chi(t)) \leq 1$.
By Theorem~\ref{thm:relaxed:faqw} each intermediate query~\eqref{eqn:intermediate} can be answered in time
$\tilde O(M)$ where
\begin{align}
	M &= \max_{(T,\chi) \in \td^\ell_\emptyset} \max_{t\in V(T)} |S^{(T,\chi)}_{\chi(t)}|.
	\label{eqn:M}
\end{align}

It remains to show how to compute tables $S^{(T,\chi)}_{\chi(t)}$ that form a solution to~\eqref{eqn:equiv:form}; to do so, we apply
distributivity of $\vee$ over $\wedge$ to rewrite the right-hand side of~\eqref{eqn:equiv:form} as follows.
Let $\calM$ be the collection of {\em all} maps $\beta : \td^\ell_\emptyset \to 2^\calV$
such that $\beta(T,\chi) = \chi(t)$ for some $t \in V(T)$; in other words, $\beta$ selects
one bag $\chi(t)$ out of each tree decomposition $(T,\chi)$. Then, from the distributive
law we have
\begin{align}
	\bigvee_{(T,\chi) \in \td^\ell_\emptyset} \bigwedge_{t\in V(T)}
	S^{(T,\chi)}_{\chi(t)}
	&\equiv
	\bigwedge_{\beta \in \calM}
	\bigvee_{(T,\chi) \in \td^\ell_\emptyset} S^{(T,\chi)}_{\beta(T,\chi)},
	\label{eqn:equiv:form:2}
\end{align}
which means to solve the relational equation~\eqref{eqn:equiv:form} we can instead
solve the equation
\begin{align}
	\bigwedge_{\beta \in \calM}
	\bigvee_{(T,\chi) \in \td^\ell_\emptyset} S^{(T,\chi)}_{\beta(T,\chi)}
	&\equiv \bigwedge_{K \in \calE_s} R_K.
	\label{eqn:equiv:form:3}
\end{align}
To solve the above equation, for each $\beta \in \calM$ we can find tables $S^{(T,\chi)}_{\beta(T,\chi)}$ that form a solution to the following equation
\begin{align}
	\bigvee_{(T,\chi) \in \td^\ell_\emptyset} S^{(T,\chi)}_{\beta(T,\chi)}
	&\equiv \bigwedge_{K \in \calE_s} R_K.
	\label{eqn:equiv:form:3.5}
\end{align}
To do that, for each $\beta \in \calM$, we compute a solution to
the following disjunctive Datalog rule:
\begin{align}
	\bigvee_{(T,\chi) \in \td^\ell_\emptyset}
	W^{(T,\chi)}_{\beta(T,\chi)}
	&\leftarrow \bigwedge_{K \in \calE_s} R_K.
	\label{eqn:equiv:form:4}
\end{align}
Once we obtain the relations $W^{(T,\chi)}_{\beta(T,\chi)}$, we can semijoin-reduce them
against the input relations $R_K$ (i.e. replace $W^{(T,\chi)}_{\beta(T,\chi)}$ with
$W^{(T,\chi)}_{\beta(T,\chi)}\Join R_K$ for each input relation $R_K$ where $K \subseteq \beta(T,\chi)$),
in order to obtain $S^{(T,\chi)}_{\beta(T,\chi)}$
that solve~\eqref{eqn:equiv:form:3}.
Once we obtain those $S^{(T,\chi)}_{\beta(T,\chi)}$, we plug them in~\eqref{eqn:intermediate}
to obtain an $\faqai$ query of the form~\eqref{eqn:our:query} for each relaxed tree decomposition
$(T, \chi) \in \td^\ell_\emptyset$.
We use Theorem~\ref{thm:relaxed:faqw} to solve each one of those queries in time $\tilde O(M)$ where $M$
was given by~\eqref{eqn:M}. This is the step of the algorithm where the additive inequalities
$\bigwedge_{S \in \calE_\ell} \left[\sum_{v \in S}\theta^S_v \leq 0 \right]$
participate in the computation.
Once we obtain the solutions to queries~\eqref{eqn:intermediate},
we use \eqref{eqn:intermediate:union} to obtain the answer of the original $\faqai$ query.

The only step in the above algorithm that we haven't specified yet is how to evaluate each
disjunctive Datalog rule~\eqref{eqn:equiv:form:4}. We do so by running
the $\panda$ algorithm, which computes the rule in time bounded by $\tilde O(N^{e(\beta)})$, where
\begin{align}
	e(\beta) &=
	\max_{h \in \ed_{\not\infty} \cap \Gamma_n}
	\min_{(T,\chi) \in \td^\ell_\emptyset}
	h(\beta(T,\chi)).
\end{align}
Maximizing over $\beta \in \calM$, the runtime is bounded by $\tilde O(N^w)$, where
\begin{align}
	w &=
	\max_{\beta \in \calM} e(\beta) \\
	&=
	\max_{\beta \in \calM}
	\max_{h \in \ed_{\not\infty} \cap \Gamma_n}
	\min_{(T,\chi) \in \td^\ell_\emptyset}
	h(\beta(T,\chi)) \\
	&=
	\max_{h \in \ed_{\not\infty} \cap \Gamma_n}
	\max_{\beta \in \calM}
	\min_{(T,\chi) \in \td^\ell_\emptyset}
	h(\beta(T,\chi)) \\
	&=
	\max_{h \in \ed_{\not\infty} \cap \Gamma_n}
	\min_{(T,\chi) \in \td^\ell_\emptyset}
	\max_{t \in V(T)}
	h(\chi(t)) = \subfaqw_\ell(Q).\label{eq:next-to-last}
\end{align}
The first equality in~\eqref{eq:next-to-last} follows from the minimax lemma in~\cite{panda-arxiv}.
Our reasoning above also shows that $M$ from~\eqref{eqn:M} is bounded
by $N^{\subfaqw_\ell(Q)}$.
\ep

\subsection{Relaxed polymatroids}
\label{subsec:relaxed-poly}

A key step in the proof of Theorem~\ref{thm:subfaqw:boolean} is to find the Boolean tensor
decomposition~\eqref{eqn:equiv:form} of the product over $R_K$. In a non-Boolean semiring, this becomes a tensor decomposition on this semiring:
\begin{align}
     \bigotimes_{K \in \calE_s} R_K
     &=
    \bigoplus_{(T,\chi) \in \td^\ell_F} \bigotimes_{t\in V(T)} S^{(T,\chi)}_{\chi(t)}.
     \label{eqn:equiv:form:general}
\end{align}
In order to compute this tensor decomposition, we can still follow the script of the proof
of Theorem~\ref{thm:subfaqw:boolean}, working on the parameter space of the input factors
$R_K$; however, for the equality in~\eqref{eqn:equiv:form:general}
to hold (it is an identity over the value-space of the factors), it suffices to ensure the
following property:

\begin{quote}
    For any $\bm x_{\calV}$ s.t. $\bigotimes_{K\in \calE_s} R_K(\bm x_K) \neq \bm 0$,
    there is {\em exactly one} tree
    decomposition $(T,\chi) \in \td^\ell_F$ for which
        \begin{align}
    \bigotimes_{t\in V(T)} S^{(T,\chi)}_{\chi(t)}(\bm x_{\chi(t)}) =
    \bigotimes_{K\in \calE_s} R_K(\bm x_K),
    \label{eqn:desired}
\end{align}
while for the other TDs, the left-hand side \mbox{above is $\bm 0$}.
\end{quote}
Essentially, the property ensures that we do not have to perform inclusion-exclusion (IE) over
the tree decompositions in $\td_F^\ell$.\footnote{IE is difficult for two reasons: (1) IE
computation explodes the runtime, and (2) in a general semiring there may not be additive
inverses and thus IE may not even apply.}
We do not know how to ensure this property in general. However, under a relaxed notion
of polymatroids, the property above holds.
Since this idea applies to $\faq$ queries in general, we start with our result on
$\faq$ queries first, before specializing it to $\faqai$.

\subsubsection{$\faq$ over an arbitrary semiring}

To explain how we can guarantee the property~\eqref{eqn:desired} for an $\faq$
query over an arbitrary semiring, consider the following example. Suppose that
we would like to evaluate the (aggregate) query
\begin{align}
    Q() &= \bigoplus_{\bm x_{[4]}} R_{12}(x_1,x_2) \otimes R_{23}(x_2,x_3) \otimes R_{34}(x_3,x_4) \otimes R_{41}(x_4,x_1).
    \label{eqn:4:cycle}
\end{align}
We write $R_{ij}$ instead of $R_{ij}(x_i,x_j)$ for short. The factors $R_{ij}$
are functions of two variables $R_{ij} : \dom(X_i) \times \dom(X_j) \to \R$, and they
are represented by {\em ternary} relations in a database. Abusing notation we will also use $R_{ij}$ to refer to its support, i.e., the binary relation over $(X_i,X_j)$ such that $(x_i, x_j) \in R_{ij}$ iff $R_{ij}(x_i,x_j) \neq \bm 0$.

There are only two non-trivial tree decompositions for the
``$4$-cycle'' query~\eqref{eqn:4:cycle}: one with bags $\{1,2,3\}$ and $\{3,4,1\}$,
and the other with bags $\{1,2,4\}$ and $\{2,3,4\}$.\footnote{The trivial TD with one bag $\{1, 2, 3, 4\}$ can always be replaced by a non-trivial TD in the considered bounds/algorithms without making them any worse. Similarly, redundant TDs can be replaced by non-redundant ones.}
To evaluate the query,
we first solve the relational equation~\eqref{eqn:equiv:form:general}, but only on the
supports; i.e., we would like to find relations $S_{123}, S_{341}$, $S_{234}$, and $S_{412}$
such that
\begin{align}
        & R_{12} \wedge R_{23} \wedge R_{34} \wedge R_{41} \equiv (S_{123}
        \wedge S_{341}) \vee (S_{234} \wedge S_{412})
        \label{eqn:4:cycle:rel} \equiv\\
        & (S_{123} \vee S_{234}) \wedge
                (S_{123} \vee S_{412})\ \wedge
                (S_{341} \vee S_{234}) \wedge
                (S_{341} \vee S_{412}).\nonumber
\end{align}
The second $\equiv$ is due to the distributivity of $\vee$ over $\wedge$. Since the last
formula is in CNF, we can solve each term separately by solving $4$ different disjunctive
Datalog rules:
\begin{align}
    (S_{123} \vee S_{234}) &\leftarrow R_{12} \wedge R_{23} \wedge R_{34}
    \wedge R_{41}, \label{eqn:4:cycle:rel:1}\\
    (S_{123} \vee S_{412}) &\leftarrow R_{12} \wedge R_{23} \wedge R_{34}
    \wedge R_{41}, \label{eqn:4:cycle:rel:2}\\
    (S_{341} \vee S_{234}) &\leftarrow R_{12} \wedge R_{23} \wedge R_{34}
    \wedge R_{41}, \label{eqn:4:cycle:rel:3}\\
    (S_{341} \vee S_{412}) &\leftarrow R_{12} \wedge R_{23} \wedge R_{34}
    \wedge R_{41}. \label{eqn:4:cycle:rel:4}
\end{align}
Applying the proof-to-algorithm conversion idea from $\panda$~\cite{panda-pods,panda-arxiv}, the above disjunctive
Datalog rules can be solved with the $\panda$ algorithm. It is beyond the scope of this article to describe the $\panda$ algorithm in full details. However, we can
describe a solution.
Let $N = \max\{|R_{12}|,|R_{23}|,$ $|R_{34}|,|R_{41}|\}$. For each input relation/factor,
define their ``light'' parts as follows.
\begin{align}
    R_{ij}^\ell &:= \{ (x_i,x_j) \in R_{ij} \ : \ |\sigma_{X_i=x_i}R_{ij}|\leq
    \sqrt N\}.
\end{align}
Also, for every $R_{ij}$, define $R^h_{ij} := R_{ij} \setminus R^\ell_{ij}$. Then, one
can verify that the following is a solution to the relational
equations \eqref{eqn:4:cycle:rel:1}-\eqref{eqn:4:cycle:rel:4}:
\begin{align*}
    S_{ijk} &= R_{ij} \Join R^l_{jk} \cup \pi_i R^h_{ij} \Join R_{jk}.
\end{align*}
The above is not yet solution to \eqref{eqn:4:cycle:rel}.
However we can refine it as follows to obtain such a solution:
\begin{align}
    S_{ijk} &= (R_{ij} \Join R^l_{jk} \cup \pi_i R^h_{ij} \Join R_{jk})\Join R_{ij} \Join R_{jk}\nonumber\\
    &=R_{ij} \Join R^l_{jk} \cup R^h_{ij} \Join R_{jk}.\label{eqn:S_ijk-rel}
\end{align}
(These extra relations $R_{ij} \Join R_{jk}$ that are joined into $S_{ijk}$ to turn them
into a solution to~\eqref{eqn:4:cycle:rel} will be referred to as ``filters''
in the proof of Theorem~\ref{thm:subfaqw:general} below.)
It is straightforward to verify that each $S_{ijk}$ can be computed in
$\tilde O(N^{1.5})$-time.
However, \eqref{eqn:4:cycle:rel} alone is not enough to guarantee \eqref{eqn:desired}. Instead, we now need $S_{ijk}$ to satisfy the following stronger condition (where $\myxor$ denotes the {\em exclusive} OR):
\begin{align}
    & R_{12} \wedge R_{23} \wedge R_{34} \wedge R_{41} \equiv (S_{123}
    \wedge S_{341}) \myxor (S_{234} \wedge S_{412})
    \equiv\label{eqn:S_ijk-rel:xor}\\
    & (S_{123} \myxor S_{234}) \wedge
    (S_{123} \myxor S_{412})\ \wedge
    (S_{341} \myxor S_{234}) \wedge
    (S_{341} \myxor S_{412}).\nonumber
\end{align}
Luckily, in this particular example, our previous solution for $S_{ijk}$ from \eqref{eqn:S_ijk-rel} happens to be a solution to \eqref{eqn:S_ijk-rel:xor} as well.
Once we have the relations $S_{ijk}$ from \eqref{eqn:S_ijk-rel}, we can extend them naturally into factors (so that
they are represented by $4$-ary relations) satisfying~\eqref{eqn:desired}. In particular,
as functions with range $\R$, they are defined by
\begin{align}
    S_{ijk}(x_i,x_j,x_k) &:= R_{ij}(x_i,x_j) \otimes R_{jk}(x_j,x_k).
    \label{eqn:S_ijk-fun}
\end{align}
Finally the query $Q$ from \eqref{eqn:4:cycle}
can be computed by taking the sum of two queries:
\begin{eqnarray}
Q() = \left[\bigoplus_{\bm x_{[4]}} S_{123}(x_1, x_2, x_3)\otimes S_{341}(x_3, x_4, x_1)\right] \bigoplus \left[\bigoplus_{\bm x_{[4]}} S_{234}(x_2, x_3, x_4)\otimes S_{412}(x_4, x_1, x_2)\right]
\end{eqnarray}

The above sketch does not work for a general $\faq$ query because the relational solution returned by $\panda$ is not guaranteed to satisfy~\eqref{eqn:desired}.
(If we could do that, then we would have been able to solve $\#\text{\sf CSP}$ queries in submodular width time, but the latter is unlikely to be possible since the submodular width tightly characterizes the hardness of $\text{\sf CSP}$ queries~\cite{MR3144912}.)
We could however restrict $\panda$ forcing it to maintain~\eqref{eqn:desired}
at the cost of weakening the runtime bound achieved by $\panda$.
In particular, $\panda$'s runtime is upperbounded by the submodular ($\faq$) width, which is a maximum over some set of polymatroids (See Section~\ref{subsec:insideout:panda}).
We will now replace these polymatroids with a superset, called {\em $\calE$-polymatroids}, leading to a larger version of the submodular ($\faq$) width called
``sharp submodular ($\faq$) width''.
The latter captures the runtime of our new version of $\panda$, called $\spanda$.

\bdefn[$\calE$-polymatroids and $\Gamma_{n|\calE}$]
Given a collection $\calE$ of subsets of $\calV$, a set function $h: 2^\calV
\to \R_+$ is said to be a {\em $\calE$-polymatroid} if it satisfies the following:
(i) $h(\emptyset) = 0$,
(ii) $h(X) \leq h(Y)$ whenever $X \subseteq Y$, and
(iii) $h(X\cup Y)+h(X\cap Y)\leq h(X)+h(Y)$ for every pair $X, Y \subseteq \calV$
\underline{\bf such that $X \cap Y \subseteq S$ for some $S \in \calE$}.\footnote{The \underline{\bf underlined} part is the only distinction between $\calE$-polymatroids and polymatroids. If we drop it, we get back the original definition of polymatroids.}
In particular, a $2^\calV$-polymatroid is a polymatroid as defined in
Section~\ref{subsec:width:params}. For $\calV=[n]$, let $\Gamma_{n|\calE}$
denote the set
of all $\calE$-polymatroids on $\calV$.
\label{defn:calE-poly}
\edefn

The following definition is a straightforward generalization of $\subfaqw$ from~\eqref{eqn:subfaqw},
where we replace $\Gamma_n$ by the relaxed polymatroids $\Gamma_{n|\calE}$.

\bdefn[\#-submodular $\faq$-width]
Given an $\faq$ query~\eqref{eqn:faq} whose hypergraph is
$\calH = (\calV, \calE = \calE_{\not\infty} \cup \calE_{\infty})$, its
{\em \#-submodular $\faq$-width}, denoted by $\ssubfaqw(Q)$, is defined by
\begin{align}
\ssubfaqw(Q) &:=
\max_{h \in \ed_{\not\infty} \cap \Gamma_{n|\calE_{\not\infty}}}
\min_{(T,\chi) \in \td_F}
\max_{t\in V(T)} h(\chi(t)).
\label{eqn:ssubfaqw}
\end{align}
When there are no free variables, i.e., $F=\emptyset$, we define $\ssubw(Q) := \ssubfaqw(Q)$, to mirror the case when $\faqw(Q) = \fhtw(Q)$.
\edefn

Under the above new width parameter,
we can now maintain condition~\eqref{eqn:desired} allowing us to solve $\faq$ queries over any semiring:

\bthm
Any $\faq$ query $Q$ of the form~\eqref{eqn:faq} on any semiring can be
answered in time
$\tilde O(N^{\ssubfaqw(Q)}+|Q|)$.
\label{thm:subfaqw:general}
\ethm

The proof of Theorem~\ref{thm:subfaqw:general} involves
an appropriate adaptation of $\panda$ called $\spanda$, to be described below.
Appendix~\ref{app:panda} presents an overview of the original $\panda$ algorithm. Readers unfamiliar
with $\panda$ are recommended to read that appendix first before reading the following proof.

\bp
The $\panda$ algorithm~\cite{panda-pods,panda-arxiv} takes as input a disjunctive
Datalog query of
the form
\begin{equation}
	\bigvee_{B \in \calB}G_B(\bm x_B) \leftarrow \bigwedge_{K\in\calE} R_K(\bm x_K).
	\label{eqn:disjunctive:datalog:query}
\end{equation}
The above query has an input relation $R_K$ for each hyperedge $K
\in \calE$ in the query's hypergraph $\calH=(\calV, \calE)$.
The output to the above query is a collection of tables $G_B$, one for each
``goal'' (or ``target'') $B$ in the
collection of goals $\calB$.
The output tables $(G_B)_{B\in\calB}$ must satisfy the logical implication
in~\eqref{eqn:disjunctive:datalog:query}:
In particular, for each tuple $\bm x_\calV$ that satisfies the conjunction
$\bigwedge_{K\in\calE} R_K(\bm x_K)$,
the disjunction $\bigvee_{B \in \calB}G_B(\bm x_B)$ must hold.
Query~\eqref{eqn:simple:ddlog} is an example of
\eqref{eqn:disjunctive:datalog:query}.
A disjunctive Datalog query~\eqref{eqn:disjunctive:datalog:query} can have many
valid outputs.
The $\panda$ algorithm computes one such output in time
$\tilde O(N^{e})$, where
\begin{equation}
	e =
	\max_{h\in\ed_{\not\infty}\cap\Gamma_n}\min_{B\in\calB} h(B).
	\label{eq:panda-runtime}
\end{equation}
(Recall notation from Section~\ref{subsec:insideout:panda}.)

In what follows, we describe a variant of $\panda$, called $\spanda$, that takes
a disjunctive Datalog query~\eqref{eqn:disjunctive:datalog:query}, and computes
the following:
\bi
\item A collection of tables $(G_B)_{B\in\calB}$ that form a valid
output to query~\eqref{eqn:disjunctive:datalog:query}, i.e. that satisfy the
logical implication in~\eqref{eqn:disjunctive:datalog:query}.
\item Moreover, associated with each output table $G_B$, $\spanda$ additionally
computes a collection of ``filter'' tables $\left(F_K^{(B)}\right)_{K \in
	\calE}$, one table $F_K^{(B)}$ for each hyperedge $K \in \calE$ in the input
hypergraph $\calH$. The output tables $G_B$ along with the associated filters
$\left(F_K^{(B)}\right)_{K \in \calE}$ satisfy the following condition:
For each tuple $\bm x_\calV$ that satisfies the conjunction
$\bigwedge_{K\in\calE} R_K(\bm x_K)$, there is {\em exactly one} target
$B\in\calB$
where the conjunction $\displaystyle{\bigwedge_{K \in
		\calE}
	F_K^{(B)}(\bm x_K)}$ holds,
and for that target $B$, $G_B(\bm x_B)$ holds as well.
In particular, the following equivalences hold:
\begin{eqnarray}
	\mybigxor{B \in \calB}\left[
	\displaystyle{\bigwedge_{K
			\in
			\calE}
		F_K^{(B)}(\bm x_K)}\right]
	\quad&\equiv&\quad\bigwedge_{K\in\calE} R_K(\bm x_K),\label{eq:spanda:invar1}\\
	\left[\displaystyle{\bigwedge_{K
			\in
			\calE}
		F_K^{(B)}(\bm x_K)}\right]
	&\equiv&
	\left[G_B(\bm x_B) \wedge
	\displaystyle{\bigwedge_{K
			\in
			\calE}
		F_K^{(B)}(\bm x_K)}\right],
	\quad\forall B\in\calB,\label{eq:spanda:invar2}
\end{eqnarray}
where $\myxor$ above denotes the exclusive OR.
Equations~\eqref{eq:spanda:invar1} and \eqref{eq:spanda:invar2} together imply
\begin{eqnarray}
	\mybigxor{B \in \calB}\left[G_B(\bm x_B) \wedge
\displaystyle{\bigwedge_{K
        \in
        \calE}
    F_K^{(B)}(\bm x_K)}\right]
\quad&\equiv&\quad\bigwedge_{K\in\calE} R_K(\bm x_K),\label{eq:spanda:invar3}
\end{eqnarray}
Comparing the above to \eqref{eqn:disjunctive:datalog:query}, note that the purpose of the filters $F_K^{(B)}$ is to keep the goals $G_B(\bm x_B)$ disjoint from one another allowing us to replace $\vee$ with $\myxor$ and ultimately maintain condition~\eqref{eqn:equiv:form:general}.
(As we will see later, in $\spanda$, we start with filters $F_K^{(B)}$ that are identical to the corresponding input relations $R_K$, and we keep removing tuples from $F_K^{(B)}$ to maintain \eqref{eq:spanda:invar1} and \eqref{eq:spanda:invar2} throughout the algorithm.)
\ei
$\spanda$ computes the above output tables $(G_B)_{B\in\calB}$ and $\left(\left(F_K^{(B)}\right)_{K \in
	\calE}\right)_{B\in\calB}$ in time $\tilde O(N^{e'})$ where
\begin{equation}
	e' =
	\max_{h\in\ed_{\not\infty}\cap\Gamma_{n|\calE_{\not\infty}}}\min_{B\in\calB}
	h(B).
	\label{eq:spanda-runtime}
\end{equation}

Now we briefly explain how to tweak the $\panda$ algorithm into $\spanda$
satisfying the above characteristics. We refer the reader to Appendix~\ref{app:panda}
and~\cite{panda-pods,panda-arxiv}
for more details about $\panda$.
At a high level, the $\panda$ algorithm starts with proving an exact upperbound
on $e$ from~\eqref{eq:panda-runtime} using a sequence of proof steps, called
the {\em proof sequence}
(see Lemmas~\ref{lmm:panda:LP-form},~\ref{lmm:panda:LP-to-inequality}, and~\ref{lmm:panda:ps}).
Then $\panda$ interprets each step in the proof sequence as a relational
operator, and then uses
this sequence of relational operators as a query plan to actually compute the
query in
time $\tilde
O(N^e)$.
One of the proof steps used in $\panda$ is {\em the decomposition step}
$h(Y) \rightarrow h(X) + h(Y|X)$ for some $X\subseteq Y\subseteq \calV$.
The relational operator corresponding to this decomposition step is the
``partitioning'' operator, in which
we take an input (or intermediate) table $R_Y$ and partition it
into a small number $k = O(\log |R_Y|)$ of tables $R_Y^{(1)}, \ldots, R_Y^{(k)}$,
based on the degrees of variables in $Y$ with respect to variables in
$X\subseteq Y$.
In particular, define the degree of $Y$ w.r.t. a tuple $\bm t_X \in \pi_X R_Y$
and w.r.t. to $X$
as follows:
\begin{eqnarray}
	\deg_{R_Y}(Y|\bm t_X) &:=& \left|\left\{\bm t'_Y \in R_Y\suchthat \bm t'_X
	= \bm t_X\right\}\right|,\\
    \deg_{R_Y}(Y|X) &:=& \max_{\bm t_X \in \pi_X R_Y}\deg_{R_Y}(Y|\bm t_X).
\end{eqnarray}
In the partitioning step, we partition tuples $\bm t_X \in \pi_X R_Y$ into $k$
buckets based on $\deg_{R_Y}(Y|\bm t_X)$ and partition $R_Y$ accordingly.
Specifically, for each $j \in [k]$, we define
\begin{eqnarray}
R_X^{(j)} &:=& \left\{\bm t_X \in \pi_X R_Y \suchthat 2^{j-1} \leq \deg_{R_Y}(Y|\bm t_X) < 2^j\right\},\\
R_Y^{(j)} &:=& \left\{\bm t_Y \in R_Y \suchthat \bm t_X \in R_X^{(j)} \right\}.
\end{eqnarray}
After partitioning, $\panda$ creates $k$ independent branches of the problem,
where in the $j$-th branch, $R_Y$ is replaced by both $R_X^{(j)}$ and $R_Y^{(j)}$.
Note that for each $j \in [k]$, the following holds:
\begin{equation}
\log_2 |R_Y| \quad\geq\quad \log_2 |R_X^{(j)}| \quad+\quad \log_2 \left[\deg_{R_Y^{(j)}}(Y|X)\right] - 1
\end{equation}
The above inequality mirrors the proof step $h(Y) \rightarrow h(X) + h(Y|X)$
exemplifying the way the entire $\panda$ algorithm mirrors the proof sequence
of the bound in~\eqref{eq:panda-runtime} allowing its runtime to be bounded by~\eqref{eq:panda-runtime} (see~\cite{panda-pods,panda-arxiv} for more details).
After each partitioning step, $\panda$ continues on each one of the $k$ branches of the problem independently and ends up computing a potentially different target
$G_{B}$ for some $B\in \calB$ within each branch.

From the proof sequence construction described in~\cite{panda-pods,panda-arxiv},
we note the following:
If the constructed proof sequence that is used to prove the bound on $e$
in \eqref{eq:panda-runtime} contains a decomposition step
$h(Y) \rightarrow h(X) + h(Y|X)$,
then the proof of the bound on $e$ must have relied on some submodularity
constraint on
$h$ of the
form $h(X) + h(Y \cup Z) \leq h(Y) + h(X \cup Z)$ for some $Z\subseteq\calV$
where $Y \cap Z = \emptyset$.
In particular, such a submodularity can be broken down into the sum of two inequalities:
\begin{eqnarray}
h(Y) &=& h(X) + h(Y|X)\\
h(Y|X) &\geq& h(Y\cup Z | X \cup Z)
\end{eqnarray}
which in turn are converted into two proof steps in the proof sequence:
\begin{eqnarray}
h(Y) &\rightarrow& h(X) + h(Y|X)\label{eq:decomp-submod}\\
h(Y|X) &\rightarrow& h(Y\cup Z | X \cup Z)
\end{eqnarray}
Moreover, the above is the only place in the proof sequence construction~\cite{panda-pods,panda-arxiv} where a decomposition step~\eqref{eq:decomp-submod}
is introduced.
However, the new bound~\eqref{eq:spanda-runtime} used in $\spanda$ only
relies on submodularities $h(X) + h(Y \cup Z) \leq h(Y) + h(X \cup Z)$
where $X \subseteq K$ for some $K \in \calE$. (Recall
$\Gamma_{n|\calE_{\not\infty}}$ from
Definition~\ref{defn:calE-poly}.)
Therefore, in $\spanda$, whenever we apply a partitioning step of $R_Y$ into
$R_Y^{(1)},\ldots, R_Y^{(k)}$ based on the degrees $\deg_{R_Y}(Y|\bm t_X)$ of
$\bm t_X \in \pi_X R_Y$,
we know that there is some input relation $R_K$ with $X \subseteq K$.
Therefore we can refine the corresponding filter $F_K^{(B)}$
by semijoining it with $R_X^{(j)}=\pi_X R_Y^{(j)}$  on the $j$-th branch, i.e. by taking $F_K^{(B)}
\leftarrow F_K^{(B)} \ltimes R_X^{(j)}$.
Moreover, this update of filters $F_K^{(B)}$
maintains~\eqref{eq:spanda:invar1} and~\eqref{eq:spanda:invar2}.
(Initially, we start with filters $F_K^{(B)}$ that are identical to the corresponding input
relations $R_K$, which trivially satisfy both~\eqref{eq:spanda:invar1} and~\eqref{eq:spanda:invar2}.)

Now that we have described the $\spanda$ algorithm satisfying the above
properties,
we explain how to use it as a blackbox to solve an $\faq$ query $Q$
of the form~\eqref{eqn:faq} in time
$\tilde O(N^{\ssubfaqw(Q)}+|Q|)$.
Following the same notation as in the proof of
Theorem~\ref{thm:subfaqw:boolean},
let $\calM$ be the collection of {\em all} maps $\beta : \td^\ell_F \to
2^\calV$
such that $\beta(T,\chi) = \chi(t)$ for some $t \in V(T)$; in other words,
$\beta$ selects
one bag $\chi(t)$ out of each tree decomposition $(T,\chi)$.
Let $\bB$ be the collection of images of all $\beta \in \calM$, i.e.
\begin{equation}
    \bB = \{ \calB \suchthat \calB = \image(\beta)
    \text{ for some } \beta \in \calM \}.
    \label{eqn:defn:bB}
\end{equation}
For each $\calB \in \bB$, we use $\spanda$ to solve the following rule (i.e. to produce relations $G_B$ and $F_K^{(B)}$ that satisfy the equivalence):
\begin{equation}
	\mybigxor{B\in\calB}\left[G_{B}
	\wedge
	\displaystyle{\bigwedge_{K
			\in
			\calE}
		F_K^{(B)}}\right]
	\quad\equiv\quad\bigwedge_{K\in\calE} R_K.
\end{equation}
The solutions collectively satisfy the following:
\[
\bigwedge_{\calB\in\bB}\mybigxor{B\in\calB}\left[G_{B}
\wedge
\displaystyle{\bigwedge_{K
        \in
        \calE}
    F_K^{(B)}}\right]
\quad\equiv\quad\bigwedge_{K\in\calE} R_K.
\]
Let $M=|\bB|$ and suppose $\bB =\{\calB_1, \calB_2, \ldots, \calB_M\}$.
By distributing the conjunction $\bigwedge_{\calB\in\bB}$ over $\myxor$, we get
\[\mybigxor{(B_1, \ldots, B_M)\in\prod_{i=1}^M\calB_i}\quad
\bigwedge_{j\in[M]}\left[G_{B_j}
\wedge
\displaystyle{\bigwedge_{K
        \in
        \calE}
    F_K^{(B_j)}}\right]
\quad\equiv\quad\bigwedge_{K\in\calE} R_K.\]
Using the same diagonalization argument from~\cite{panda-pods,panda-arxiv}, we can prove
the following claim:
\begin{claim}
For every $(B_1, \ldots, B_M)\in\prod_{i=1}^M\calB_i$, there must
exist a tree decomposition $(\bar T, \bar \chi)\in\td^\ell_F$ such that for every $t\in V(\bar T)$, $\bar\chi(t) = B_j$ for some $j \in [M]$.
\label{clm:diag}
\end{claim}
Assuming Claim~\ref{clm:diag} is correct, and thanks to~\eqref{eq:spanda:invar2}, we can rewrite the conjunction as
\begin{equation}\bigwedge_{j\in[M]}\left[G_{B_j}
\wedge
\displaystyle{\bigwedge_{K
        \in
        \calE}
    F_K^{(B_j)}}\right]
\quad\equiv\quad
\left[\bigwedge_{t \in V(\bar T)}G_{\bar\chi(t)}\right]
\wedge
\left[\bigwedge_{j\in[M]}\bigwedge_{K \in \calE} F_K^{(B_j)}\right].
\label{eq:drop-Gs}
\end{equation}
The right-hand side of~\eqref{eq:drop-Gs} is an $\faq$ query.
We solve it by running $\InsideOut$ over the tree
decomposition $(\bar T, \bar\chi)$.
We repeat the above for every $(B_1, \ldots, B_M)\in\prod_{i=1}^M\calB_i$.
Afterwards, because we have an {\em
	exclusive}
OR over $(B_1, \ldots, B_M)\in\prod_{i=1}^M\calB_i$, we can simply sum up corresponding query results.

From~\eqref{eq:spanda-runtime}, the total runtime is $\tilde O(N^w + |Q|)$,
where
\begin{align*}
	w &=
	\max_{\beta \in \calM}
	\max_{h \in \ed_{\not\infty} \cap \Gamma_{n|\calE_{\not\infty}}}
	\min_{(T,\chi) \in \td_F}
	h(\beta(T,\chi)) \\
	&=
	\max_{h \in \ed_{\not\infty} \cap \Gamma_{n|\calE_{\not\infty}}}
	\max_{\beta \in \calM}
	\min_{(T,\chi) \in \td_F}
	h(\beta(T,\chi)) \\
	&=
	\max_{h \in \ed_{\not\infty} \cap \Gamma_{n|\calE_{\not\infty}}}
	\min_{(T,\chi) \in \td_F}
	\max_{t \in V(T)}
	h(\chi(t))
	= \ssubfaqw(Q).
\end{align*}
Finally we include the proof of Claim~\ref{clm:diag} for completeness, following the corresponding proof in~\cite{panda-pods,panda-arxiv}.
Consider a fixed $(B_1, \ldots, B_M)\in\prod_{i=1}^M\calB_i$.
Assume to the contrary that for every tree decomposition $(T, \chi)\in\td^\ell_F$, there is some bag $\bar\beta(T, \chi) := \chi(t)$ for some $t \in V(T)$ such that
$\chi(t) \notin \{B_1, \ldots, B_M\}$. By definition of $\bB$, $\image(\bar\beta)=\calB_j$ for some $j \in [M]$. Therefore, $B_j = \bar\beta(T, \chi)$ for some $(T, \chi)\in\td^\ell_F$.
But this contradicts the claim that for every tree decomposition $(T, \chi)\in\td^\ell_F$, $\bar\beta(T, \chi)\notin\{B_1, \ldots, B_M\}$.
\ep

The following proposition shows
that while $\ssubfaqw(Q)$ can be larger than $\subfaqw(Q)$, it is not larger
than
$\faqw(Q)$ and can be unboundedly smaller for classes of queries.

\bprop[Connecting $\ssubfaqw$ to $\subfaqw$ and $\faqw$]~
\be
\item[(a)] For any $\faq$ query $Q$, the following holds:
\begin{align}
    \subfaqw(Q) \leq \ssubfaqw(Q) \leq \faqw(Q).
    \label{eqn:gaps}
\end{align}
In particular, when $Q$ has no free variables, we have
\begin{align}
    \subw(Q) \leq \ssubw(Q) \leq \fhtw(Q).
    \label{eqn:gaps:no:free}
\end{align}
\item[(b)] Furthermore, there are classes of queries $Q$ for which the gap
between
$\ssubfaqw(Q)$ and $\faqw(Q)$ is unbounded,
and so is the gap between $\ssubw(Q)$ and $\fhtw(Q)$.
\ee
\label{prop:ssubw-fhtw}
\eprop

\bp
First we prove part (a). The first inequality  in~\eqref{eqn:gaps}
follows directly from the definitions of $\ssubfaqw$ and $\subfaqw$ along with
the
fact that $\Gamma_n \subseteq \Gamma_{n|\calE_{\not\infty}}$.
To prove the second inequality in~\eqref{eqn:gaps}, we use the
following variant of the {\em
	Modularization Lemma} from~\cite{panda-arxiv}:
\begin{claim}[Variant of the Modularization Lemma~\cite{panda-arxiv}]
	{\em Given a hypergraph $\calH=(\calV=[n], \calE)$ and a set $B\subseteq \calV$, we
		have
		\begin{equation}
			\max_{h \in \ed \cap \Gamma_{n|\calE}} h(B)=
			\max_{h \in \ed \cap \Mod_n} h(B),
			\label{eq:modularization}
		\end{equation}
		where $\ed$ is given by \eqref{eqn:ed} and $\Mod_n$ denotes the set of all
		modular functions $h:2^\calV \rightarrow
		\R_+$. (A function $h:2^\calV \rightarrow
		\R_+$ is modular if $h(X) = \sum_{i \in X} h(i), \forall X \subseteq \calV$.)
	}
	\label{clm:modularization}
\end{claim}
\bp[Proof of Claim~\ref{clm:modularization}]
Obviously, the LHS of \eqref{eq:modularization} is lowerbounded by the RHS.
Next, we prove LHS $\leq$ RHS.
W.L.O.G. we assume $B = [k]$ for some $k \in [n]$.
Let $h^* = \argmax_{h \in \ed \cap \Gamma_{n|\calE}} h(B)$.
Define a function $\bar h:2^\calV\rightarrow \R_+$ as follows:
\[\bar h(F) = \sum_{i\in F} (h^*([i])-h^*([i-1])).\]
Obviously $\bar h \in \Mod_n$ and $\bar h(B) = h^*(B)$.
Next, we prove $\bar h \in \ed$ by proving that
for every $F\subseteq [n]$ where $F\subseteq E$ for some
$E\in\calE$, the following holds:
$\bar h(F) \leq h^*(F)$.

The proof is by  induction on $|F|$.
The base case when $|F|=0$ is trivial.
For the inductive step, consider some $F$ where $F \subseteq E$ for some $E \in
\calE$.
Let $j$ be the maximum integer in $F$, then
by noting that $|F\cap [j-1]|< |F|$, we have
\begin{eqnarray*}
	\bar h(F) &=& h^*([j]) - h^*([j-1]) + \sum_{i \in F-\{j\}}
	(h^*([i])-h^*([i-1]))\\
	&=& h^*([j]) - h^*([j-1]) + \bar h(F \cap [j-1])\\
	&=& h^*(F \cup [j-1]) - h^*([j-1]) + \bar h(F \cap [j-1])\\
	&\leq & h^*(F \cup [j-1]) - h^*([j-1]) +
	h^*(F \cap [j-1])
	\leq  h^*(F).
\end{eqnarray*}
The first inequality above is by induction hypothesis,
and the second inequality follows from the fact that $h^*$ is
a $\calE$-polymatroid (recall Definition~\ref{defn:calE-poly}).
Both steps rely on the fact that $F\cap[j-1] \subseteq E$ for some $E\in\calE$.
Consequently, $\bar h \in \ed \cap \Mod_n$.
Since $\bar h(B) = h^*(B)$, this proves Claim~\ref{clm:modularization}.
\ep
Now we prove the second inequality in~\eqref{eqn:gaps}:
\begin{eqnarray*}
	\ssubfaqw(Q) &=&
	\max_{h \in \ed_{\not\infty} \cap \Gamma_{n|\calE_{\not\infty}}}
	\min_{(T,\chi) \in \td_F}
	\max_{t\in V(T)} h(\chi(t))\\
	&\leq&
	\min_{(T,\chi) \in \td_F}
	\max_{h \in \ed_{\not\infty} \cap \Gamma_{n|\calE_{\not\infty}}}
	\max_{t\in V(T)} h(\chi(t))\\
	&=& \min_{(T,\chi) \in \td_F}
	\max_{t\in V(T)}\max_{h \in \ed_{\not\infty} \cap \Gamma_{n|\calE_{\not\infty}}}
	h(\chi(t))\\
	\text{(Claim~\ref{clm:modularization})}&=& \min_{(T,\chi) \in \td_F}
	\max_{t\in V(T)}\max_{h \in \ed_{\not\infty} \cap \Mod_n}
	h(\chi(t))\\
	\text{(Strong duality)}&=& \min_{(T,\chi) \in \td_F}
	\max_{t\in V(T)}\rho^*_{\calE_{\not\infty}}(\chi(t))
	= \faqw(Q).
\end{eqnarray*}
The fact that $\max_{h \in \ed_{\not\infty} \cap \Mod_n}
h(\chi(t)) = \rho^*_{\calE_{\not\infty}}(\chi(t))$ follows
from the two sides being dual linear programs. (Recall the definition of
$\rho^*$
from Section~\ref{subsec:width:params}.)

Now, we prove part (b) of Proposition~\ref{prop:ssubw-fhtw}.
In ~\cite{panda-arxiv}, we constructed a class of
graphs/queries
where the gap between $\fhtw$ and $\subw$ is unbounded.
We will re-use the same construction here and prove that the upperbound on
$\subw$ that we proved in~\cite{panda-arxiv} is also an upperbound on $\ssubw$.
The upperbound proof is going to be different from~\cite{panda-arxiv} though
since here we can only use $\calE$-polymatroid properties to prove the bound
(recall Definition~\ref{defn:calE-poly}).

Given integers $m$ and $k$, consider a graph $\calH=(\calV, \calE)$ which is
an
``$m$-fold $2k$-cycle'': The vertex set $\calV := I_1\cup\ldots\cup
I_{2k}$ is a disjoint union of $2k$-sets of vertices.
Each set $I_j$ has $m$ vertices in it, i.e., $I_j :=\{I_j^1, I_j^2, \ldots,
I_j^m\}$.
There is no edge between any two vertices within the set $I_j$ for every
$j \in [2k]$, i.e., $I_j$ is an independent set.
The edge set $\calE$ of the hypergraph is the union of $2k$ complete
bipartite graphs $K_{m,m}$:
\[ \calE := (I_1 \times I_2) \cup (I_2 \times I_3) \cup \cdots \cup
(I_{2k-1} \times I_{2k}) \cup (I_{2k} \times I_1).
\]
Finally consider an $\faq$ query $Q$ that has a finite-sized input factor $R_K$
for every $K \in \calE$, i.e., $\calE_{\not\infty} = \calE$ and
$\calE_{\infty}=\emptyset$ (recall notation from
Section~\ref{subsec:insideout:panda}).
Assuming $Q$ has no free variables, then $\faqw(Q)=\fhtw(Q)$
and $\ssubfaqw(Q) = \ssubw(Q)$.

We proved in~\cite{panda-arxiv} that $\fhtw(Q) \geq 2m$.
Next we prove that $\ssubw(Q)$ $\leq m(2-1/k)$.
Let $h$ be any function in $\ed_{\not\infty} \cap
\Gamma_{n|\calE_{\not\infty}}$.
We recognize two cases:
\bi
\item Case 1: $h(I_i) \leq \theta$ for some $i \in [2k]$. WLOG assume $h(I_1)
\leq
\theta$. Consider the TD

\begin{tikzpicture}[color=green!50!black, scale = .85, every
node/.style={transform shape, inner sep=1.5pt}]
\node[draw,ellipse] (123) {$I_1 \cup I_2 \cup I_3$};
\node[draw,ellipse, right = 0.25of 123] (134) {$I_1 \cup I_3 \cup I_4$};
\node[draw,ellipse, right = 0.5in of 134] (1k) {$I_1 \cup I_{2k-1} \cup
	I_{2k}$};
\draw (123) -- (134);
\draw[dotted] (134) -- (1k);
\end{tikzpicture}

For bag $B=I_1 \cup I_i \cup I_{i+1}$, using $\calE_{\not\infty}$-polymatroid
properties
(Definition~\ref{defn:calE-poly}), we
have
\begin{eqnarray*}
	h(B) &\leq& h(I_1)+h(I_i \cup I_{i+1})\\
	&\leq& h(I_1)+\sum_{j = 1}^m h\left(\left\{I_i^j, I_{i+1}^j\right\}\right)
	\leq \theta+m.
\end{eqnarray*}
\item Case 2: $h(I_i) > \theta$ for all $i \in [2k]$. Consider the TD

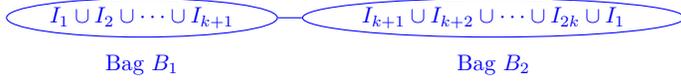
\begin{tikzpicture}[color=blue, scale = .85, every node/.style={transform
	shape,
	inner sep=1.5pt}]
\node[draw,ellipse] (B1) {$I_1 \cup I_2 \cup \cdots \cup I_{k+1}$};
\node[draw,ellipse, right = 0.15in of B1] (B2) {$I_{k+1} \cup I_{k+2} \cup
	\cdots \cup
	I_{2k} \cup I_1$};
\draw (B1) -- (B2);
\node[below = 0.21 of B1] {Bag $B_1$};
\node[below = 0.21 of B2] {Bag $B_2$};
\end{tikzpicture}

For convenience, given any vertex $I_i^j$, define the vertex set $\calV_i^j$ as
follows:
\[\calV_i^j := I_1\cup I_2\cup\ldots\cup I_{i-1} \cup\left\{I_i^1, I_i^2,
\ldots,
I_i^{j-1}\right\}.\]
From $\calE_{\not\infty}$-polymatroid properties, we have
\begin{eqnarray*}
	h(B_1) &=& h(I_1 \cup I_2) + \sum_{i=3}^{k+1}\sum_{j=1}^m
	h\left(\left\{I_i^j\right\} \cup \calV_i^j \;|\; \calV_i^j\right)\\
	&\leq& h(I_1 \cup I_2) + \sum_{i=3}^{k+1}\sum_{j=1}^m
	h\left(\left\{I_i^j, I_{i-1}^j\right\} \;|\; \left\{I_{i-1}^j\right\}\right)\\
	&=& h(I_1 \cup I_2) + \sum_{i=3}^{k+1}\sum_{j=1}^m
	h\left(\left\{I_i^j, I_{i-1}^j\right\}\right)
	- \sum_{i=3}^{k+1}\sum_{j=1}^m
	h\left(\left\{I_{i-1}^j\right\}\right)\\
	&\leq& h(I_1 \cup I_2) + \sum_{i=3}^{k+1}\sum_{j=1}^m
	h\left(\left\{I_i^j, I_{i-1}^j\right\}\right)
	- \sum_{i=3}^{k+1}h(I_{i-1})\\
	&\leq& \sum_{i=2}^{k+1}\sum_{j=1}^m
	h\left(\left\{I_i^j, I_{i-1}^j\right\}\right)
	- \sum_{i=3}^{k+1}h(I_{i-1})
	\leq km-(k-1)\theta.
\end{eqnarray*}
\ei
In a symmetric way, we can also show that $h(B_2)\leq km-(k-1)\theta$.
By setting $\theta = (1-1/k)m$, we prove that
$\ssubw(Q) \leq m(2-1/k)$.
Since $\fhtw(Q) \geq 2m$, this proves part (b).
\ep

\begin{example}
Consider again the count query $Q$ in~\eqref{eqn:4:cycle}, which we showed
earlier how to compute in time $\tilde O(N^{1.5})$.
Since $Q$ has no free variables, $\faqw(Q) = \fhtw(Q) = 2$ and $\ssubfaqw(Q) =
\ssubw(Q)$.
In the proof of Proposition~\ref{prop:ssubw-fhtw},
we show that $\ssubw(Q) \leq 1.5$.
Therefore, the $\spanda$ algorithm from the proof of
Theorem~\ref{thm:subfaqw:general}
can compute~\eqref{eqn:4:cycle} in time $\tilde O(N^{1.5})$.
In fact, the $\tilde O(N^{1.5})$ algorithm we described earlier
for~\eqref{eqn:4:cycle} is just
a specialization of $\spanda$.
The proof of Proposition~\ref{prop:ssubw-fhtw} offers a family of similar
examples.\qed
\end{example}

\subsubsection{$\faqai$ over an arbitrary semiring}

Finally, we put everything together to solve the $\faqai$ problem. The only
(very natural) change is to replace the tree decompositions by their relaxed version, and
the technical details flow through.

\bdefn
Given an $\faqai$ query~\eqref{eqn:our:query} whose hypergraph is
$\calH = (\calV, \calE = \calE_s \cup \calE_\ell = \calE_{\not\infty} \cup
\calE_{\infty})$, its
{\em relaxed \#-submodular $\faq$-width}, denoted
by $\ssubfaqw_\ell(Q)$, is \mbox{defined by}
\begin{align}
\ssubfaqw_\ell(Q) &:=
\max_{h \in \ed_{\not\infty} \cap \Gamma_{n|\calE_{\not\infty}}}
\min_{(T,\chi) \in \td_F^\ell}
\max_{t\in V(T)} h(\chi(t)).
\label{eqn:relaxed:ssubfaqw:general}
\end{align}
When $F = \emptyset$, we define $\ssubw_\ell(Q) :=
\ssubfaqw_\ell(Q)$.
\edefn

\bthm
Any $\faqai$ query $Q$ of the form~\eqref{eqn:our:query} on any semiring can be
computed in time
$\tilde O(N^{\ssubfaqw_\ell(Q)}+|Q|)$.
\label{thm:relaxed:subfaqw:general}
\ethm

The proof of the above theorem is very similar to that of Theorem~\ref{thm:subfaqw:general}.
The key difference is that instead of running $\InsideOut$ on individual $\faq$
queries obtained after applying $\spanda$,
we now run the $\InsideOut$ variant from Theorem~\ref{thm:relaxed:faqw}.
The proof is thus omitted.

\begin{example}
Consider the following count query (which is similar to the counting version of query $Q_2$ from Example~\ref{ex:intro-query3}):
    \begin{align}
Q() = \sum_{a, b, c, d} R(a,b) \cdot S(b,c) \cdot T(c,d) \cdot \bm 1_{a +b + c
+ d\leq 0}.
\end{align}
Let $N := \max\{|R|, |S|, |T|\}$.
For the above query $\faqw(Q) = \faqw_\ell(Q)$ $= \ssubfaqw(Q) = 2$.
Any of the previously known algorithms, including the one from
Theorem~\ref{thm:relaxed:faqw} {\em and} the one from
Theorem~\ref{thm:subfaqw:general}, would need time $O(N^2)$ to compute $Q$.
We show below that $\ssubfaqw_\ell(Q) \leq 1.5$.
As an example of Theorem~\ref{thm:relaxed:subfaqw:general}, we also show how to
compute the above query in $\tilde O(N^{1.5})$.
(Using the same method, we can also solve the counting version of $Q_2$
from Example~\ref{ex:intro-query3} in the same time.)

First we prove that for the above query, $\ssubfaqw_\ell(Q) \leq 1.5$.
Here $F = \emptyset$.
We will use two {\em relaxed} tree decompositions in $\td_F^\ell$:
The first $(T_1, \chi_1)$ has two bags $\{a, b, c\}$ and $\{c, d\}$.
The second $(T_2, \chi_2)$ has two bags $\{a, b\}$ and $\{b, c, d\}$.
(Both are relaxed TDs because the ligament edge $\bm 1_{a +b + c
	+ d\leq 0}$ is not
contained in any bag; recall Definition~\ref{defn:relaxed:td}.)
Following~\eqref{eqn:relaxed:ssubfaqw:general}, for each $h \in
\ed_{\not\infty} \cap \Gamma_{n|\calE_{\not\infty}}$, we will pick one TD or
the other.
In particular, given some $h \in
\ed_{\not\infty} \cap \Gamma_{n|\calE_{\not\infty}}$:
\bi
\item If $h(b) \geq 1/2$, then $h(bc|b) \leq 1/2$. We pick $(T_1, \chi_1)$.
From $\calE_{\not\infty}$-polymatroid properties (Def.~\ref{defn:calE-poly}), we have
\begin{eqnarray*}
	h(abc) &=& h(ab) + h(abc|ab) \leq h(ab) + h(bc|b) \leq 1.5,\\
	h(cd) &\leq& 1.
\end{eqnarray*}
\item If $h(b) < 1/2$, we pick $(T_2, \chi_2)$.
\begin{eqnarray*}
	h(ab) &\leq& 1,\\
	h(bcd) &=& h(b) + h(bcd|b) \leq h(b) + h(cd) \leq 1.5.
\end{eqnarray*}
\ei
This proves that $\ssubfaqw_\ell(Q) \leq 1.5$.

Finally, as a special case of $\spanda$, we explain how to solve the above
query in time $\tilde O(N^{1.5})$ (where recall $N:=\max\{|R|, |S|, |T|\}$).
Let
\begin{eqnarray*}
	S^\ell &:=&\left\{(b, c) \in S \suchthat |\{c' \suchthat (b, c') \in S\}|\leq
	\sqrt{N}\right\},\\
	S^h &:=& S \setminus S^\ell.
\end{eqnarray*}
Now we can write
\begin{eqnarray*}
	Q() &=& \sum_{a, b, c, d} R(a,b)  \left(S^\ell(b,c)+S^h(b, c)\right)
	T(c,d) \cdot
	\bm 1_{a +b + c
		+ d\leq 0}\\
	&=& Q^\ell() + Q^h(), \text{where}\\
	Q^\ell() &:=& \sum_{a, b, c, d} \underbrace{R(a,b) \cdot S^\ell(b,c)}_{U(a, b, c)} \cdot
	T(c,d) \cdot
	\bm 1_{a +b + c
		+ d\leq 0},\\
	Q^h() &:=& \sum_{a, b, c, d} R(a,b) \cdot \underbrace{S^h(b,c) \cdot T(c,d)}_{W(b, c, d)} \cdot
	\bm 1_{a +b + c
		+ d\leq 0}.\\
\end{eqnarray*}
Both $U$ and $W$ above have sizes $\leq N^{1.5}$.
Using the algorithm from the proof of Theorem~\ref{thm:relaxed:faqw},
$Q^\ell$ can be answered in time $O(N^{1.5}\log N)$ using the relaxed TD $(T_1,
\chi_1)$,
while $Q^h$ can be answered in the same time using $(T_2, \chi_2)$.
\qed
\label{ex:4:path:ineq:count}
\end{example}

\section{Applications to relational Machine Learning}
\label{sec:applications}

Our $\faqai$ formalism and solution are directly
applicable to learning a class of machine learning models, which includes
supervised models (e.g., robust regression, SVM classification), and
unsupervised models (e.g., clustering via $k$-means). In this section, we show that the core
computation of these optimization problems can be
formulated in $\faqai$ over the sum-product semiring.

\subsection{Training ML models over databases}
\label{sec:applications:intro}

A typical machine learning model is learned over a training dataset $\bm G$. We
consider the common scenario where the input data is a relational database $\bm I$,
and the training dataset $\bm G$ is the result of a feature extraction join
query $Q$ over $\bm I$~\cite{SOC:SIGMOD:16,ANNOS:PODS:2018, ANNOS:TODS:2020, KuNaPa15,
  MADlib:2012}. Each tuple $(\bm x, y) \in \bm G$ consists of a vector of
features $\bm x$ of length $n$ and a label $y$. We consider that the feature extraction
query $Q$ has the hypergraph $\calH = (\calV, \calE_s)$, where $\calE_s$
is the set of its skeleton hyperedges.

A supervised machine learning model is a function $f_{\bm\beta}(\bm x)$ with
parameters $\bm\beta$ that is used to predict the label $y$ for unlabeled
data. The parameters are obtained by minimizing the objective function:
\begin{align}
  J(\vec\beta) = \sum_{(\bm x,y) \in \bm G} {\mathcal L}
  \left(y,f_{\bm\beta}(\bm x)\right) + \lambda
  \Omega(\vec\beta), \label{eqn:generic:J}
\end{align}
where $\mathcal L(a,b)$ is a loss function, $\Omega$ is a regularizer, e.g.,
$\ell_1$ or $\ell_2$ norm, and the constant $\lambda \in (0,1)$ controls the
influence of regularization.

Previous work has shown that for polynomial loss functions, such as square loss
$\mathcal L(a, b) = (a - b)^2$, the core computation for optimizing the
objective $J(\vec\beta)$ amounts to $\faq$ evaluation~\cite{ANNOS:PODS:2018}.  In many instances, however, the loss function
is non-polynomial, either due to the structure of the loss, or the presence of
non-polynomial components embedded within the model structure (e.g., ReLU
activation function in neural nets)~\cite{murphy2013}.

\nop{
}

Examples of commonly used non-polynomial loss functions are: (1) hinge loss, used
to learn classification models like linear support vector machines
(SVM)~\cite{murphy2013}, or generalized low rank models (glrm) with boolean
principal component analysis (PCA)~\cite{glrm}; (2) Huber loss,
used to learn regression models that are robust to
outliers~\cite{murphy2013}; (3) scalene loss, used to learn quantile
regression models~\cite{glrm}; (4) epsilon insensitive loss, used to
learn SVM regression models~\cite{murphy2013}; and (5) ordinal hinge loss, used to learn ordinal regression models or ordinal PCA (another
glrm)~\cite{glrm}.

Any optimization problem with the above non-polynomial loss functions can
benefit from our evaluation algorithm for $\faqai$ by reformulating computations
in the optimization algorithm as $\faqai$ expressions over the feature extraction join query $Q$.
We next exemplify this reformulation for the following problems:
\bi
\item Learning a robust linear regression model using Huber loss, which can be solved with
gradient-descent optimization
\item Learning a linear regression model using the scalene, epsilon
insensitive, and ordinal hinge loss functions.
\item Learning a linear support vector machine
(SVM) for binary classification using hinge loss, which can be solved with
subgradient-based optimization algorithms or with a cutting-plane algorithm for
the primal formulation of linear SVM classification.
\item We also consider $k$-means unsupervised
clustering and give an $\faqai$ reformulation of the computation done in an iteration of the algorithm over the dataset $\bm G$.
\ei

The advantage of $\faqai$ reformulation is that the $\faqai$ expressions for the aforementioned optimization problems can be evaluated over relaxed tree
decompositions of the feature extraction query $Q$ and do not require the explicit materialization of its result $\bm G$. The size of and time to compute $\bm G$ is $\tilde O(|I|^{\rho^*(Q)})$~\cite{DBLP:conf/pods/NgoPRR12}. The solution to these optimization problems can be computed in time sub-linear in the size of $\bm G$, using $\InsideOut$ or $\spanda$.

\subsection{Background: Gradient-based Optimization}
\label{sec:background:gradient}

In this section, we overview gradient-based optimization algorithms for convex
and differentiable objective functions of the form~\eqref{eqn:generic:J}. A
gradient-based optimization algorithm employs the first-order gradient
information to optimize $J(\bm\beta)$. It repeatedly updates the parameters
$\bm\beta$ by some step size $\alpha$ in the direction of the gradient
$\grad J(\bm\beta)$ until convergence. To guarantee convergence, it is common to
use backtracking line search to ensure that the step size $\alpha$ is
sufficiently small to decrease the loss for each step. Each update step requires
two computations: (1) {\em Point evaluation}: Given $\vec\theta$, compute the
scalar $J(\vec\theta)$; and (2) {\em Gradient computation}: Given $\vec\theta$,
compute the vector $\grad J(\vec\theta)$.

There exist several variants of gradient descent algorithms, e.g., batch
gradient descent or stochastic gradient descent, as well as many different
algorithms to choose a valid step size~\cite{murphy2013}. For this work, we
consider the batch gradient descent (BGD) algorithm with the Armijo backtracking
line search condition, as depicted in Algorithm~\ref{algo:bgd}. A common choice
for setting the step size is a function that is inversely related to number of
iterations of the algorithm, for instance $\alpha = \frac{1}{\lambda t}$ at
iteration $t$, where $\lambda$ is the regularization parameter
from~\eqref{eqn:generic:J}~\cite{pegasos}.

\begin{algorithm}[th]
    \caption{BGD with Armijo line search.}
    \label{algo:bgd}
    $\bm\beta \gets $ a random point\;
    \While{not converged yet}{
        $\alpha \gets $ next step size\;
        $\bm d \gets \grad J(\bm\beta)$\;
        \texttt{// Line search with Armijo condition}\;
        \While{$\left(J(\bm\beta-\alpha \bm d) \geq J(\bm\beta)-\frac
            \alpha 2 \norm{\bm d}_2^2\right)$}{
            $\alpha \gets \alpha/2$\;
        }
        $\bm\beta \gets \bm\beta - \alpha \bm d$\;
    }
\end{algorithm}

\subsection{Robust linear regression with Huber loss}
\label{sec:huber:loss}

A linear regression model is a linear function
$f_{\bm\beta}(\bm x) = \bm\beta^\top \bm x = \sum_{i \in [n]} \beta_i x_i$ with
features $\bm x = (x_1 = 1, x_2, \ldots, x_n)$ and parameters
$\bm\beta = (\beta_1, \ldots, \beta_n)$. For a given feature vector $\bm x$, the
model is used to estimate the (continuous) label $y \in \R$. We learn the model
parameters by minimizing the objective $J(\bm\beta)$ with the Huber loss
function, which is defined as:
\begin{align}\label{eqn:huber:loss}
  {\mathcal L}(a, b) =
  \begin{cases}
    \frac{1}{2} (a - b)^2
    & \text{if } | a - b| \leq 1, \\
    \frac{1}{2}|a - b| - \frac{1}{2}& \text{otherwise}.
  \end{cases}
\end{align}

Huber loss is equivalent to the square loss when
$|a - b| \leq 1$ and to the absolute loss
otherwise\footnote{Without loss of generality, we use a simplified
  Huber loss. The threshold between absolute and square loss is given
  by a constant $\delta$ and the absolute loss is
  $\frac{\delta}{2}|a - b| - \frac{\delta^2}{2}$.}.
In contrast to the absolute loss, Huber loss is differentiable at all points. It is also more robust to outliers than the square loss.

To learn the parameters, we use batch gradient-descent optimization, which
repeatedly updates the parameters in the direction of the gradient
$\grad J(\bm\beta)$ until convergence. We provide details on gradient-based
optimization in Section~\ref{sec:background:gradient}.
In this
section, we focus on the core computation of the algorithm, which is the
repeated computation of the objective $J(\bm\beta)$ and its gradient
$\grad J(\bm\beta)$.

The gradient $\grad J(\bm\beta)$ is the vector of partial
derivatives with respect to parameters $(\beta_j)_{j\in [n]}$.
(Note that the derivative of $\bm 1_{\theta(x)\geq 0}$ with respect to $x$ for any function $\theta$ of $x$ is always $0$ whenever it is defined.)
The objective
function $J(\bm\beta)$ (with $\ell_2$ regularization) and its partial
derivative with respect to $\beta_j$ are:
\begin{align}
  J(\bm \beta)
  &=
    \frac{1}{2} \sum_{(\bm x, y) \in \bm G} \bigl[(y - f_{\bm\beta}(\bm x))^2 \cdot
    \bm 1_{|y - f_{\bm\beta}(\bm x)| \leq 1}
    + (|y - f_{\bm\beta}(\bm x)| - 1) \cdot
    \bm 1_{|y - f_{\bm\beta}(\bm x)| > 1}\bigr]+  \frac{\lambda}{2} \norm{\bm\beta}^2_2,  \label{eq:pointevalhuber}\\
  \frac{\partial J(\bm\beta)}{\partial \beta_j}
  &=  - \sum_{(\bm x,y)\in\bm G} \bigl[(y - f_{\bm\beta}(\bm x))\cdot x_j \cdot
    \bm 1_{|y - f_{\bm\beta}(\bm x)| \leq 1}  \label{eq:gradevalhuber}\\
  &\quad + \frac{1}{2}(x_j  \cdot  \bm 1_{y -
f_{\bm\beta}(\bm x) > 0} -  x_j  \cdot
    \bm 1_{y - f_{\bm\beta}(\bm x) < 0})
     \cdot  \bm 1_{|y - f_{\bm\beta}(\bm x)| > 1} \bigr] + \lambda\beta_j\nonumber \\
  &=   -  \sum_{(\bm x,y)\in\bm G}  (y - f_{\bm\beta}(\bm x))\cdot x_j \cdot
    \bm 1_{|y - f_{\bm\beta}(\bm x)| \leq 1} \nonumber \\
  &\quad {}- 1/2  \sum_{(\bm x,y)\in\bm G}  x_j
 \cdot  \bm 1_{y - f_{\bm\beta}(\bm x) > 1}
    + 1/2  \sum_{(\bm x,y)\in\bm G}  x_j
      \cdot  \bm 1_{y - f_{\bm\beta}(\bm x) < -1} \;\;+ \lambda\beta_j. \nonumber
\end{align}

Our observation is that we can compute $J(\bm \beta)$ and
$\frac{\partial J(\bm\beta)}{\partial \beta_j}$ without materializing $\bm G$,
by reformulating their data-dependent computation as a few $\faqai$
expressions. We explain the details next.

\subsubsection{Reformulating the objective $J(\bm \beta)$ with Huber loss into \faqai{} expressions}
\label{sec:pointevalhuber}

We show that the objective $J(\bm\beta)$ from~\eqref{eq:pointevalhuber}
can be reformulated into $O(n^2)$ $\faqai$ expressions of the
form~\eqref{eqn:our:query}.

First, we consider the case where $|y - f_{\bm\beta}(\bm x)| \leq 1$, i.e. the
square loss term of $J(\bm\beta)$.  For ease of notation, let
$c_1(y,\bm x) = |y - f_{\bm\beta}(\bm x)| \leq 1$.
\begin{align*}
	&\sum_{(\bm x, y) \in \bm G} (y - f_{\bm\beta}(\bm x))^2 \cdot
	\bm 1_{c_1(y,\bm x)} \\
	&= \sum_{(\bm x, y) \in \bm G} y^2  - 2 y f_{\bm\beta}(\bm x) + (f_{\bm\beta}(\bm x))^2 \cdot
	\bm 1_{c_1(y,\bm x)} \\
	&= \sum_{(\bm x, y) \in \bm G} y^2 \cdot   \bm 1_{c_1(y,\bm x)}
	- 2 \sum_{(\bm x, y) \in \bm G} y \cdot f_{\bm\beta}(\bm x) \cdot\bm 1_{c_1(y,\bm x)} + \sum_{(\bm x, y) \in \bm G} (f_{\bm\beta}(\bm x))^2 \cdot
	\bm 1_{c_1(y,\bm x)} \\
	&= \sum_{(\bm x, y) \in \bm G} y^2 \cdot   \bm 1_{c_1(y,\bm x)}
	- 2 \sum_{i \in [n]}\sum_{(\bm x, y) \in \bm G} \beta_i \cdot y  \cdot x_i \cdot\bm 1_{c_1(y,\bm x)} + \sum_{i \in [n]}\sum_{j \in [n]}\sum_{(\bm x, y) \in \bm G} \beta_i \cdot \beta_j \cdot x_i \cdot x_j \cdot \bm 1_{c_1(y,\bm x)} \\
	&= \sum_{(\bm x, y) \in \bm G} y^2 \cdot   \bm 1_{y - f_{\bm\beta}(\bm x) \leq 1}
	\cdot\bm 1_{y - f_{\bm\beta}(\bm x) \geq 0} \\
	&\qquad + \sum_{(\bm x, y) \in \bm G} y^2 \cdot
	\bm 1_{y - f_{\bm\beta}(\bm x) \geq -1} \cdot \bm 1_{y - f_{\bm\beta}(\bm x) < 0}\\
	& \qquad - 2 \sum_{i \in [n]}\sum_{(\bm x, y) \in \bm G} \beta_i \cdot y  \cdot x_i \cdot
	\bm 1_{y - f_{\bm\beta}(\bm x) \leq 1}\cdot\bm 1_{y - f_{\bm\beta}(\bm x) \geq 0}\\
	& \qquad - 2 \sum_{i \in [n]}\sum_{(\bm x, y) \in \bm G} \beta_i \cdot y  \cdot x_i \cdot
	\bm 1_{y - f_{\bm\beta}(\bm x) \geq -1} \cdot \bm 1_{y - f_{\bm\beta}(\bm x) < 0}\\
	&\qquad + \sum_{i \in [n]}\sum_{j \in [n]}\sum_{(\bm x, y) \in \bm G} \beta_i \cdot \beta_j
	\cdot x_i \cdot x_j \cdot
	\bm 1_{y - f_{\bm\beta}(\bm x) \leq 1}\cdot\bm 1_{y - f_{\bm\beta}(\bm x) \geq 0}\\
	&\qquad + \sum_{i \in [n]}\sum_{j \in [n]}\sum_{(\bm x, y) \in \bm G} \beta_i\cdot\beta_j
	\cdot x_i \cdot x_j \cdot
	\bm 1_{y - f_{\bm\beta}(\bm x) \geq -1} \cdot \bm 1_{y - f_{\bm\beta}(\bm x) < 0}\\
\end{align*}

Each summation over the training dataset $\bm G$ in the final reformulation
above can be expressed as one $\faqai$ query with two ligament hyperedges. For
instance, the first summation over $\bm G$ is equivalent to the following
$\faqai$ expression:
\begin{align*}
	Q() = \sum_{y, \bm x_\calV}  y^2 \cdot \underbrace{\bm 1_{y - f_{\bm\beta}(\bm x) \leq 1}
		\cdot\bm 1_{y - f_{\bm\beta}(\bm x) \geq 0}}_{\text{ligaments in }\calE_\ell} \cdot
	\left(\prod_{F \in \calE_s} R_F(x_F) \right)
\end{align*}

The absolute loss function for the case $|y - f_{\bm\beta}(\bm x)| > 1$ can be
reformulated similarly:
\begin{align*}
	&\sum_{(\bm x, y) \in \bm G} (|y - f_{\bm\beta}(\bm x)| - 1) \cdot
	\bm 1_{|y - f_{\bm\beta}(\bm x)| > 1}\\
	&= \sum_{(\bm x, y) \in \bm G} (y - f_{\bm\beta}(\bm x) - 1) \cdot
	\bm 1_{y - f_{\bm\beta}(\bm x) > 1}+\sum_{(\bm x, y) \in \bm G}(f_{\bm\beta}(\bm x) - y - 1) \cdot
	\bm 1_{y - f_{\bm\beta}(\bm x) < -1}\\
	&= \sum_{(\bm x, y) \in \bm G} y  \cdot \bm 1_{y - f_{\bm\beta}(\bm x) > 1}
	- \sum_{i \in [n]} \sum_{(\bm x, y) \in \bm G} \beta_i \cdot x_i \cdot
	\bm 1_{y - f_{\bm\beta}(\bm x) > 1} \\
	&\qquad-\sum_{(\bm x, y) \in \bm G} y  \cdot \bm 1_{y - f_{\bm\beta}(\bm x) < -1}
	+ \sum_{i \in [n]} \sum_{(\bm x, y) \in \bm G} \beta_i \cdot x_i \cdot
	\bm 1_{y - f_{\bm\beta}(\bm x) < -1} \\
	&\qquad- \sum_{(\bm x, y) \in \bm G} \bm 1_{y - f_{\bm\beta}(\bm x) > 1}
	- \sum_{(\bm x, y) \in \bm G} \bm 1_{y - f_{\bm\beta}(\bm x) < -1}\\
\end{align*}

All of these terms can be reformulated as $O(n)$ $\faqai$ expressions of the
form~\eqref{eqn:our:query}.

Overall, the objective $J(\bm\beta)$ with Huber loss for learning robust linear
regression models can be computed with $O(n^2)$ $\faqai$ expressions, and
without materializing the training dataset $\bm G$.
Section~\ref{sec:gradevalhuber}
shows that the same
holds for $\frac{\partial J(\bm\beta)}{\partial \beta_j}$.

\subsubsection{Reformulating the gradient $\frac{\partial J(\bm\beta)}{\partial \beta_j}$ with Huber loss into \faqai{} expressions}
\label{sec:gradevalhuber}
We rewrite the  first of the three summations in $\frac{\partial J(\bm\beta)}{\partial \beta_j}$ from~\eqref{eq:gradevalhuber} as follows:
\begin{align}
	&\sum_{(\bm x,y)\in\bm G}  (y -
	\sum_{i\in[n]}  \beta_i x_i)  \cdot
	 x_j  \cdot
	\bm 1_{|y - f_{\bm\beta}(\bm x)| \leq 1} \\
	&=  \sum_{(\bm
		x,y)\in\bm G} y  \cdot  x_j
	 \cdot
	\bm 1_{|y - f_{\bm\beta}(\bm x)| \leq 1}
	- \sum_{i\in[n]}
	\sum_{(\bm x,y)\in\bm G} \beta_i  \cdot
	 x_i  \cdot  x_j
	 \cdot
	\bm 1_{|y - f_{\bm\beta}(\bm x)| \leq 1}\nonumber\\
	&=\sum_{(\bm x,y)\in\bm G} y \cdot x_j \cdot
	\bm 1_{y - f_{\bm\beta}(\bm x) \leq 1} \cdot
	\bm 1_{y - f_{\bm\beta}(\bm x) > 0}\nonumber \\
	&\qquad + \sum_{(\bm x,y)\in\bm G}
	y \cdot x_j \cdot
	\bm 1_{y - f_{\bm\beta}(\bm x) \geq -1} \cdot
	\bm 1_{y - f_{\bm\beta}(\bm x) < 0}\nonumber\\
	&\qquad -\sum_{i\in[n]} \sum_{(\bm x,y)\in\bm G}
	\beta_i \cdot x_i \cdot x_j \cdot
	\bm 1_{y - f_{\bm\beta}(\bm x) \leq 1} \cdot
	\bm 1_{y - f_{\bm\beta}(\bm x) > 0} \nonumber \\
	&\qquad -\sum_{i\in[n]} \sum_{(\bm x,y)\in\bm G}
	\beta_i\cdot x_i \cdot x_j \cdot
	\bm 1_{y - f_{\bm\beta}(\bm x) \geq -1} \cdot
	\bm 1_{y - f_{\bm\beta}(\bm x) < 0} . \nonumber
\end{align}

The four terms can be expressed as $O(n)$ $\faqai$ expressions of the
form~\eqref{eqn:our:query}. For instance, the first part of the expression is
equivalent to the following $\faqai$ query:
\begin{align}
	Q() &= \sum_{y, \bm x_{\calV}} y \cdot x_j  \cdot
	\underbrace{
		\bm 1_{y - f_{\bm\beta}(\bm x) \leq 1} \cdot
		\bm 1_{y -  f_{\bm\beta}(\bm x) > 0} }_{\text{ligaments } \calE_\ell}\cdot
	\left(\prod_{F\in \calE_s} R_F(\bm x_F)\right).\nonumber
\end{align}

The other two summations in $\frac{\partial J(\bm\beta)}{\partial \beta_j}$ both
aggregate over $x_j$ and have one inequality that defines a ligament in
$\calE_\ell$. They can be expressed as $\faqai$ expressions.
Overall, the gradient $\grad J(\bm\beta)$ can be expressed as $O(n^2)$ $\faqai$
expressions.

\bdefn[$Q_\ell$: The ligament extension of $Q$]
\label{defn:Q_ell}
Given an $\faq$ query $Q$ with hypergraph $\calH=(\calV, \calE)$,
define the {\em ligament extension of $Q$}, denoted by $Q_\ell$, to be an $\faqai$ query with hypergraph $\calH_\ell = (\calV, \calE_s \cup \calE_\ell)$
whose set of skeleton edges $\calE_s$ is identical to $\calE$
and whose set of ligament edges $\calE_\ell$ contains a single ligament edge $\calV$, i.e. $\calE_s = \calE$ and $\calE_\ell =\{\calV\}$.
\edefn

\bthm
\label{thm:huber:loss}
Let $\bm I$ be an input database where $N$ is the largest relation in $\bm I$, and $Q$
be a feature extraction query. For any
robust linear regression model $\bm\beta^\top\bm x$, the objective $J(\bm\beta)$
and gradient $\grad J(\bm\beta)$ with Huber loss can be computed in time
$\tilde O(N^{\ssubfaqw_\ell(Q_\ell)})$ with $\spanda$ and
in time $O(N^{\faqw_\ell(Q_\ell)} \log N)$ with $\InsideOut$,
where $Q_\ell$ is the ligament extension of $Q$ (Def.~\ref{defn:Q_ell}).
\ethm

\bp
Let $n$ be the number of variables in $Q$.
We show in Sections~\ref{sec:pointevalhuber} and \ref{sec:gradevalhuber}
that we can rewrite objective $J(\bm\beta)$
and the gradient $\grad J(\bm\beta)$ into $O(n^2)$ $\faqai$
expressions with at most $|\calE_\ell| = 2$ ligament hyperedges. The overall
runtime bound for computing $J(\bm\beta)$ and $\grad J(\bm\beta)$ with
$\spanda$ follows from Theorem~\ref{thm:relaxed:subfaqw:general}, which
states that $\spanda$ can compute each $\faqai$ expression in time
$\tilde O(N^{\ssubfaqw_\ell(Q_\ell)})$.

The overall runtime bound for computing $J(\bm\beta)$ and $\grad J(\bm\beta)$
with $\InsideOut$ follows from Theorem~\ref{thm:relaxed:faqw}, which states that
$\InsideOut$ can compute each $\faqai$ expression in time $O(N^{\faqw_\ell(Q_\ell)} \log N)$.
\ep

\subsection{Further non-polynomial loss functions}
\label{sec:otherlosses}

In this section, we overview the following non-polynomial loss functions: (1) epsilon insensitive loss;
(2) ordinal hinge loss; and (3) scalene loss. For each
function, we define the loss function $\calL$, the corresponding objective
function $J(\bm\beta)$, and the partial (sub)derivative
$\frac{\partial J(\bm\beta)}{\partial \beta_j}$ which is used in
(sub)gradient-based optimization algorithms. (Recall notation from Section~\ref{sec:applications:intro}.)  In the derivations for the objective
$J(\bm\beta)$, we will focus on the loss function and ignore the regularizer for
better readability.

As in the previous section, the objective and (sub)derivative can be
reformulated into several $\faqai$ expressions of the form~\eqref{eqn:our:query}. Instead
of writing out the expressions explicitly, we annotate those terms that can be
reformulated. The actual reformulation should be clear from the examples in
the previous sections.

\subsubsection*{Epsilon insensitive loss}
The epsilon insensitive loss function~\cite{murphy2013} is defined as:
\begin{align*}
    \calL(a,b) =
    \begin{cases}
        0 & \text{if } | a - b| \leq \epsilon \\
        |a - b| - \epsilon & \text{otherwise}
    \end{cases}
\end{align*}

This loss function is used to learn SVM regression models. We consider 
a linear regression model $f_{\bm\beta} (\bm x) = \bm\beta^\top \bm x= \sum_{i \in [n]} \beta_i x_i$. The
objective function and the corresponding partial subderivative with respect to
$\beta_j$ are given by:
\begin{align*}
    J(\bm\beta)
    &= \sum_{(\bm x, y) \in \bm G} (|y - \bm\beta^\top \bm x| - \epsilon) \cdot
    \bm 1_{|y - f_{\bm\beta}(\bm x)| > \epsilon}\\
    &= \underbrace{\sum_{(\bm x, y) \in \bm G} (y - \bm\beta^\top \bm x - \epsilon) \cdot
        \bm 1_{y - f_{\bm\beta}(\bm x) > \epsilon}}
    _{\faqai \text{ query of form ~\eqref{eqn:our:query}}} + \underbrace{\sum_{(\bm x, y) \in \bm G}
        (\bm\beta^\top \bm x - y - \epsilon) \cdot
        \bm 1_{f_{\bm\beta}(\bm x)-y > \epsilon}}
    _{\faqai \text{ query of form ~\eqref{eqn:our:query}}}  \\
    \frac{\partial J(\bm\beta)}{\partial \beta_j}
    &= \underbrace{\sum_{(\bm x, y) \in \bm G}  x_j \cdot
        \bm 1_{f_{\bm\beta}(\bm x) - y > \epsilon}}
    _{\substack{\text{$\faqai$ query of form~\eqref{eqn:our:query}}\\ \text{for each $j \in [n]$}}} -
    \underbrace{\sum_{(\bm x, y) \in \bm G} x_j \cdot \bm 1_{y - f_{\bm\beta}(\bm x) > \epsilon}}
    _{\substack{\text{$\faqai$ query of form~\eqref{eqn:our:query}}\\ \text{for each $j \in [n]$}}}
\end{align*}

The objective can thus be reformulated into $O(1)$
$\faqai$ queries, while the gradient can be reformulated into $O(n)$ queries: one for each $\beta_j$ for $j\in[n]$.

\subsubsection*{Ordinal hinge loss}
The ordinal hinge loss~\cite{glrm} is defined as:
\begin{align*}
    \calL(a,b)
    &= \sum_{t = 1}^{a-1} \max(0, 1 - b + t) +
    \sum_{t = a+1}^{d} \max(0, 1 + b - t) \\
    & = \sum_{t = 1}^{d} \bigl[\max(0, 1 - b + t) \cdot \bm 1_{t < a} +
    \max(0, 1 + b - t) \cdot \bm 1_{t > a}\bigr] \\
    & = \sum_{t = 1}^{d} \bigl[(1 - b + t) \cdot \bm 1_{t < a} \cdot \bm 1_{b < t + 1}+(1 + b - t) \cdot \bm 1_{t > a} \cdot \bm 1_{b > t - 1} \bigr]\\
\end{align*}

The  loss  function is  used  to  learn  ordinal  regression models  or  ordinal
PCA~\cite{glrm}.  A   linear  ordinal   regression  model  is the linear  function
$f_{\bm\beta} (\bm  x) = \bm\beta^\top  \bm x$  which predicts an  ordinal label
$y \in [d]$.  The objective function and the partial subderivative with respect
to $\beta_j$ are given by:
\begin{align*}
    J(\bm\beta)
    &= \sum_{t = 1}^{d} \underbrace{\sum_{(\bm x, y) \in \bm G}  (1 - \bm\beta^\top  \bm x + t)
        \cdot \bm 1_{\bm\beta^\top  \bm x < 1 + t}\cdot \bm 1_{y < t}}_
    {\substack{\text{$\faqai$ query of form~\eqref{eqn:our:query}}\\ \text{for each $t \in [d]$}}}\\
    &\qquad+\sum_{t = 1}^{d} \underbrace{\sum_{(\bm x, y) \in \bm G} (1 + \bm\beta^\top  \bm x - t)
        \cdot \bm 1_{\bm\beta^\top  \bm x > t - 1}  \cdot \bm 1_{y > t}}_
    {\substack{\text{$\faqai$ query of form~\eqref{eqn:our:query}}\\ \text{for each $t \in [d]$}}}  \\
    \frac{\partial J(\bm\beta)}{\partial \beta_j}
    &=\sum_{t = 1}^{d}  \underbrace{\sum_{(\bm x, y) \in \bm G} x_j
        \cdot \bm 1_{\bm\beta^\top  \bm x > t - 1}  \cdot \bm 1_{y > t}}
    _{\substack{\text{$\faqai$ query of form~\eqref{eqn:our:query}}\\ \text{for each $t \in [d], j\in[n]$}}}- \sum_{t = 1}^{d} \underbrace{\sum_{(\bm x, y) \in \bm G}  x_j
        \cdot \bm 1_{\bm\beta^\top  \bm x < 1 + t}\cdot \bm 1_{y < t}}
    _{\substack{\text{$\faqai$ query of form~\eqref{eqn:our:query}}\\ \text{for each $t \in [d], j\in[n]$}}}\\
\end{align*}

The objective and partial subderivative can thus be reformulated as
$O(d \cdot n)$ $\faqai$ expressions.

\subsubsection*{Scalene loss}
The scalene loss function~\cite{glrm} is defined as:
\begin{align*}
    \calL(a,b)
    &= \alpha \cdot \max(0, a - b) + (1 - \alpha) \cdot \max(0, b - a) \\
    &= \alpha \cdot (a - b) \cdot \bm 1_{a > b} + (1 - \alpha) \cdot (b - a)
    \cdot \bm 1_{b > a}
\end{align*}
where $\alpha \in (0,1)$ is a constant.

The loss function is used to learn quantile regression models. We again consider a linear regression model $f_{\bm\beta} (\bm  x) = \bm\beta^\top  \bm x$. The objective
function and the partial subderivative with respect to $\beta_j$ are given by:
\begin{align*}
    J(\bm\beta)
    &= \alpha \cdot \underbrace{\sum_{(\bm x, y) \in \bm G}
        (y - f_{\bm\beta}(\bm x)) \cdot \bm 1_{y > f_{\bm\beta}(\bm x)}}
    _{\faqai \text{ query of form~\eqref{eqn:our:query}}} + (1 - \alpha)\cdot  \underbrace{\sum_{(\bm x, y) \in \bm G} (
        f_{\bm\beta}(\bm x) - y)  \cdot \bm 1_{f_{\bm \beta}(\bm x) >y}}
    _{\faqai \text{ query of form~\eqref{eqn:our:query}}}  \\
    \frac{\partial J(\bm\beta)}{\partial \beta_j}
    &= (1 - \alpha)\cdot  \underbrace{\sum_{(\bm x, y) \in \bm G} x_j  \cdot
        \bm 1_{f_{\bm \beta}(\bm x) > y}}
    _{\substack{\text{$\faqai$ query of form~\eqref{eqn:our:query}}\\ \text{for each $j \in [n]$}}}
    - \alpha \cdot\underbrace{\sum_{(\bm x, y) \in \bm G} x_j \cdot
        \bm 1_{y > f_{\bm\beta}(\bm x)}}
    _{\substack{\text{$\faqai$ query of form~\eqref{eqn:our:query}}\\ \text{for each $j \in [n]$}}}  \\
\end{align*}

The objective and partial subderivative can thus be reformulated as
$O(n)$ $\faqai$ expressions.

Overall, we can reformulate the (sub)gradients under each one of the loss functions discussed in this section as $\faqai$ queries that are {\em ligament extensions} of the feature extraction query $Q$ as per Def.~\ref{defn:Q_ell}.

\subsection{Linear support vector machines}

A linear SVM classification model is used for binary classification problems
where the label $y \in \{\pm 1\}$. For the features
$\bm x = (x_1 = 1, x_2, \ldots, x_n)$, the model learns the parameters
$\bm\beta = (\beta_1, \ldots, \beta_n)$ of a linear discriminant function
$f_{\bm\beta}(\bm x) = \bm\beta^\top\bm x$ such that $f_{\bm\beta}(\bm x)$
separates the data points in $\bm G$ into positive and negative classes with a
maximum margin. The parameters can be learned by minimizing the objective
function~\eqref{eqn:generic:J} with the hinge loss function:
\begin{align}
  \mathcal L (a, b) = \max\{0,1 - a \cdot b \}.
  \label{eqn:hinge:loss}
\end{align}

Hinge loss is non-differentiable, and thus standard gradient descent optimization is not
applicable. We next discuss two alternative approaches for solving this optimization.

The first approach is based on the observation that the loss function
is convex, and the objective admits subgradient vectors, which generalize the standard notion of
gradient. The optimization
problem can be solved with subgradient-based updates. Pegasos is a
well-known algorithm for this approach~\cite{pegasos}.

The alternative approach is to solve the primal formulation of the problem, which
avoids the non-differentiable objective by turning it into a constraint
optimization problem with slack variables. Joachims proposed a cutting-plane
algorithm which solves this optimization problem efficiently~\cite{Joachims:2006}.

For both approaches, the number of iterations of the optimization algorithm is
independent of the size $|\bm G|$ of training dataset $\bm G$~\cite{pegasos,
  Joachims:2006}. Since each iteration takes $O(|\bm G|)$ time and the number of iterations is $O(1)$, it follows that the overall time complexity is $O(|\bm G|)$.

Despite the fact that the two approaches solve the same problem, they have been
hugely influential in their own right. We therefore consider both approaches,
and show that by reformulating their computation as $\faqai$ we can solve them
asymptotically faster than materializing the training dataset $\bm G$, i.e.,
sublinear in $|\bm G|$.

\subsubsection{Background on Subgradient Descent}
\label{sec:background:subgradient}
If the objective function $J(\bm\beta)$ is convex but not differentiable, the
gradient $\grad J(\bm\beta)$ is not defined. Such objective functions do,
however, admit a subgradient, which can be used in subgradient-based
optimization algorithms. Algorithm~\ref{algo:bgd} naturally captures the batch
subgradient-descent algorithm, if the parameters are updated in the direction of
the subgradient as opposed to the gradient.

A popular application for subgradient-descent optimization algorithms is the
learning of linear SVM models. One such algorithm is the Pegasos
algorithm~\cite{pegasos}, which showed that subgradient methods can learn the
parameters of the model significantly faster than other approaches, including
Joachims' cutting plane algorithm~\cite{Joachims:2006}.

\subsubsection{Subgradient-based optimization for linear SVM classification}
\label{sec:hinge:loss}

We first use subgradient-based optimization to compute
the parameters of the SVM model;
see Section~\ref{sec:background:subgradient} for some background.
The core of the optimization
is the repeated computation of the objective and the partial
derivatives in terms of $(\beta_j)_{j\in [n]}$. The objective $J(\bm\beta)$
(with $\ell_2$ regularization) and the partial derivative
$\frac{\partial J(\bm\beta)}{\partial \beta_j}$ are:
\begin{align}
  J(\bm\beta)
  &= \sum_{(\bm x,y)\in\bm G} \max\{0,1- y\cdot\bm\beta^\top \bm x\} + \frac{\lambda}{2}\norm{\bm\beta}^2_2,   \label{eqn:hinge:objective}
\\
  \frac{\partial J(\bm\beta)}{\partial \beta_j}
  &= -\sum_{(\bm x,y)\in\bm G} y\cdot x_j \cdot \bm 1_{y\cdot\bm\beta^\top \bm x \leq 1} + \lambda\beta_j.
   \label{eqn:hinge:partial}
\end{align}

Both $J(\bm\beta)$ and
$\frac{\partial J(\bm \beta)}{\partial \beta_j}$ can be reformulated as $\faqai$
expressions and computed without materializing $\bm G$. We first rewrite the
objective:
\begin{align} \label{eqn:rewrite:obj:hinge}
  &\sum_{(\bm x,y)\in\bm G}
    \max\{0,1- y\cdot(\bm\beta^\top \bm x)\} + \frac{\lambda}{2}\norm{\bm\beta}^2_2  \\
    &\qquad = \frac{\lambda}{2}\norm{\bm\beta}^2_2+ \sum_{(\bm x,y)\in\bm G} (1- y\cdot(\bm\beta^\top \bm x)) \cdot
     \bm 1_{y\cdot(\bm\beta^\top \bm x) \leq 1} \\
   & \qquad= \frac{\lambda}{2}\norm{\bm\beta}^2_2+ \sum_{(\bm x,y)\in\bm G}
     (1-(\bm\beta^\top \bm x)) \cdot \bm 1_{y=1} \cdot \bm 1_{\bm\beta^\top \bm x \leq 1}+ \sum_{(\bm x,y)\in\bm G}
     (1+(\bm\beta^\top \bm x)) \cdot \bm 1_{y=-1} \cdot \bm 1_{\bm\beta^\top \bm x \geq -1}\nonumber\\
    &\qquad= \frac{\lambda}{2}\norm{\bm\beta}^2_2 +
\!\!\!\!\underbrace{\sum_{(\bm x,y)\in\bm G}
    \bm 1_{y=1} \cdot \bm 1_{\bm\beta^\top \bm x \leq 1}}_{\text{$\faqai$ of the
       form~\eqref{eqn:our:query}}} - \sum_{i=1}^n \underbrace{\sum_{(\bm
x,y)\in\bm G}
    \beta_i \cdot x_i \cdot \bm 1_{y=1} \cdot \bm 1_{\bm\beta^\top \bm x \leq 1}}
    _{\text{$\faqai$ of the form~\eqref{eqn:our:query}}} \nonumber\\
     &\qquad\qquad+  \underbrace{\sum_{(\bm x,y)\in\bm G}
    \bm 1_{y=-1} \cdot \bm 1_{\bm\beta^\top \bm x \geq -1}}_{\text{$\faqai$ of the
    form~\eqref{eqn:our:query}}}  +  \sum_{i=1}^n  \underbrace{\sum_{(\bm
x,y)\in\bm G}
   \beta_i \cdot x_i \cdot \bm 1_{y=-1}  \cdot \bm 1_{\bm\beta^\top \bm x \geq -1}}
   _{\text{$\faqai$ of the form~\eqref{eqn:our:query}}}. \nonumber
\end{align}
In the above, the sum $\sum_{(\bm x,y)\in\bm G}
\bm 1_{y=1} \cdot \bm 1_{\bm\beta^\top \bm x \leq 1}$
for example can be expressed as an $\faqai$ query of the form \eqref{eqn:our:query} as follows:
\begin{align}
    Q() &= \sum_{y, \bm x_{\calV}}
    \underbrace{\bm 1_{y=1} \cdot \bm 1_{\bm\beta^\top \bm x \leq 1}}_{\text{ligaments } \calE_\ell}\cdot
    \left(\prod_{F\in \calE_s} R_F(\bm x_F)\right).\nonumber
\end{align}
$\frac{\partial J(\bm\beta)}{\partial \beta_j}$ can also be rewritten into
two $\faqai$ expressions:
\begin{align}
  &-\sum_{(\bm x,y)\in\bm G} y\cdot x_j \cdot \bm 1_{
    y(\bm\beta^\top \bm x) \leq 1} + \lambda\beta_j
     = \lambda\beta_j - \underbrace{\sum_{(\bm x,y)\in\bm G}
     x_j \cdot \bm 1_{y=1} \cdot \bm 1_{\bm\beta^\top \bm x \leq
    1}}_{\text{$\faqai$ of the form~\eqref{eqn:our:query}}} +
    \underbrace{\sum_{(\bm x,y)\in\bm G}
    x_j \cdot \bm 1_{y=-1} \cdot \bm 1_{\bm\beta^\top \bm x \geq
    -1}}_{\text{$\faqai$ of the form~\eqref{eqn:our:query}}}. \nonumber
\end{align}

\bthm
\label{thm:hinge:loss}
Let $\bm I$ be an input database where $N$ is the largest relation in $\bm I$, and
$Q$ be a feature extraction query.  For
any linear SVM classification model $\bm\beta^\top\bm x$, the objective
$J(\bm\beta)$ and gradient $\grad J(\bm\beta)$ with hinge loss can be computed
in time $\tilde O(N^{\ssubfaqw_\ell(Q_\ell)})$ with $\spanda$
and in time $O(N^{\faqw_\ell(Q_\ell)} \log N)$ with
$\InsideOut$,
where $Q_\ell$ is the ligament extension of $Q$ (Def.~\ref{defn:Q_ell}).
\ethm

\bp
Let $n$ be the number of variables in $Q$.
We show above that $J(\bm\beta)$ and
$\grad J(\bm\beta)$ can be rewritten into $O(n)$ $\faqai$ expressions with a
single ligament hyperedge (i.e. $|\calE_\ell| = 1$). The overall runtime bound
for computing $J(\bm\beta)$ and $\grad J(\bm\beta)$ with $\spanda$
follows from Theorem~\ref{thm:relaxed:subfaqw:general}, which states that
$\spanda$ can compute each $\faqai$ query in time
$\tilde O(N^{\ssubfaqw_\ell(Q_\ell)})$. The runtime for computing $J(\bm\beta)$
and $\grad J(\bm\beta)$ with $\InsideOut$ follows from
Theorem~\ref{thm:relaxed:faqw}: This is $O(N^{\faqw_\ell(Q_\ell)} \cdot \log N)$ for a
$\faqai$ query $Q$.
\ep

\subsubsection{Cutting-plane algorithm for linear SVM classification in primal space}
\label{sec:svm:primal}

An alternative to learning linear SVM via subgradient-based
optimization is to pose the problem as a constraint optimization
problem. The equivalent formulation for minimizing the
objective~\eqref{eqn:hinge:objective} is the primal formulation of linear
SVM~\cite{murphy2013}:
\begin{align}
  \min_{\bm \beta, \xi_{\bm x, y}\geq 0}
  &\qquad \frac{1}{2} \norm{\bm\beta}^2 + \frac{C}{|G|} \sum_{(\bm x,y)\in\bm G}\xi_{\bm
    x,y} \label{eqn:svm:primal}\\
  \text{s.t.}
  &\qquad y \cdot f_{\bm\beta}(\bm x) \geq 1-\xi_{\bm x,y},
  && \forall (\bm x,y) \in \bm G.\nonumber
\end{align}
where $\xi_{\bm x,y}$ are slack variables and $C$ is the regularization
parameter.

The optimization problem solves for the hyperplane $f_{\bm\beta}(\bm x)$ that
classifies the data points $(\bm x, y) \in \bm G$ into two classes, so that the
margin between the hyperplane and the nearest data point for each class is
maximized. For each $(\bm x, y) \in \bm G$, the slack variable $\xi_{\bm x,y}$
encodes how much the point violates the margin of the hyperplane.

Joachims' cutting-plane algorithm
solves~\eqref{eqn:svm:primal} in linear time over the training
dataset~\cite{Joachims:2006}. The algorithm solves the following {\em structural
  classification} SVM formulation, which is equivalent
to~\eqref{eqn:svm:primal}:
\begin{align}
  \min_{\bm\beta, \xi \geq 0} &\qquad \frac 1 2 \norm{\bm\beta}^2 + C \xi
         \label{eqn:svm:primal:structural}\\
  \text{s.t.}
       &\qquad \frac{1}{|\bm G|} \inner{\bm\beta, \sum_{(\bm x, y) \in \bm T} y \bm x} \geq
         \frac{1}{|\bm G|} |\bm T| - \xi, && \forall \bm T \subseteq \bm G. \nonumber
\end{align}
This formulation has $2^{|\bm G|}$ constraints, one for each possible subset
$\bm T \subseteq \bm G$, and a single slack variable $\xi$ that is shared by all
constraints.

\begin{algorithm}[t]
  \caption{Training classification SVM via~\eqref{eqn:svm:primal:structural}}
  \label{algo:cutting:plane}
  $\calW \gets \emptyset$\tcp*[r]{Working set}
  $t \gets 0$\;
  \Repeat {
    $\frac{|\bm T^{(t)}|}{|\bm G|} - \frac{1}{|\bm G|} \inner{\bm\beta^{(t)}, \sum_{(\bm x,y) \in \bm
        T^{(t)}} y\bm x}
    \leq \xi^{(t)} + \epsilon$
  }
  {
    $t \gets t+1$\;
    $(\bm\beta^{(t)},\xi^{(t)}) \gets \argmin_{\bm\beta,\xi\geq 0}
    \left\{ \frac 1 2 \norm{\bm\beta}_2^2+C\xi \right.$\;
    \qquad $\text{ s.t. } \left.  \frac{1}{|\bm G|} \inner{\bm\beta, \sum_{(\bm x,y) \in \bm T}y\bm x} \geq
      \frac{|\bm T|}{|\bm G|}-\xi, \forall \bm T \in \calW \right\}$\;
    $\bm T^{(t)} := \{ (\bm x, y) \in \bm G \suchthat y\inner{\bm\beta^{(t)},\bm x}<1\}$\;
    $\calW \gets \calW \cup \{\bm T^{(t)}\}$
  }
\end{algorithm}

Algorithm~\ref{algo:cutting:plane} presents Joachims' cutting-plane algorithm
for solving~\eqref{eqn:svm:primal:structural}. It iteratively
constructs a set of constraints $\calW$, which is a subset of all constraints
in~\eqref{eqn:svm:primal:structural}. In each round $t$, it first
computes the optimal value for $\bm\beta^{(t)}$ and $\xi^{(t)}$ over the current
working set $\calW$. Then, it identifies the constraint $\bm T^{(t)}$ that is
most violated for the current $\bm\beta^{(t)}$, and adds this constraint to
$\calW$. It continues until $\bm T^{(t)}$ is violated by at most
$\epsilon$. Joachims showed that Algorithm~\ref{algo:cutting:plane} finds the
$\epsilon$-approximate solution to~\eqref{eqn:svm:primal:structural} in
$O(1)$-many iterations~\cite{Joachims:2006}. Hence $|\calW|$ and the
number of constraints of the optimization problem are bounded by a
number {\em independent} of $|\bm G|$.

Next, we consider the inner optimization problem at line 5. Although $|\calW|$
is small, the number $n$ of variables can still be large. This prohibits
solving with quadratic programming as it can take up to
$O(n^3)$~\cite{murphy2013}.  Its Wolfe dual, on the other hand, is a
quadratic program with only a constant number of variables that is independent
of $n$ and one constraint.  Let
$\bm x_{\bm T} = \sum_{(\bm x, y) \in \bm T} y\bm x$. We next present the derived Wolfe dual.

\subsubsection*{Wolfe dual for optimization problem at line 5 of
	Algorithm~\ref{algo:cutting:plane}}
\label{sec::wolfedual}

We consider the inner optimization problem at line 5 of
Algorithm~\ref{algo:cutting:plane}, show how to derive the Wolfe
dual~\eqref{eqn:svm:dual} from the structural SVM classification
formulation~\eqref{eqn:svm:primal:structural}. Let
$\bm x_{\bm T} = \sum_{(\bm x, y) \in \bm T} y\bm x$. The inner optimization
problem at line 5 of Algorithm~\ref{algo:cutting:plane} is of the form:
\begin{align}
	\min_{\bm\beta, \xi}
	& \qquad \frac 1 2 \norm{\bm\beta}^2 + C \xi \label{eqn:svm:primal:2}\\
	\text{s. t.}
	& \qquad \inner{ \bm\beta, \bm x_{\bm T}} \geq |\bm T| - |\bm G|\xi
	&& \bm T \in \calW\nonumber\\
	& \qquad \xi \geq 0.\nonumber
\end{align}

The Lagrangian function of this optimization problem is:
\begin{align*}
	&L(\bm\beta, \xi, \bm \alpha, \gamma) \\
	&=  \frac 1 2 \norm{\bm\beta}^2 + C \xi
	+ \sum_{\bm T \in \calW} \alpha_{\bm T}(|\bm T| - |\bm G|\xi -
	\inner{\bm\beta,\bm x_{\bm T}}) - \gamma\xi \\
	&=  \frac 1 2 \norm{\bm\beta}^2 -
	\inner{ \bm\beta, \sum_{\bm T\in \calW} \alpha_{\bm T} \bm x_{\bm T}} +
	\sum_{\bm T \in \calW} |\bm T| \alpha_{\bm T}+\left(C - |\bm G| \sum_{\bm T \in \calW} \alpha_{\bm T} -\gamma \right)\xi.
\end{align*}
where $\bm\alpha = (\alpha_{\bm T})_{\bm T \in \mathcal W}$ and $\gamma$ are
Lagrange multipliers.

Since the Lagrangian is convex and continuously differentiable, we can define
the Wolfe dual as the following optimization problem:
\begin{align}
	\max_{\bm\beta, \xi}
	& \qquad L(\bm\beta, \xi, \bm \alpha, \gamma) \\
	\text{s. t.}
	& \qquad \grad_{\bm \beta} L = \bm\beta - \sum_{\bm T \in \calW}
	\alpha_{\bm T}\bm x_{\bm T} = 0 \nonumber\\
	& \qquad \grad_{\xi}L = C - |\bm G| \sum_{\bm T\in \calW}\alpha_{\bm T} -\gamma = 0
	\nonumber \\
	& \qquad \bm\alpha \geq \bm 0 ,\gamma \geq 0. \nonumber
\end{align}

The optimal condition for $\bm\beta$ is
$\bm\beta = \sum_{\bm T \in \calW} \alpha_{\bm T}\bm x_{\bm T} $. We use this
equality to rewrite the above dual formulation and attain the following optimization
problem:
\begin{align}
  \max_{\bm \alpha \geq \bm 0}
  & \quad -\frac 1 2 \inner{
    \sum_{\bm T \in \calW}\alpha_{\bm T}\bm x_{\bm T},
    \sum_{\bm T \in \calW}\alpha_{\bm T}\bm x_{\bm T}}
    + \sum_{\bm T \in \calW} |\bm T| \alpha_{\bm T} \label{eqn:svm:dual}\\
  \text{s.t.}
  & \quad \sum_{\bm T \in \calW} \alpha_{\bm T} \leq \frac{C}{|\bm G|}\nonumber
\end{align}
where $\bm\alpha = (\alpha_{\bm T})_{\bm T \in \mathcal W}$ is the vector of
constraints.

\bthm
\label{thm:cutting:plane}
Let $\bm I$ be an input database where $N$ is the largest relation in $\bm I$, and
$Q$ be a feature extraction query. A linear SVM classification model can be
learned over the training dataset $Q(I)$ with Joachims' cutting-plane algorithm
in time $\tilde O(N^{\ssubfaqw_\ell(Q_\ell)})$ with $\spanda$
and in time $O(N^{\faqw_\ell(Q_\ell)} \log N)$ with $\InsideOut$,
where $Q_\ell$ is the ligament extension of $Q$ (Def.~\ref{defn:Q_ell}).
\ethm

\bp

Recall that for each iteration $t$ of Algorithm~\ref{algo:cutting:plane}, we add
one set $\bm T^{(t)}$ to $\calW$, and $\bm T^{(t)}$ is associated with a coefficient
vector $\bm\beta^{(t)}$. Our main observation is that we do not have to
materialize the set $\bm T^{(t)}$, since it is completely determined by the data
and the coefficient vector $\bm\beta^{(t)}$. Thus, instead of storing
$\bm T^{(t)}$ we can simply store $\bm\beta^{(t)}$ and reformulate the data
dependent term $\bm x_{T^{(t)}}$ in~\eqref{eqn:svm:dual} as a computation over
$\bm G$:
\begin{align*}
	\forall \bm T^{(t)} \in \mathcal W \;:\;\;\; \bm x_{\bm T^{(t)}} =
	\sum_{(\bm x, y) \in \bm T^{(t)}} y \bm x = \sum_{(\bm x,y) \in \bm G} y \bm x \cdot \bm
	1_{y\inner{\bm\beta^{(t)}, \bm x} < 1}.
\end{align*}

The vector $\bm x_{T^{(t)}}$ has size $n$. For each $j\in [n]$, we can compute
the $j$'th component of $\bm x_{T^{(t)}}$ as the summation of the following two
$\faqai$ expressions, which are of form~\eqref{eqn:our:query}:
\begin{align*}
	Q_1() &= \sum_{\bm x_\calV,y} y \cdot x_j \cdot \bm 1_{y = 1}
	\cdot \bm 1_{\sum_{j \in [n]} \beta_j^{(t)} \cdot x_j < 1} \cdot
	\left(\prod_{F \in \calE_s} R_F(\bm x_F)\right),\\
	Q_2() &= \sum_{\bm x_\calV,y} y \cdot  x_j  \cdot  \bm 1_{y = -1}
	\cdot  \bm 1_{\sum_{j \in [n]} \beta_j^{(t)}  \cdot
		 x_j > -1}  \cdot
	\left(\prod_{F \in \calE_s}  R_F(\bm
	x_F)\right).
\end{align*}

$Q_1$ and $Q_2$ have a single ligament hyperedge (i.e. $|\calE_\ell| = 1$).
Theorem~\ref{thm:relaxed:subfaqw:general} states that $\spanda$
computes $Q_i$ for $i\in[2]$ in time
$\tilde O(N^{\ssubfaqw_\ell(Q_i)})$. Consequently, the optimization problem at line
5 of Algorithm~\ref{algo:cutting:plane} can be computed in time
$\tilde O(N^{\ssubfaqw_\ell(Q_i)})$. This
determines the runtime of Algorithm~\ref{algo:cutting:plane}.

Using $\InsideOut$, the runtime of Algorithm~\ref{algo:cutting:plane}
follows from Theorem~\ref{thm:relaxed:faqw}: This is $O(N^{\faqw_\ell(Q_i)} \log N)$ for $Q_i$.
\ep

\subsection{$k$-means clustering}
\label{sec:k:means}

We next consider $k$-means clustering, which is a popular unsupervised machine learning algorithm.

An unsupervised machine learning model is computed over a dataset
$\bm G \subseteq \R^n$, for which each tuple $\bm x \in \bm G$ is a vector of
features without a label. A clustering task divides $\bm G$ into $k$
clusters of ``similar'' data points with respect to the $\ell_2$ norm:
$\bm G = \cup_{i=1}^{k} \bm G_i$, where $k$ is a given fixed positive
integer. Each cluster $\bm G_i$ is represented by a cluster mean
$\bm \mu_i \in \mathbb{R}^n$. One of the most ubiquitous clustering methods,
Lloyd's $k$-means clustering algorithm (also known as the $k$-means method),
involves the optimization problem~\eqref{eq:kmeans-optimization}
with respect to the partition
$(\bm G_i)_{i \in [k]}$ and the $k$ means $(\bm \mu_i)_{i\in [k]}$.
Other norms or distance measures can be used, e.g., if we replace $\ell_2$ with
$\ell_1$-norm, then we get the $k$-median problem. The subsequent development
considers the $\ell_2$-norm.

Lloyd's algorithm can be viewed as a special instantiation of the {\em
  Expectation-Maximization} (EM) algorithm. It iteratively computes two updating
steps until convergence. First, it updates the cluster assignments for each
$(\bm G_i)_{i \in [k]}$:
\begin{align}
    \bm G_i &= \left\{
        \bm x \in \bm G \suchthat \norm{\bm x-\bm \mu_i}^2 \leq
            \norm{\bm x-\bm \mu_j}^2, \forall j \in [k]\setminus\{i\}
        \right\}
        \label{eqn:kmeans:update1}
\end{align}
and then it updates the corresponding $k$-means $(\bm \mu_i)_{i \in [k]}$:
\begin{align}
    \bm\mu_i &= \frac{1}{|\bm G_i|} \sum_{\bm x \in \bm G_i} \bm x.
    \label{eqn:kmeans:update2}
\end{align}

Our observation is that we can reformulate both update steps~\eqref{eqn:kmeans:update1} and~\eqref{eqn:kmeans:update2} as
$\faqai$ expressions, without explicitly computing the partitioning
$(\bm G_i)_{i\in [k]}$. For a given set of $k$-means $(\bm\mu_j)_{j \in [k]}$,
let $c_{ij}(\bm x)$ be the following function:
\begin{align}
  c_{ij}(\bm x) &=
  \sum_{l \in [n]}[(x_l - \mu_{i,l})^2 - (x_l - \mu_{j,l})^2]
  = \sum_{l \in [n]} [\mu_{i,l}^2 - 2  x_l (\mu_{i,l} + \mu_{j,l})
    - \mu_{j,l}^2].
\end{align}
where $\mu_{j,l}$ is the $l$'th component of mean vector $\bm\mu_j$. A data
point $\bm x \in \bm G$ is closest to center $\bm\mu_i$ if and only if
$c_{ij}(\bm x) \leq 0$ holds $\forall j \in [k]$.
We use this inequality to reformulate the mean vector
$\bm\mu_i$ as $O(n)$
$\faqai$ expressions. First, we express $|\bm G_i|$ as:
\begin{align}
  Q_i() = \sum_{\bm x} \left(\prod_{j \in [k]} \bm 1_{c_{ij}(\bm x) \leq 0} \right)
  \left( \prod_{F \in \calE_s} R_F(\bm x_F) \right).
\end{align}
Then, for each $l \in [n]$, the sum $\sum_{\bm x \in \bm G_i} x_l$ can be
reformulated in $\faqai$ as follows (similarly to \eqref{eq:kmeans-query}):
\begin{align}
  Q_{il}() = \sum_{\bm x} x_l \left(\prod_{j \in [k]} \bm 1_{c_{ij}(\bm x) \leq 0}
  \right)\left( \prod_{F \in \calE_s} R_F(\bm x_F) \right).
\end{align}
Each component $(\mu_{i,l})_{l \in [n]}$ equals the division of
$Q_{il}$ by $Q_{i}$.

Overall, the mean vector $\bm\mu_i$ can be computed with $O(n)$
$\faqai$ expressions of the form~\eqref{eqn:our:query}.

\bthm
\label{thm:k:means}
Let $\bm I$ be an input database where $N$ is the largest relation in $\bm I$, and
$Q$ be a feature extraction query where $n$ is the number of its variables.
Each iteration of Lloyd's $k$-means algorithm can be computed in time
$\tilde O(N^{\ssubfaqw_\ell(Q_\ell)})$ with $\spanda$ and
in time $O(N^{\faqw_\ell(Q_\ell)} \log^{k-1} N)$ with $\InsideOut$,
where $Q_\ell$ is the ligament extension of $Q$ (Def.~\ref{defn:Q_ell}).
\ethm

\bp

We have shown above that each mean vector
$(\mu_j)_{j \in [k]}$ can be computed with $O(n)$ $\faqai$ expression of the
form~\eqref{eqn:our:query}, where each query has $|\calE_\ell| = k$ ligament
hyperedges. For $\spanda$, the overall runtime to update all
$k$-means follows from Theorem~\ref{thm:relaxed:subfaqw:general} (respectively
Theorem~\ref{thm:relaxed:faqw}), which states that the algorithm can compute
each $\faqai$ expression of form~\eqref{eqn:our:query} in time
$\tilde O(N^{\ssubfaqw_\ell(Q_\ell)})$.  Using $\InsideOut$, the runtime follows from
Theorem~\ref{thm:relaxed:faqw}: Any $\faqai$ query $Q$ of form~\eqref{eqn:our:query} can be computed in time
$O( N^{\faqw_\ell(Q_\ell)} \log^{k-1} N)$.
\ep

\section{Conclusion}
We presented a theoretical and algorithmic framework for solving a special class of functional aggregate queries that arise naturally within many in-database machine learning problems and captures a variety of database queries including inequality joins.
In this query class, called $\faqai$, some of the input factors happen to be additive inequalities over some input variables. We showed that $\faqai$ queries can be solved
more efficiently than general $\faq$ queries by relaxing the notion of tree decompositions leading to relaxed versions of commonly used width parameters.

While $\faq$ queries over the Boolean semiring are solvable within the tighter bound of submodular width~\cite{MR3144912, panda-pods}, such a bound is not known to be achievable
over arbitrary semirings, including count queries. Therefore, we first introduced a counting analog of the submodular width, denoted $\ssubw$, by relaxing the notion of polymatroids, and showed how to meet this bound for $\faq$ queries over any semiring. We then turned our attention back to the special case of $\faqai$ and showed how to strengthen the bound further in this case.

We showed how to use our framework to solve several common machine learning problems over relational data asymptotically faster than both out-of-database  and previously known in-database machine learning solutions.
These problems include $k$-means clustering, support vector machines, and regression over a variety of non-polynomial loss functions.

One interesting open problem is to prove a hardness result for count queries with unbounded $\ssubw$. On one hand, this would show the tightness of our positive result for solving $\faq$ queries over arbitrary semirings within $\ssubw$ bound.
On another, this would mirror the previously known dichotomy result for query classes over the Boolean semiring based on the submodular width~\cite{MR3144912}.

Another remaining problem is to measure the gap between the submodular width and its counting version $\ssubw$. More precisely, is there a class of queries where the submodular width is unboundedly smaller than $\ssubw$?

Marx~\cite{MR3144912} showed a class of queries where the submodular width is bounded while the fractional hypertree width is unbounded.
Proposition~\ref{prop:ssubw-fhtw} showed a class of queries where the gap between $\ssubw$ and the fractional hypertree width is unbounded (but $\ssubw$ is also unbounded). It remains open to show whether there exists a query class where $\ssubw$ is bounded and the fractional hypertree width is unbounded.

While the $\faqai$ framework can be used to optimize machine learning problems over several
non-polynomial loss functions including those presented in Section~\ref{sec:huber:loss}
and~\ref{sec:otherlosses},
other classes of loss functions are not representable as $\faqai$ queries
and do not benefit from this framework yet.
These classes include for example
the logistic and exponential losses commonly used for classification problems.
It would be interesting to see if such loss functions could eventually be optimized
in the same way in the in-database machine learning setting.

\section*{Acknowledgments}
This project has received funding from the European Union's Horizon 2020 research and innovation programme under grant agreement No 682588.
LN gratefully acknowledges support from NSF grants CAREER DMS-1351362 and CNS-1409303, Adobe Research and Toyota Research, and a Margaret and Herman Sokol Faculty Award.
BM's was supported in part by a Google Research Award, and NSF grants CCF-1830711, CCF-1824303, and CCF-1733873.

\bibliographystyle{acm}
\bibliography{main}

\appendix

\section{Recovering Two Existing Results}
\label{sec:applications-db}

In this section we review two prior results concerned with the evaluation of queries with inequalities: the evaluation of Core XPath queries over XML documents via relational encoding in the pre/post plane and the exact inference for IQ queries with inequality joins over probabilistic databases.
Our main observation is that their linearithmic complexity is due to the same structural property behind relaxed tree decompositions: Such queries admit trivially a relaxed tree decomposition, where each bag corresponds to one relation in the query and the ligament edges, i.e., the inequality joins, are covered by neighboring bags.

\subsection{Core XPath Queries}

We consider the problem of evaluating Core XPath queries over XML documents.
An XML document is represented as a rooted tree whose nodes follow the {\em document order}. Core XPath queries define traversals of such trees using two constructs: (1) a {\em context node} that is the starting point of the traversal; and (2) a tree of {\em location steps} with one distinguished branch that selects nodes and all other branches conditioning this selection. Given a context node $v$, a location step selects a set of the nodes in the tree that are accessible from $v$ via the step's {\em axis}. This set of nodes provides the context nodes for the next step, which is evaluated for each such node in turn. The result of the location step is the set of nodes accessible from any of its input context nodes, sorted in document order.

The {\em preorder rank} $pre(v)$ of a node $v$ is the index of $v$ in the list of all nodes in the tree that are visited in the (depth-first, left-to-right) preorder traversal of the tree; this order is the document order. Similarly, the {\em postorder rank} $post(v)$ of $v$ is its index in the list of all nodes in the tree that are visited in the (depth-first, left-to-right) postorder tree traversal. We can use the pre/post-order ranks of nodes to define the main axes \texttt{descendant}, \texttt{ancestor}, \texttt{following}, and \texttt{preceding}~\cite{Grust:2002}. Given two nodes $v$ and $v'$ in the tree, the four axes are defined using the pre/post two-dimensional plane:
\begin{itemize}
\item $v'$ is a descendant of $v$ or equivalently $v$ is an ancestor of $v'$
\[
\text{iff } pre(v)<pre(v')\wedge post(v')<post(v)
\]
\item $v'$ follows $v$ or equivalently $v$ precedes $v'$
\[
\text{iff } pre(v)<pre(v')\wedge post(v)<post(v')
\]
\end{itemize}

The remaining axes \texttt{parent}, \texttt{child}, \texttt{following-sibling}, and \texttt{preceding-sibling} are restrictions of the four main axes, where we also use the parent information $par$ for each node:
\begin{itemize}
\item $v'$ is a child of $v$ or equivalently $v$ is a parent of $v'$
\[
\text{iff } v = par(v')
\]
\item $v'$ is a following sibling of $v$ or equivalently $v$ is a preceding sibling of $v'$
\[
\text{iff } pre(v)<pre(v')\wedge post(v)<post(v')\wedge par(v) = par(v')
\]
\end{itemize}

We follow the standard approach to reformulate XPath evaluation in the relational domain~\cite{Grust:2002}. We represent the document by a factor $\bm G$ in the Boolean semiring with schema $(pre,post,par,tag)$. For each node in the tree there is one tuple in $\bm G$ with $pre$ and $post$ ranks, label $tag$, and preorder rank $par$ of the parent node. A query with $n$ location steps is mapped to an $\faqai$ expression $Q$ that is a join of $n+1$ copies of $\bm G$ where the join conditions are the inequalities encoding the axes of the $n$ steps. The first copy $\bm G_0$ is for the initial context node(s). The axis of the $i$-th step is translated into the conjunction of inequalities between pre/post rank variables of the copies $\bm G_{i-1}$ and $\bm G_i$. The query $Q$ has one free variable: This is the preorder rank variable from the copy of $\bm G$ corresponding to the location step that selects the result nodes.

\begin{example}
The Core XPath query
\begin{align*}
    v/\text{descendant}::a[\text{descendant}::c]/\text{following}::b
\end{align*}
selects all $b$-labeled nodes following $a$-labeled nodes that are descendants of the given context node $v$ and that have at least one $c$-labeled descendant node. The steps in the above textual representation of the query are separated by \text{/}. The brackets \text{[ ]} delimit a condition on the selection of the $a$-labeled nodes. We can reformulate this query in $\faqai$ over the Boolean semiring as follows:
\begin{align*}
    Q(pre_b)\leftarrow
    &\quad\bm G_v(pre_v,post_v,p_v,tag_v)\wedge \bm G_a(pre_a,post_a,p_a,\text{'a'}) \wedge\\
    &\quad\bm G_c(pre_c,post_c,p_c,\text{'c'})\wedge \bm G_b(pre_b,post_b,p_b,\text{'b'}) \wedge\\
    &\quad pre_v < pre_a\wedge post_a < post_v \wedge \hspace*{0.35em}\text{// a is descendant of v}\\
    &\quad pre_a < pre_c \wedge post_c < post_a \wedge \hspace*{0.55em}\text{// c is descendant of a}\\
    &\quad pre_a < pre_b \wedge post_a < post_b  \hspace*{1.5em}\text{// b is following a} \qed
\end{align*}
\end{example}
The hypergraph of a relational encoding of a Core XPath query has one skeleton hyperedge for each copy of the document factor and one ligament edge for each pair of inequalities over two of these copies. Any two skeleton hyperedges may only have one node, i.e., query variable, in common to express the parent/child or sibling relationship between their corresponding steps. This hypergraph admits a trivial relaxed tree decomposition, which mirrors the tree structure of the query. In particular, there is one bag of the decomposition consisting of the variables of each copy of the document factor. Each ligament edge represents a pair of inequalities over variables of two neighboring bags. The running intersection property holds since the equalities are by construction only over variables from neighboring bags.

It is known that the time complexity of answering a Core XPath query $Q$ with $n$ location steps over an XML document $\bm G$ is $O(n\cdot|\bm G|)$ (Theorem 8.5~\cite{GKP:XPath:2002}; it assumes the document factor sorted). We can show a linearithmic time complexity result using our $\faqai$ reformulation of Core XPath queries and the trivial relaxed tree decomposition.

\nop{
; a more refined complexity analysis shows that Core XPath is LOGCFL-complete and without conditions it is complete for nondeterministic logarithmic space~\cite{GKP:XPath:2003}
}

\begin{proposition}
For any Core XPath query $Q$ with $n$ location steps and XML document $\bm G$, the query answer can be computed in time $O(n\cdot|\bm G|\cdot\log|\bm G|)$.
\end{proposition}

\begin{proof}
Let $\varphi$ be the $\faqai$ reformulation of $Q$ and $F$ the factor representing the XML document $\bm G$.
There is a one-to-one correspondence between the trivial relaxed tree decomposition and the XPath query, with one bag per location step.
Let $n$ be the number of location steps in $Q$, or equivalently the number of bags in the tree decomposition.
We consider this trivial tree decomposition and choose its root as the bag corresponding to the location step that selects the answer node set. Our evaluation algorithm proceeds in a bottom-up left-to-right traversal of the tree decomposition and eliminates one bag at a time.

We index the bags and their corresponding factors in this traversal order. The first factor to eliminate is then denoted by $F_1$ while the last factor, which corresponds to the location step selecting the answer node set, is denoted by $F_n$.

We initially create factors $S_j$ that are copies of factors $F_j$ corresponding to leaf bags in the tree. Consider now two factors $S_j$ and $F_i$ corresponding to a leaf bag and respectively to its parent bag. Let $\phi_{i,j}$ be the conjunction of inequalities defining the axis relationship between the location steps corresponding to these bags. We then compute a new factor $S_i$ that consists of those tuples in $F_i$ that join with some tuples in $S_j$. This is expressed in $\faqai$ over the Boolean semiring:
\begin{align*}
S_i(pre_i,post_i,p_i,t_i) \quad\leftarrow\quad F_i(pre_i,post_i,p_i,t_i)\wedge S_j(pre_j,post_j,p_j,t_j) \wedge \phi_{i,j}
\end{align*}
The conjunction $\phi_{i,j}$ only has two inequalities on variables between the two bags. Computing $S_i$ takes time $O(|F|\log |F|)$ following the algorithm from the proof of Theorem~\ref{thm:relaxed:faqw}. We can sort both $F_i$ and $S_j$ in ascending order on the preorder column and in descending order on the postorder column. For each tuple $t$ in $F_i$, the tuples in $S_j$ that join with $t$ form a contiguous range in $S_j$. To assert whether $t$ is in $S_i$, it suffices to check that this range is not empty. There are $n$ such steps and $|F|=|F_i|=|\bm G|$, with an overall time complexity of $O(n \cdot |\bm G|\log |\bm G|)$.
\end{proof}

\subsection{Probabilistic Queries with Inequalities}

The problem of query evaluation in probabilistic databases is \#P-hard for general queries and probabilistic database formalisms~\cite{SORK:PDB:2011}. Extensive prior work focused on charting the tractability frontier of this problem, with positive results for several classes of queries on so-called tuple-independent probabilistic databases. We discuss here one such class of queries with inequality joins called IQ~\cite{OH09}.

A tuple-independent probabilistic database is a database where each tuple $t$ is associated with a Boolean random variable $v(t)$ that is independent of the other tuples in the database. This is the database formalism of choice for studies on query tractability since inference is hard already for trivial queries on more expressive probabilistic database formalisms~\cite{SORK:PDB:2011}.

FAQ factors naturally capture tuple-independent probabilistic databases: A tuple-independent probabilistic relation $R$ is a factor that maps each tuple $t$ in $R$ to the probability that the associated random variable $v(t)$ is true.

We next define the class IQ of inequality queries and later show how to recover the linearithmic time complexity for their inference.

\begin{definition}[adapted from Definitions 3.1, 3.2~\cite{OH09}]
Let a hypergraph $\calH=(\calV=[n], \calE_s \cup \calE_\ell)$, where
$\calE_s$ and $\calE_\ell$ are disjoint,
$\calE_s$ consists of pairwise disjoint sets,
$\calE_\ell$ consists of sets $\{i,j\}$
for which there is a vector $c_{i,j}\in\{[1,-1]^{\text{T}},[-1,1]^{\text{T}}\}$,
and
$\forall F\in\calE_s: |(\bigcup_{I\in\calE_\ell} I) \cap F|\leq 1$.
An IQ query has the form
\begin{equation}
Q() \leftarrow \bigwedge_{F\in\calE_s} R_F(\bm X_F)\wedge \bigwedge_{\{i,j\}\in\calE_\ell}
[X_i,X_j]^{\text{T}}\cdot c_{i,j}\leq 0
\label{eq:iq}
\end{equation}
where $(R_F)_{F\in\calE_s}$ are distinct factors.
\qed
\end{definition}

The edges (i.e., binary hyperedges) in $\calE_\ell$ correspond to inequalities of the query variables. These inequalities are restricted so that there is at most one node (query variable) from any hyperedge in $\calE_s$. Inequalities on variables of the same factor are not in $\calE_\ell$; they can be computed trivially in a pre-processing step.

The inequalities may only have the form $X_i\leq X_j$ or $X_j\leq X_i$.
They induce an {\em inequality graph} where $X_i$ is a parent of $X_j$ if $X_i\leq X_j$.
This graph can be minimized by removing edges corresponding to redundant inequalities implied by other inequalities~\cite{IO:INEQSAT:1997}.  Each graph node thus corresponds to precisely one factor. We categorize the IQ queries based on the structural complexity of their inequality graphs into (forests of) paths, trees, and graphs.

\begin{example}\label{ex:iq-queries}
Consider the following IQ queries:
\begin{align*}
Q_1() &\rightarrow R(A)\wedge S(B)\wedge T(C)\wedge A\leq B\wedge B\leq C\\
Q_2() &\rightarrow R(A)\wedge S(B)\wedge T(C)\wedge A\leq B\wedge A\leq C
\end{align*}
The inequalities form a path in $Q_1$ and a tree in $Q_2$.
\end{example}

The probability a query over a probabilistic database $\bm I$ is the probability of its {\em lineage}~\cite{SORK:PDB:2011}. The lineage is a propositional formula over the random variables associated with the input tuples. It is equivalent to the disjunction of all possible derivations of the query answer from the input tuples.

\begin{example}\label{ex:iq-lineage}
Consider the factors $R$, $S$, $T$, where $r_i$, $s_j$, $t_k$ denote the variables associated with  the tuples in these factors and for a random variable $a$, $p_a$ denotes the probability that $a=\text{true}$:

\begin{center}
\begin{small}
\begin{tabular}{ccc}
\begin{tabular}{c|c|c}
$R$ & $A$ & \\\hline
$r_1$  & 1 & $p_{r_1}$\\
$r_2$  & 2 & $p_{r_2}$\\
$r_3$  & 3 & $p_{r_3}$
\end{tabular}
&
\begin{tabular}{c|c|c}
$S$ & $B$ & \\\hline
$s_1$  & 2 & $p_{s_1}$\\
$s_2$  & 3 & $p_{s_2}$\\
$s_3$  & 4 & $p_{s_3}$
\end{tabular}
&
\begin{tabular}{c|c|c}
$T$ & $C$ & \\\hline
$t_1$  & 3 & $p_{t_1}$\\
$t_2$  & 4 & $p_{t_2}$\\
$t_3$  & 5 & $p_{t_3}$
\end{tabular}
\end{tabular}
\end{small}
\end{center}

The lineage of $Q_1$ and $Q_2$ over these factors is:
\begin{center}
\begin{minipage}{0.3\textwidth}
\begin{align*}
& r_1 [s_1(t_1+t_2+t_3) + s_2(t_2+t_3)+s_3t_3]+\\
& r_2 [\phantom{s_1(t_1+t_2+t_3) +\;\;} s_2(t_2+t_3)+s_3t_3]+\\
& \underbrace{r_3[\phantom{s_1(t_1+t_2+t_3) + s_2(t_2+t_3)+\;\;} s_3t_3]}_{\text{lineage of } Q_1}
\end{align*}
\end{minipage}
\begin{minipage}{0.2\textwidth}
\begin{align*}
& r_1 (s_1 + s_2 + s_3)(t_1+t_2+t_3)+\\
& r_2 (\phantom{s_1 +\;\; } s_2 + s_3)(t_1+t_2+t_3)+\\
& \underbrace{r_3 (\phantom{s_1 + s_2 +\;\;} s_3)(\phantom{t_1+\;\;} t_2+t_3)}_{\text{lineage of } Q_2}
\end{align*}
\end{minipage}
\end{center}
\end{example}

Prior work (Theorem 4.7~\cite{OH09}) showed that the probability of an IQ query
$Q$ with an inequality tree with $k$ nodes over a tuple-independent
probabilistic database of size $N$ can be computed in time $O(2^k\cdot N\log
N)$ using a construction of the query lineage in an Ordered Binary Decision
Diagram (OBDD). We show next that a variant of the algorithm in the proof of
Lemma~\ref{lmm:two:bags}, adapted from counting to {\em weighted counting},
i.e., probability computation, can compute the probability in time $O(N\log
N)$, thus shaving off an exponential factor in the number of inequalities.

We first explain this result using two examples, which draw on a crucial observation made in prior work~\cite{OH09}: The lineage of IQ queries has a chain structure: For each factor, there is an order on its random variables that defines a chain of logical implications between their cofactors in the lineage: the cofactor of the first variable implies the cofactor of the second variable, which implies the cofactor of the third variable, and so on.

\begin{example}\label{ex:iq-inference}
We continue Example~\ref{ex:iq-lineage}. The lineage of $Q_1$ and $Q_2$ is arranged so that the chain structure becomes apparent. This structure allows for an equivalent rewriting of the lineage~\cite{OH09}, as shown next for the lineage $\phi_{r_1}$ of $Q_1$ (for a random variable $a$, $\overline{a}$ denotes its negation):
\begin{align*}
\phi_{r_i} &= r_i \phi_{s_i} + \overline{r_i}\phi_{r_{i+1}}, \forall i\in[3];\hspace*{2em} \phi_{r_4} = \mbox{false}\\
\phi_{s_j} &= s_j \phi_{t_j} + \overline{s_j}\phi_{s_{j+1}}, \forall j\in[3];\hspace*{2em} \phi_{s_4} = \mbox{false}\\
\phi_{t_k} &= t_k \phantom{\phi_{t_j}} + \overline{t_k}\phi_{t_{k+1}}, \forall k\in[3];\hspace*{1.6em} \phi_{t_4} = \mbox{false}
\end{align*}

In disjunctive normal form, the lineage of $Q_1$ may have size cubic in the size of the database. The factorization of the lineage in Example~\ref{ex:iq-lineage} lowers the size to quadratic. The above rewriting further reduces the size to linear. The rewritten form can be read directly from the input factors following the structure of the inequality tree.

Since the above expressions are sums of two mutually exclusive formulas, their probabilities are the sums of the probabilities of their respective two formulas. Their probabilities can be computed in one bottom-up right-to-left  pass: First for $\phi_{t_k}$ in decreasing order of $k$, then for $\phi_{s_j}$ in decreasing order of $j$, and finally for $\phi_{r_i}$ in decreasing order of $i$. We extend the probability function $p$ from input random variables to formulas over these variables. The probability of $Q_1$'s lineage, which is also the probability of $Q_1$, is ($\forall i, j, k \in [3]$):
\begin{align*}
p(\phi_{r_i}) &= p(r_i)\cdot p(\phi_{s_i}) + [1- p(r_i)]\cdot p(\phi_{r_{i+1}}) \\
p(\phi_{s_j}) &= p(s_j)\cdot p(\phi_{t_j}) + [1-p(s_j)]\cdot p(\phi_{s_{j+1}}) \\
p(\phi_{t_k}) &= p(t_k) \phantom{\cdot p(\phi_{t_j})\;} + [1-p(t_k)]\cdot p(\phi_{t_{k+1}})
\end{align*}
Since there are no variables $r_4$, $s_4$, and $t_4$, we use $p(\phi_{r_4}) = p(\phi_{s_4}) = p(\phi_{t_4}) = 0$. This computation corresponds to a decomposition of $\phi_{r_1}$ that can be captured by a linear-size OBDD~\cite{OH09}.

The probability of the lineage $\psi_{r_1}$ of $Q_2$ is computed similarly ($\forall i, j, k \in [3]$):
\begin{align*}
p(\psi_{r_i}) &= p(r_i)\cdot p(\psi_{s_i})\cdot p(\psi_{t_i}) + [1- p(r_i)]\cdot p(\psi_{r_{i+1}})\\
p(\psi_{s_j}) &= p(s_j)  + [1-p(s_j)]\cdot p(\psi_{s_{j+1}})\\
p(\psi_{t_k}) &= p(t_k) + [1-p(t_k)]\cdot p(\psi_{t_{k+1}})
\end{align*}
This computation would correspond to a decomposition of $\psi_{r_1}$ that can be captured by an OBDD with several nodes for a random variable from $S$ and $T$; in general, such an OBDD would have a size linear in $N$ but with an additional exponential factor in the size of the inequality tree due to the inability to represent succinctly the products of lineage over $T$ and of lineage over $S$~\cite{OH09}. (OBDDs with AND nodes can capture such products without this exponential factor, though in this article we do not use them.)\qed
\end{example}

\begin{proposition}
Given a tuple-independent probabilistic database ${\bm I}$ of size $N$ and an IQ query $Q$ with a forest of inequality trees, we can compute the probability of $Q$ over ${\bm I}$ in time $O(N\log N)$.
\end{proposition}

\begin{proof}
We next present the inference algorithm for a given IQ query $Q$ with an
inequality tree. It uses a minor variant of the algorithm from the proof of
Lemma~\ref{lmm:two:bags} to compute a functional aggregate query with additive
inequalities over two factors.

We first reduce the input database ${\bm I}$ to a simplified database of unary and nullary factors that is constructed by aggregating away all query variables that do not contribute to inequalities.

Let us partition $\calE_s$ into the hyperedges $\calE_1$ that contain query variables involved in inequalities and all other hyperedges $\calE_2$.

We reduce each factor $(R_F)_{F\in\calE_1}$ with a query variable $X_i$ occurring in inequalities to a unary factor $S_{\{i\}}$ by aggregating away all other query variables. For an $X_i$-value $x_i$, $S_{\{i\}}(x_i)$ gives the probability of the disjunction of the independent random variables associated with the tuples in $R_F$ that have the $X_i$-value $x_i$:
\begin{align*}
S_{\{i\}}(x_i) = 1 - \prod_{{\bm x}\in\dom({\bm X}_F-\{i\})} \big(1 - R_F({\bm x}_F)\big)
\end{align*}
We also reduce all factors $(R_F)_{F\in\calE_2}$ with no query variable occurring in inequalities to one nullary factor $S_\emptyset$ by aggregating away all query variables. $S_\emptyset()$ gives the probability of the conjunction of all factors without query variables in inequalities:
\begin{align*}
S_\emptyset() = \prod_{F\in\calE_2l} \big[ 1 - \prod_{{\bm x}\in\dom({\bm X}_F)} \big(1 - R_F({\bm x}_F)\big)\big]
\end{align*}

This simplification reduces the set $\calE_s$ of hyperedges to a new set $\calE_u$ of unary edges, one per query variable in the inequalities, and one nullary edge: $\calE_u=\{\emptyset\}\cup\bigcup_{\{i,j\}\in\calE_\ell}\{\{i\},\{j\}\}$. The simplification does not affect the inference problem: The probability of $Q$ is the same as the probability of the query $Q'$ over $\calE_u\cup\calE_\ell$:
\begin{equation}
Q'() \leftarrow \bigwedge_{F\in\calE_u} S_F(\bm X_F)\wedge
\bigwedge_{\{i,j\}\in\calE_\ell} [X_i,  X_j]^{\text{T}}\cdot c_{i,j} \leq 0
\end{equation}
The hypergraph of $Q'$ trivially admits the relaxed tree decomposition whose structure is that of the inequality tree of $Q'$ (and of $Q$): The skeleton edges are $\calE_u$ and the ligament edges are $\calE_\ell$.

The inference algorithm traverses the inequality tree bottom-up and eliminates one level of
query variables at a time. For a variable $X_p$ with children $X_{c_1},\ldots,X_{c_k}$, it computes recursively the factor
\begin{equation*}
    Q_p(x_p) \quad=\quad
    S_p(x_p)\cdot \prod_{i\in[k]} S_{c_i}(\text{lub}_i(x_p)) + (1 - S_p(x_p))\cdot Q_p(\text{lsub}_p(x_p))
\end{equation*}
We use $\text{lub}_i(x_p)$ to find the value in $S_{c_i}$ that is the least upper bound of $x_p$ and $\text{lsub}_p(x_p)$ to find the value in $Q_{p}$ that is the least strict upper bound of $x_p$, i.e., the next value in ascending order. The definition of $Q_p$ is recursive: It first computes the probability for $x_p$ and then for its previous values. In case $X_p$ has no children, i.e., $k=0$ the product over $S_{c_i}$ is one.

The probability of $Q$ is then the product of $S_\emptyset$ and the probability of the first tuple in the factor of the root variable. If $Q$ has a forest of inequality trees, then the subqueries for the trees would be disconnected and thus correspond to independent random variables. The probability of $Q$ is then the product of the probabilities of the independent subqueries.
\end{proof}

The case of inequality graphs can be reduced to that of inequality trees by variable elimination. The elimination of a variable $X_i$ repeatedly replaces it in the query by a value from its domain. The inequality graph of this residual query has no node for $X_i$ and none of its edges. By removing $k$ variables to obtain an inequality tree, the complexity of computing the query probability increases by at most the product of the sizes of the factors having these $k$ variables.

\section{Omitted Details about Tree Decompositions}
\label{appendix:td}

Here we prove Proposition~\ref{prop:non-redundant-td}, which is re-stated below.
\bprop[Re-statement of Proposition~\ref{prop:non-redundant-td}]
For every tree decomposition $(T, \chi)$ of a query $Q$, there exists a non-redundant
tree decomposition $(T', \chi')$ of $Q$ that satisfies
\[\{\chi'(t) \suchthat t \in V(T')\} \subseteq \{\chi(t) \suchthat t \in V(T)\}. \]
Moreover, if $(T, \chi)$ is $F$-connex, then $(T', \chi')$ can be chosen to be $F$-connex as well.
\label{prop:non-redundant-td:app}
\eprop
\bp
Given a redundant tree decomposition $(T, \chi)$, by Definition~\ref{defn:non-redundant-td}
there must exist $t_1 \neq t_2 \in V(T)$
where $\chi(t_1) \subseteq \chi(t_2)$. We claim that $t_1$ and $t_2$ can be chosen to be {\em adjacent}
in the tree $T$.
In particular, if $t_1$ and $t_2$ from Definition~\ref{defn:non-redundant-td} are already adjacent,
we are done. Otherwise, consider the node $t_1'$ that is adjacent to $t_1$ on the path from $t_1$
to $t_2$ in the tree $T$. By the running intersection property, we have $\chi(t_1)\subseteq\chi(t_1')$.
Therefore if we replace $t_2$ with $t_1'$, we obtain two new adjacent nodes $t_1$ and $t_2$
satisfying $\chi(t_1) \subseteq \chi(t_2)$.

Now we modify the tree decomposition $(T, \chi)$ by removing $t_1$ from $T$ and connecting all
the neighbors of $t_1$ (other than $t_2$) directly to $t_2$.
It is straightforward to verify that this modification results in a valid tree decomposition
$(T', \chi')$. Moreover this modification maintains the $F$-connex property of the original tree
decomposition, if it was $F$-connex in the first place.
If the new tree decomposition $(T', \chi')$ is non-redundant, we are done.
Otherwise, we inductively repeat the above process by finding a new adjacent pair $t_1\neq t_2$
satisfying $\chi(t_1) \subseteq \chi(t_2)$.
(This induction is over the number of bags in the tree decomposition since each time we are
dropping one bag.)
\ep

\section{The $\InsideOut$ Algorithm}
\label{app:insideout}
In this section, we aim to provide a proof sketch for Theorem~\ref{thm:insideout}.
We refer the reader to~\cite{DBLP:conf/pods/KhamisNR16} and its extended version~\cite{faq-arxiv}
for more details.
The proof also sheds light on many omitted technical details in the proofs of Theorems~\ref{thm:relaxed:faqw}
and~\ref{thm:subfaqw:boolean} including how to generalize theorems from the case of no
free variables $F = \emptyset$ to the case of an arbitrary set of free variables.

\begin{theorem}[Re-statement of Theorem~\ref{thm:insideout}]
$\InsideOut$ answers query~\eqref{eqn:faq} in time $O(N^{\faqw(Q)}\cdot\log N+|Q|)$.
\label{thm:insideout:app}
\end{theorem}
\bp
Let $Q$ be an $\faq$-query of the form~\eqref{eqn:faq}
with hypergraph $\calH=(\calV=[n],\calE)$ and free variables $F\subseteq \calV$.
Let $w:=\faqw(Q)$.
By definition of $\faqw(Q)$ from~\eqref{eqn:faqw2}, there must exist an $F$-connex tree decomposition
$(T, \chi) \in \td_F$ where all bags $t\in V(T)$ satisfy
\begin{equation}
    \rho^*_{\calE_{\not\infty}}(\chi(t))\leq w.
    \label{eq:IO:w-bound}
\end{equation}
Moreover by Proposition~\ref{prop:non-redundant-td}, the above tree decomposition $(T, \chi)$
can be assumed to be non-redundant.
By Definition~\ref{defn:F-connex}, there must exist a (possibly empty) subset
$V'\subseteq V(T)$ that forms a connected subtree of $T$ and satisfies
$\bigcup_{t\in V'}\chi(t)=F$.
Fix a root $r$ of the tree decomposition $(T, \chi)$ to be:
\bi
\item either an arbitrary node from $V'$ if $V'$ is not empty,
\item or an arbitrary node from $V(T)$ if $V'$ is empty.
\ei
Based on the above choice of the root $r$, the following holds:
\begin{claim}
If $V' \neq V(T)$, then there must exist a {\em leaf} node
$t_1\in V(T)\setminus V'$.
\label{clm:IO:non-free-leaf}
\end{claim}
If $V'$ is empty, then the above claim holds trivially.
Otherwise, the above claim holds because the root $r$ belongs to the connected subtree
$V'\neq V(T)$.

We recognize two cases:

{\bf Case 1:} $F \neq \calV$.
In this case, $V' \neq V(T)$
(since $F = \bigcup_{t \in V'} \chi(t)$ and
$\calV = \bigcup_{t \in V(T)} \chi(t)$).
By Claim~\ref{clm:IO:non-free-leaf}, let $t_1$ be a leaf node
from $V(T)\setminus V'$, and let $t_2$ be the parent of $t_1$. Let $L=\chi(t_1)$, $U=\chi(t_2)$,
and $M=L\setminus U$.
Because the tree decomposition $(T, \chi)$ is non-redundant
(thanks to Proposition~\ref{prop:non-redundant-td}), we have $M\neq \emptyset$.

\begin{claim}
For any $K\in\calE$ with $K \cap M \neq \emptyset$,
we must have $K \subseteq L$.
\label{clm:IO:rip:leafs}
\end{claim}
The above claim holds by the definition of a tree decomposition from  Section~\ref{subsec:width:params}:
Otherwise, the running intersection property would break.

To rewrite query~\eqref{eqn:faq}, we need to utilize the notion of indicator projection from
Definition~\ref{defn:indicator:proj} along with its property given by~\eqref{eq:indicator:proj:prop}.
Query~\eqref{eqn:faq} can be written as:
\begin{eqnarray}
    Q(\bm x_F) &=&
    \bigoplus_{\bm x_{\calV \setminus F}}
    \bigotimes_{K\in\calE} R_K(\bm x_K)\\
    &=&\bigoplus_{\bm x_{(\calV \setminus F)\setminus M}}
    \bigoplus_{\bm x_{M}}
    \left[
    \bigotimes_{\substack{K\in\calE\\K\cap M = \emptyset}} R_K(\bm x_K)
    \otimes
    \bigotimes_{\substack{K\in\calE\\K\cap M \neq \emptyset}} R_K(\bm x_K)
    \right]\\
    \text{(by \eqref{eq:indicator:proj:prop})}&=&\bigoplus_{\bm x_{(\calV \setminus F)\setminus M}}
    \bigoplus_{\bm x_{M}}
    \left[
    \bigotimes_{\substack{K\in\calE\\K\cap M = \emptyset}} R_K(\bm x_K)
    \otimes
    \bigotimes_{\substack{K\in\calE\\K\cap M \neq \emptyset}} R_K(\bm x_K)
    \otimes
    \bigotimes_{\substack{K\in\calE_{\not\infty}\\K\cap L \neq \emptyset}} \pi_{K,L}(\bm x_{K\cap L})
    \right]\\
    &=&\bigoplus_{\bm x_{(\calV \setminus F)\setminus M}}
    \left[
    \bigotimes_{\substack{K\in\calE\\K\cap M = \emptyset}} R_K(\bm x_K)
    \otimes
    \bigoplus_{\bm x_{M}}\underbrace{\left(
    \bigotimes_{\substack{K\in\calE\\K\cap M \neq \emptyset}} R_K(\bm x_K)
    \otimes
    \bigotimes_{\substack{K\in\calE_{\not\infty}\\K\cap L \neq \emptyset}} \pi_{K,L}(\bm x_{K\cap L})
    \right)}_{\text{defined as }\Phi_{t_1}(\bm x_{L})}
    \right]
    \label{eqn:IO:distributivity}
\end{eqnarray}
The last equality above holds because of the distributive property of semirings.
We define the product inside the inner sum $\bigoplus_{\bm x_{M}}$ to be a query
$\Phi_{t_1}(\bm x_{L})$, which is associated with the bag $t_1$.
Note that by Claim~\ref{clm:IO:rip:leafs}, all factors $R_K(\bm x_K)$ and $\pi_{K,L}(\bm x_{K\cap L})$
in this product involve only
variables from $\bm x_L$.

Query $\Phi_{t_1}$ can be computed with the help of worst-case
optimal join algorithms~\cite{DBLP:conf/pods/000118,
DBLP:conf/pods/NgoPRR12, DBLP:conf/icdt/Veldhuizen14}.
In particular, for every $K\in\calE_{\not\infty}$ where $K\cap L \neq \emptyset$,
define $\overline\pi_{K,L}$ to be the {\em support} of the factor $\pi_{K,L}$, i.e.
\[\overline\pi_{K,L}:=
\left\{\bm x_{K\cap L} \in \prod_{v\in K \cap L} \dom(X_v)\suchthat
\pi_{K,L}(\bm x_{K \cap L}) \neq \bm 0
\right\}.\]
$\overline\pi_{K,L}$ can be viewed as a relation over variables $\bm x_{K\cap L}$.
Solving the $\faq$-query $\Phi_{t_1}$ can be reduced to solving the
join query $\overline\Phi_{t_1}$ defined as follows:
\begin{eqnarray}
    \overline\Phi_{t_1}(\bm x_{L}) &:=&
\Join_{\substack{K\in\calE_{\not\infty}\\K\cap L \neq \emptyset}} \overline\pi_{K,L}(\bm x_{K\cap L})
\end{eqnarray}
This is because once we solve the join query $\overline\Phi_{t_1}$,
the $\faq$-query $\Phi_{t_1}$ can be computed as follows:
\begin{eqnarray*}
    \Phi_{t_1}(\bm x_{L}) =
    \begin{cases}
        \bigotimes_{\substack{K\in\calE\\K\cap M \neq \emptyset}} R_K(\bm x_K)
        &\text{if $\bm x_{L}\in \overline\Phi_{t_1}$}\\
        \bm 0&\text{ otherwise}
    \end{cases}
\end{eqnarray*}
where $\overline\Phi_{t_1}$ above denotes the {\em output} of the join query $\overline\Phi_{t_1}$.
The join query $\overline\Phi_{t_1}$ can be computed using a worst-case optimal
join algorithm in time
$O(N^{\rho^*_{\calE_{\not\infty}}(L)}\cdot \log N))$,
which is $O(N^w \cdot \log N)$ by~\eqref{eq:IO:w-bound}.

Once we have computed $\Phi_{t_1}$, we use it to compute $\Psi_{t_1}$ defined as follows:
\begin{eqnarray}
\Psi_{t_1}(\bm x_{L\cap U}) &:=&
\bigoplus_{\bm x_{M}} \Phi_{t_1}(\bm x_{L}).
\end{eqnarray}
The above can be computed by sorting tuples $\bm x_{L}$ that satisfy
$\Phi_{t_1}(\bm x_{L}) \neq \bm 0$
lexicographically based on $(\bm x_{L\cap U}, \bm x_{M})$
so that tuples $\bm x_L$ sharing the same $\bm x_{L\cap U}$-prefix become consecutive.
Then for each distinct $\bm x_{L\cap U}$-prefix, we aggregate away $\Phi_{t_1}(\bm x_L)$ over all
tuples $\bm x_L$ sharing that prefix.

Finally, expression~\eqref{eqn:IO:distributivity} can be rewritten as:
\begin{eqnarray}
    Q(\bm x_F) &=&\bigoplus_{\bm x_{(\calV \setminus F)\setminus M}}
    \left[
    \left(\bigotimes_{\substack{K\in\calE\\K\cap M = \emptyset}} R_K(\bm x_K)\right)
    \otimes\Psi_{t_1}(\bm x_{L\cap U})
    \right]
    \label{eq:IO:M-elimination}
\end{eqnarray}
The above is an $\faq$-query of the same form as~\eqref{eqn:faq}.
It admits an $F$-connex tree decomposition that results from the original $F$-connex
tree decomposition $(T, \chi)$ by removing the leaf bag $t_1$.
In particular, the newly added hyperedge $L\cap U$ (corresponding to $\Psi_{t_1}(\bm x_{L\cap U})$)
is contained in $\chi(t_2)$, and all other properties of $F$-connex tree decompositions continue to
hold after the removal of $t_1$.
Moreover thanks to the fact that $M\neq \emptyset$, the new query~\eqref{eq:IO:M-elimination}
has strictly less variables than the original query~\eqref{eqn:faq}.
In particular, the new query only involves the variables $\bm x_{\calV\setminus M}$
while the original query involves $\bm x_\calV$.
(We say that variables $\bm x_M$ have been {\em eliminated} from the original query hence
the term ``variable elimination''.)
By induction on the number of variables,
we can solve the original query~\eqref{eqn:faq} in the claimed time of $O(N^{w}\cdot\log N+|Q|)$.
(In the base case, we have an $\faq$-query with no variables, where the theorem
holds trivially.)

{\bf Case 2:} $F = \calV$.
Let $t_1$ be an arbitrary leaf node and $t_2$ its parent.
Let $L, U$ and $M$ be defined as before.
Claim~\ref{clm:IO:rip:leafs} continues to hold.
In this case, query~\eqref{eqn:faq} can be written as:
\begin{eqnarray}
Q(\bm x_{\calV}) &=&
\bigotimes_{K\in\calE} R_K(\bm x_K)\\
&=&
\bigotimes_{\substack{K\in\calE\\K\cap M = \emptyset}} R_K(\bm x_K)
\otimes
\bigotimes_{\substack{K\in\calE\\K\cap M \neq \emptyset}} R_K(\bm x_K)\\
\text{(by \eqref{eq:indicator:proj:prop})}&=&
\bigotimes_{\substack{K\in\calE\\K\cap M = \emptyset}} R_K(\bm x_K)
\otimes
\underbrace{
\bigotimes_{\substack{K\in\calE\\K\cap M \neq \emptyset}} R_K(\bm x_K)
\otimes
\bigotimes_{\substack{K\in\calE_{\not\infty}\\K\cap L \neq \emptyset}} \pi_{K,L}(\bm x_{K\cap L})
}_{\text{defined as }\Phi_{t_1}(\bm x_L)}
\label{eq:IO:free-vars}
\end{eqnarray}
Just like in the previous case, we use a worst-case optimal join algorithm to compute
$\Phi_{t_1}$ above in time
$O(N^{\rho^*_{\calE_{\not\infty}}(L)}\cdot \log N)) = O(N^w\cdot \log N)$.
Once we do, we compute its indicator projection:
\begin{align}
    \Psi_{t_1}(\bm x_{L\cap U})  &:= \pi_{L, U}(\bm x_{L\cap U}) =
    \begin{cases}
        \bm 1 & \exists \bm x_{M} \text { s.t. } \Phi_{t_1}((\bm x_{L\cap U}, \bm x_{M})) \neq \bm 0,\\
        \bm 0 & \text{otherwise.}
    \end{cases}
\end{align}
Now~\eqref{eq:IO:free-vars} can be written as:
\begin{eqnarray}
Q(\bm x_\calV) &=&
\bigotimes_{\substack{K\in\calE\\K\cap M = \emptyset}} R_K(\bm x_K)
\otimes
\Phi_{t_1}(\bm x_L)\\
\text{(by \eqref{eq:indicator:proj:prop})} &=&
\underbrace{
\bigotimes_{\substack{K\in\calE\\K\cap M = \emptyset}} R_K(\bm x_K)
\otimes
\Psi_{t_1}(\bm x_{L\cap U})}_{\text{defined as }Q'(\bm x_{\calV \setminus M})}
\otimes
\Phi_{t_1}(\bm x_L)
\label{eq:IO:output-free}
\end{eqnarray}
Note that thanks to the indicator projection $\Psi_{t_1}$ that is included
in the new query $Q'$ above, the following holds:
For every tuple
$\bm x_{\calV \setminus M}$ that satisfies
$Q'(\bm x_{\calV \setminus M}) \neq \bm 0$, there must exist at least one tuple $\bm x_M$
that satisfies
$Q((\bm x_{\calV \setminus M}, \bm x_M)) \neq \bm 0$. This in turn implies that:
\begin{equation}
    |Q'| \leq |Q|.
    \label{eq:IO:Q'-Q}
\end{equation}

By induction on the number of variables, we solve the new query $Q'$
(which doesn't have a bag $t_1$ nor variables $\bm x_M$) in time $O(N^{w}\cdot\log N+|Q'|)$,
which is $O(N^{w}\cdot\log N+|Q|)$ thanks to~\eqref{eq:IO:Q'-Q}.
Finally, we compute the original query $Q$ using the expression
\begin{eqnarray}
    Q(\bm x_\calV) = Q'(\bm x_{\calV\setminus M}) \otimes \Phi_{t_1}(\bm x_L).
\end{eqnarray}
In particular, the above expression can be computed in $O(|Q|)$ time as follows.
First, we index tuples $\bm x_L = (\bm x_{L \cap U}, \bm x_M)$ satisfying
$\Phi_{t_1}(\bm x_L) \neq \bm 0$ so that for a given
$\bm x_{L \cap U}$ we can enumerate in constant delay all tuples $\bm x_M$ where
$\Phi_{t_1}((\bm x_{L \cap U}, \bm x_M)) \neq \bm 0$.
After that, we iterate over tuples $\bm x_{\calV\setminus M}$ satisfying
$Q'(\bm x_{\calV\setminus M}) \neq \bm 0$,
extract the $\bm x_{L \cap U}$-part out of each such $\bm x_{\calV\setminus M}$-tuple,
and then use the previous index
of $\Phi_{t_1}$ to enumerate
$\bm x_M$-tuples corresponding to $\bm x_{L \cap U}$.
\ep

\section{The $\panda$ Algorithm}
\label{app:panda}
In this section, we give an overview of the $\panda$ algorithm developed in~\cite{panda-pods}
along with its extended version~\cite{panda-arxiv}.
The aim is to fill out omitted technical details in the proof of
Theorem~\ref{thm:subfaqw:general}, which introduces a variant of $\panda$ called $\spanda$.

Following notation from Section~\ref{sec:intro:faqai},
the input to the $\panda$ algorithm is as follows:
\bi
\item A multi-hypergraph~\footnote{See definition in Section~\ref{sec:intro:faqai}.} $\calH=(\calV =[n], \calE)$.
\item A relation $R_K$ associated with each hyperedge $K \in \calE$. The arity of $R_K$ is $|K|$.
\item A disjunctive Datalog query of the form
\begin{equation}
	\bigvee_{B \in \calB}G_B(\bm x_B) \leftarrow \bigwedge_{K\in\calE} R_K(\bm x_K)
	\label{eqn:disjunctive:datalog:query:app}
\end{equation}
where $\calB \subseteq 2^\calV$.
\ei

The output of $\panda$ is a collection of tables $(G_B)_{B\in \calB}$
that form a solution to the disjunctive Datalog query~\eqref{eqn:disjunctive:datalog:query:app}
(which can have many solutions).
In particular, the tables $(G_B)_{B\in \calB}$ must satisfy the following condition:
\begin{quote}
    Each tuple $\bm x_\calV$ that satisfies the conjunction
    $\bigwedge_{K\in\calE} R_K(\bm x_K)$ must also satisfy the disjunction
    $\bigvee_{B \in \calB}G_B(\bm x_B)$.
\end{quote}

Following notation from Sections~\ref{subsec:width:params} and~\ref{subsec:insideout:panda},
the runtime of $\panda$ is $\tilde O(N^{e})$, where
\begin{equation}
	e =
	\max_{h\in\ed_{\not\infty}\cap\Gamma_n}\min_{B\in\calB} h(B).
	\label{eq:panda-runtime:app}
\end{equation}
(Recall that $\tilde O$ hides a polylogarithmic factor in $N$.)
We start with some preliminaries. The following lemma shows how to convert the expression~\eqref{eq:panda-runtime:app}
into a linear program.
\begin{lemma}[\cite{panda-pods, panda-arxiv}]
    There exists a non-negative vector $\bm \lambda := (\lambda_B)_{B \in \calB}$ satisfying
    $\norm{\bm \lambda}_1 = 1$ and
    \begin{equation}
        \max_{h\in\ed_{\not\infty}\cap\Gamma_n}\min_{B\in\calB} h(B) =
        \max_{h\in\ed_{\not\infty}\cap\Gamma_n}\sum_{B \in\calB} \lambda_B \cdot h(B).
        \label{eq:panda:LP-form}
    \end{equation}
\label{lmm:panda:LP-form}
\end{lemma}
Note that the right-hand side of~\eqref{eq:panda:LP-form} is a linear program: Its variables are
$\left(h(S)\right)_{S\subseteq \calV}$, its objective function is $\sum_{B \in\calB} \lambda_B \cdot h(B)$,
and its constraints are $h \in \Gamma_n$ and $h \in \ed_{\not\infty}$, which are all linear.
(Recall the definitions of $\Gamma_n$ and $\ed_{\not\infty}$ from Section~\ref{subsec:width:params}
and~\eqref{eqn:ed:f} respectively.)
Our next step is to reduce solving this linear program
into finding a Shannon inequality, defined below.
\bdefn[Shannon inequality]
Given real constants $(\alpha_S)_{S \subseteq \calV}$
(where each $\alpha_S$ could be either positive, negative, or zero),
the linear inequality $\sum_{S \subseteq \calV} \alpha_S \cdot h(S) \geq 0$ is called a
{\em Shannon inequality} if it holds for all $h \in \Gamma_n$.
\edefn
Let $\opt$ be the optimal solution to the linear program from the right-hand side of~\eqref{eq:panda:LP-form}:
    \begin{equation}
        \opt := \max_{h\in\ed_{\not\infty}\cap\Gamma_n}\sum_{B \in\calB} \lambda_B \cdot h(B).
        \label{eq:panda:opt}
    \end{equation}
By linear programming duality, the following lemma was proved in~\cite{panda-arxiv}.
\begin{lemma}[\cite{panda-pods, panda-arxiv}]

    There exists a non-negative vector $\bm C = (C_K)_{K \in \ed_{\not\infty}}$ satisfying
    the following conditions:
    \bi
    \item The inequality
    \begin{equation}
        \sum_{B \in\calB} \lambda_B \cdot h(B) \leq \sum_{K \in \ed_{\not\infty}} C_K\cdot h(K)
        \label{eq:panda:LP-to-inequality}
    \end{equation}
    is a Shannon inequality.
    \item
    \begin{equation}
        \sum_{K \in \ed_{\not\infty}} C_K = \opt.
        \label{eq:panda:opt:dual}
    \end{equation}
    \ei
    \label{lmm:panda:LP-to-inequality}
\end{lemma}

Shannon-flow inequalities~\cite{panda-arxiv} is a special class of Shannon inequalities that
subsumes inequality~\eqref{eq:panda:LP-to-inequality}. It enjoys certain properties that
the $\panda$ algorithm relies on.
Given $X \subset Y \subseteq \calV$, let $h(Y|X)$ denote
\begin{equation}
    h(Y|X) := h(Y) - h(X).
\end{equation}

\bdefn[Shannon-flow inequality~\cite{panda-arxiv}]
Given $\calB \subseteq 2^\calV$, let
$\bm\lambda = (\lambda_B)_{B\in\calB}$ be a non-negative vector.
Let $\bm \delta = (\delta_{Y|X})_{X \subset Y \subseteq \calV}$ be another non-negative vector.
A {\em Shannon-flow inequality} is a Shannon inequality that has the following form:
\begin{equation}
    \sum_{B \in\calB} \lambda_B \cdot h(B) \leq
    \sum_{X \subset Y \subseteq \calV} \delta_{Y|X}\cdot h(Y|X).
    \label{eq:panda:shannon-flow}
\end{equation}
\edefn
Note that~\eqref{eq:panda:LP-to-inequality} is a special case of~\eqref{eq:panda:shannon-flow}
where $\delta_{K|\emptyset} = C_K$ for $K \in \ed_{\not\infty}$ and $\delta_{Y|X} = 0$
otherwise.

\blmm[Proof sequence construction~\cite{panda-pods,panda-arxiv}]
Every Shannon-flow inequality~\eqref{eq:panda:shannon-flow} admits a proof of the following form.
Start from the right-hand side of~\eqref{eq:panda:shannon-flow}, apply a sequence of proof steps
each of which replaces a term (or more) with a smaller term (or more),
until we end up with the left-hand side of~\eqref{eq:panda:shannon-flow}
(which proves that the left-hand side is smaller than the right-hand side).
Each proof step in the sequence
has one of the following forms:
\begin{eqnarray}
h(Y) \rightarrow h(X) + h(Y | X), & \text{for }X \subset Y \subseteq \calV, &\text{(decomposition step)}
\label{eq:step:decomposition}\\
h(X) + h(Y | X) \rightarrow h(Y), & \text{for }X \subset Y \subseteq \calV, &\text{(composition step)}
\label{eq:step:composition}\\
h(X|X\cap Y) \rightarrow h(X \cup Y | Y), & \text{for }X, Y \subseteq \calV, &\text{(submodularity step)}
\label{eq:step:submodularity}\\
h(Y) \rightarrow h(X), & \text{for }X \subset Y \subseteq \calV, &\text{(monotonicity step)}
\label{eq:step:monotonicity}
\end{eqnarray}
\label{lmm:panda:ps}
\elmm

Each proof step in~\eqref{eq:step:decomposition}-\eqref{eq:step:monotonicity} is interpreted
as replacing the term(s) on the left-hand side of the step with the terms(s) on the right-hand side.
Note that for each step in~\eqref{eq:step:decomposition}-\eqref{eq:step:monotonicity}
and each $h \in \Gamma_n$,
the right-hand side of the step is guaranteed to be smaller than the left-hand side.
For example, consider the submodularity step~\eqref{eq:step:submodularity}, where we replace
$h(X|X\cap Y)$ with $h(X \cup Y | Y)$. Because $h \in \Gamma_n$, it must satisfy the inequality
$h(X \cup Y) + h(X \cap Y) \leq h(X) + h(Y)$ (Recall the definition of $\Gamma_n$ in
Section~\ref{subsec:width:params}). But this inequality can be rearranged into
$h(X \cup Y | Y) \leq h(X|X\cap Y)$.
Similarly consider the monotonicity step~\eqref{eq:step:monotonicity}, where we replace $h(Y)$
with $h(X)$ for some $X \subset Y$. Since $h \in \Gamma_n$, it must satisfy $h(X) \leq h(Y)$
whenever $X \subset Y$.

The $\panda$ algorithm starts from the target runtime bound $N^e$
where $e$ is given by~\eqref{eq:panda-runtime:app},
computes a corresponding Shannon-flow inequality~\eqref{eq:panda:LP-to-inequality}
from Lemma~\ref{lmm:panda:LP-to-inequality}
(where $\norm{\bm \lambda}_1 = 1$ thanks to Lemma~\ref{lmm:panda:LP-form}), and then
uses Lemma~\ref{lmm:panda:ps} to construct
a proof sequence $s$ for this inequality
consisting of $\ell$ proof steps $s=(s_1, \ldots, s_\ell)$.
After that the algorithm mimics the process of using this proof sequence to prove
inequality~\eqref{eq:panda:LP-to-inequality}.
In particular, it starts from the right-hand side of~\eqref{eq:panda:LP-to-inequality}
associating each entropy term $h(K)$ with a corresponding input relation $R_K$.
After that it starts applying the proof steps one by one: Each time a proof step $s_i$
is applied to replace some entropy terms on the right-hand side of~\eqref{eq:panda:LP-to-inequality}
with new entropy terms,
the algorithm takes the relations associated with the old terms,
applies some relational operator on them to produce new relations,
and associates the new relations with the new entropy terms.
At the end of the proof sequence, we would have completely transformed the right-hand
side of~\eqref{eq:panda:LP-to-inequality} into the left-hand side completing the proof.
At that time, $\panda$ would have computed relations $G_B$ associated with entropy terms
$h(B)$ on the left-hand side of~\eqref{eq:panda:LP-to-inequality}.
Those particular $G_B$ relations form a solution to the input disjunctive Datalog
rule~\eqref{eqn:disjunctive:datalog:query:app}.
Moreover the algorithm ensures that every relational operator that was performed
while mimicking the proof sequence took time within our target runtime bound
of $N^e$.

Before formally describing the invariants maintained by the algorithm, we need some notation.
\bdefn[Degrees in a relation]
Given a relation $R_Y$ and a set $X \subset Y$,
the degree of $Y$ w.r.t. a tuple $\bm t_X \in \pi_X R_Y$
and w.r.t. to $X$
are defined as follows:
\begin{eqnarray}
	\deg_{R_Y}(Y|\bm t_X) &:=& \left|\left\{\bm t'_Y \in R_Y\suchthat \bm t'_X
	= \bm t_X\right\}\right|,\\
    \deg_{R_Y}(Y|X) &:=& \max_{\bm t_X \in \pi_X R_Y}\deg_{R_Y}(Y|\bm t_X).
\end{eqnarray}
As a special case, we have $\deg_{R_Y}(Y|\emptyset) = |R_Y|$.
\edefn

Although the $\panda$ algorithm starts from a Shannon-flow inequality of the special
from~\eqref{eq:panda:LP-to-inequality}, after applying a decomposition proof step
$h(K) \rightarrow h(X) + h(K | X)$ (for some $X \subset K$) replacing some term $h(K)$
on the right-hand side of~\eqref{eq:panda:LP-to-inequality} with new terms $h(X) + h(K | X)$,
the resulting inequality no longer falls under the special form~\eqref{eq:panda:LP-to-inequality}.
Instead it falls back to the more general form of a Shannon-flow inequality~\eqref{eq:panda:shannon-flow}.
Therefore, in general the $\panda$ algorithm maintains a Shannon-flow
inequality~\eqref{eq:panda:shannon-flow}.

The $\panda$ algorithm maintains the following invariants:
\bi
\item[(I1)] Every term $h(Y|X)$ on the right-hand side of~\eqref{eq:panda:shannon-flow}
is associated with a relation $R_Z$ satisfying $Y \setminus X\subseteq Z\subseteq Y$.
The relation $R_Z$ is called the {\em guard} of the term $h(Y|X)$.
(Note that if $X = \emptyset$, then $Z = Y$.)
\item[(I2)] The guards satisfy the following:
\begin{eqnarray}
    \sum_{\substack{X \subset Y \subseteq \calV\\R_Z : R_Z \text{ guards }h(Y|X)}}
    \delta_{Y|X}\cdot \underbrace{\log_2\left[\deg_{R_Z}(Z|X\cap Z)\right]}_{\text{defined as $n_{Y|X}$}}
    &\leq& e \cdot \log_2 N,
    \label{eq:panda:I2}
\end{eqnarray}
where $N$ is the input size and $e$ is given by~\eqref{eq:panda-runtime:app}.
For convenience, we define $n_{Y|X} := \log_2\left[\deg_{R_Z}(Z|X\cap Z)\right]$ where
$R_Z$ is the guard of $h(Y|X)$.
\item[(I3)] Every guard $R_Z$ satisfies
\begin{equation}
    |R_Z| \leq N^e.
    \label{eq:panda:I3}
\end{equation}
\ei
Initially, the above invariants are satisfied.
In particular, inequality~\eqref{eq:panda:shannon-flow} at the beginning is just
\eqref{eq:panda:LP-to-inequality} where each $h(K)$ is guarded by $R_K$ thus satisfying
invariant (I1). Moreover,~\eqref{eq:panda:I2} is satisfied as follows:
\begin{eqnarray}
    \sum_{K \in \ed_{\not\infty}} C_K\cdot n_{K|\emptyset} &\leq&
    \sum_{K \in \ed_{\not\infty}} C_K\cdot \log_2 N\\
    \text{(by~\eqref{eq:panda:opt:dual})}&=&\opt \cdot \log_2 N\\
    \text{(by~\eqref{eq:panda:LP-form})}&=&e \cdot \log_2 N.
\end{eqnarray}
Also~\eqref{eq:panda:I3} is satisfied because initially each input relation $R_K$ satisfies
$|R_K| \leq N$. (It is straightforward to verify that $e$ defined by~\eqref{eq:panda-runtime:app}
is at least $1$.)

Next we describe how $\panda$ handles each type of proof
steps~\eqref{eq:step:decomposition}-\eqref{eq:step:monotonicity} while maintaining the above
invariants and also ensuring that all operations are performed in time $N^e$.

{\bf Case 1:} Submodularity step
$h(X|X\cap Y) \rightarrow h(X \cup Y | Y)$ for some $X, Y \subseteq \calV$.
Let $R_Z$ be the guard of the term $h(X|X\cap Y)$.
We can directly use $R_Z$ as a guard of the new term $h(X \cup Y | Y)$ thus
satisfying invariant (I1).
Since both terms share the same guard, we have $n_{X \cup Y | Y} = n_{X|X\cap Y}$ hence
the left-hand side of~\eqref{eq:panda:I2} remains unchanged and invariant~\eqref{eq:panda:I2}
remains satisfied. Invariant~\eqref{eq:panda:I3} remains satisfied as well.

{\bf Case 2:} Monotonicity step
$h(Y) \rightarrow h(X)$ for $X \subset Y$.
Let $R_Y$ be the guard of $h(Y)$. We use $R_X := \pi_X R_Y$ as a guard of the new term $h(X)$.
We have $n_{X|\emptyset} = \log_2|R_X| \leq \log_2|R_Y| = n_{Y|\emptyset}$ hence
the left-hand side of~\eqref{eq:panda:I2} does not increase and invariant~\eqref{eq:panda:I2}
remains satisfied.
Invariant~\eqref{eq:panda:I3} remains satisfied because $|R_X| \leq |R_Y| \leq N^e$.
Moreover since $|R_Y| \leq N^e$ by invariant~\eqref{eq:panda:I3}, the projection
$\pi_X R_Y$ can be computed in our target runtime bound of $N^e$.

{\bf Case 3:} Decomposition step
$h(Y) \rightarrow h(X) + h(Y | X)$ for $X \subset Y$.
Let $R_Y$ be the guard of $h(Y)$.
In this case, we partition $R_Y$ into a small number $k = O(\log |R_Y|)$ of
relations $R_Y^{(1)}, \ldots, R_Y^{(k)}$, branch the execution of the algorithm into $k$ different branches
where $R_Y$ is replaced with $R_Y^{(j)}$ on the $j$-th branch for $j \in [k]$, and continue the algorithm
on each branch separately, and combine the outputs at the very end.
Because of the logarithmic number of branches created at each decomposition step, the runtime
of the algorithm blows up from the ideal bound of $N^e$ to $\tilde O(N^e)$ where $\tilde O$ hides
a polylogarithmic factor in $N$.

In particular, we partition tuples $\bm t_X \in \pi_X R_Y$ into $k$
buckets based on $\deg_{R_Y}(Y|\bm t_X)$ and partition $R_Y$ accordingly.
Specifically, for each $j \in [k]$, we define
\begin{eqnarray}
R_X^{(j)} &:=& \left\{\bm t_X \in \pi_X R_Y \suchthat 2^{j-1} \leq \deg_{R_Y}(Y|\bm t_X) < 2^j\right\},\\
R_Y^{(j)} &:=& \left\{\bm t_Y \in R_Y \suchthat \bm t_X \in R_X^{(j)} \right\}.
\end{eqnarray}
After partitioning, $\panda$ creates $k$ independent branches of the problem,
where in the $j$-th branch, $R_Y$ is replaced by both $R_X^{(j)}$ and $R_Y^{(j)}$.
Note that for each $j \in [k]$, the size of $R_X^{(j)}$ is at most $|R_Y|/2^{j-1}$ therefore:
\begin{equation}
    |R_Y| \quad\geq\quad |R_X^{(j)}| \quad\times\quad \frac{\deg_{R_Y^{(j)}}(Y|X)}{2}\\
\label{eq:panda:partition}
\end{equation}
Moreover if we partition each $R_Y^{(j)}$ further into two parts, we can get rid of the division by $2$
in~\eqref{eq:panda:partition} at the cost of doubling the number of branches.

Now on the $j$-th branch, we replace the term $h(Y)$ with the two terms $h(X)$ and $h(Y|X)$,
which are guarded by $R_X^{(j)}$ and $R_Y^{(j)}$ respectively.
By taking the $\log$ of both sides of~\eqref{eq:panda:partition} (and ignoring the division by 2),
we have $n_{Y|\emptyset} \geq n_{X|\emptyset} + n_{Y|X}$
hence the left-hand side of~\eqref{eq:panda:I2} does not increase and invariant~\eqref{eq:panda:I2}
remains satisfied. Moreover, because $\max(|R_X^{(j)}|, |R_Y^{(j)}|) \leq |R_Y| \leq N^e$,
invariant~\eqref{eq:panda:I3} remains satisfied as well.
Finally because $|R_Y| \leq N^e$ thanks to invariant~\eqref{eq:panda:I3},
the above partitioning of $R_Y$ can be performed in time $N^e$ as needed.

{\bf Case 4:} Composition step
$h(X) + h(Y|X) \rightarrow h(Y)$ for $X \subset Y$.
Let $R_X$ be the guard of $h(X)$ and $R_Z$ be the guard of $h(Y|X)$.
Recall from (I1) that this implies $Y\setminus X \subseteq Z \subseteq Y$.
In this case, we compute the join $R_Y := R_X \Join R_Z$ by going over tuples in $R_X$,
projecting each one of them onto $Z$, and finding the matching tuple in $R_Z$
(which can be done efficiently by a proper indexing of $R_Z$). The output size of this join satisfies
\begin{equation}
|R_Y| \quad\leq\quad |R_X| \quad\times\quad\deg_{R_Z}(Z | X \cap Z).
\label{eq:panda:join}
\end{equation}
Moreover the join can be computed in time proportional to $|R_Y|$.
We use the join result $R_Y$ as a guard for the new term $h(Y)$ that results from applying
the proof step.
From~\eqref{eq:panda:join}, we have $n_{Y|\emptyset} \leq n_{X|\emptyset} + n_{Y|X}$
hence the left-hand side of~\eqref{eq:panda:I2} does not increase and invariant~\eqref{eq:panda:I2}
remains satisfied.

It remains to verify that the above step maintains invariant~\eqref{eq:panda:I3}
and can also be performed in the desired time of $N^e$. As it turns out, neither is true:
Both the size of the new relation $R_Y$ and the time it takes to compute it can exceed $N^e$.
In order to enforce these invariants, some new technical ideas are needed that are beyond the scope
of this short introduction to $\panda$. We refer the reader to~\cite{panda-arxiv} for
a detailed explanation
of how to handle this last case properly without violating the invariants of the algorithm.

\end{document}